\documentclass[12pt]{article}
\usepackage[utf8]{inputenc}
\usepackage{amsfonts,amssymb,amsmath}
\usepackage{amsthm}
\usepackage{mathrsfs}
\usepackage[colorlinks=true,linkcolor=blue,citecolor=blue,breaklinks=true]{hyperref}
\usepackage{breakcites}
\usepackage{algorithm}
\usepackage{algorithmicx}
\usepackage[noend]{algpseudocode}
\usepackage{bm} \usepackage{dsfont}
\usepackage{graphicx, float, tikz}
\usepackage{colonequals}  \usepackage{color}
\usepackage{multicol}
\usepackage{tikz}
\usepackage{fancybox}
\usepackage[breakable, theorems, skins]{tcolorbox}
\usetikzlibrary{matrix,arrows}
\usetikzlibrary{positioning}
\usetikzlibrary{decorations}
\usetikzlibrary{decorations.pathreplacing}
\usetikzlibrary{shapes.geometric, chains, calc, fit, matrix}
\tikzstyle{system}=[rectangle,draw,fill=lightgray,minimum height=0.8cm,minimum
            width=0.8cm,thick]
\tikzstyle{BC}=[system]
\tikzstyle{resource}=[system]
\tikzstyle{RO}=[resource, minimum width=1cm]
\tikzstyle{protocol}=[circle, inner sep=0.7mm, draw]
\tikzstyle{simulator}=[circle, inner sep=0.7mm, draw]
\tikzstyle{memory}=[resource]
\tikzstyle{distinguisher}=[resource,fill=white,minimum width=3.5cm,
minimum height=1.2cm]
\tikzstyle{link}=[]

\usepackage{cleveref}

\usepackage{paralist}
\usepackage{enumerate}
\usepackage{enumitem}
\usepackage[n,
advantage,
operators,
sets,
adversary,
landau ,
probability,
notions,
logic,
ff,
mm,
primitives,
events,
complexity,
asymptotics,
keys]{cryptocode}

\renewcommand{\secpar}{\lambda}
\usepackage{tikz}
\usetikzlibrary{arrows,automata}

\usepackage{braket}

\usepackage[framemethod=tikz]{mdframed}
\usepackage{subcaption}

\usepackage{algpseudocode}

\usepackage{cleveref}

\usepackage{geometry}
\newcommand{\bbN}{\mathbb{N}}

\newcommand{\bbR}{\mathbb{R}}

\newcommand{\calA}{\mathcal{A}}
\newcommand{\calB}{\mathcal{B}}

\newcommand{\calD}{\mathcal{D}}
\newcommand{\calE}{\mathcal{E}}
\newcommand{\calF}{\mathcal{F}}

\newcommand{\calH}{\mathcal{H}}
\newcommand{\calI}{\mathcal{I}}

\newcommand{\calL}{\mathcal{L}}

\newcommand{\calP}{\mathcal{P}}

\newcommand{\calR}{\mathcal{R}}
\newcommand{\calS}{\mathcal{S}}
\newcommand{\calT}{\mathcal{T}}

\newcommand{\calV}{\mathcal{V}}
\newcommand{\calW}{\mathcal{W}}
\newcommand{\calX}{\mathcal{X}}
\newcommand{\calY}{\mathcal{Y}}

\newcommand{\Hyb}{\mathsf{Hyb}}

\DeclareMathOperator{\Tr}{\mathsf{Tr}}

\newcommand{\email}[1]{\href{mailto:#1}{#1}}

\newcommand{\NP}{\mathsf{NP}}
\newcommand{\secp}{{\lambda}}

\newcommand{\QPT}{\mathsf{QPT}}
\newcommand{\cE}{\mathcal{E}}

\newcommand{\FHE}{\mathsf{FHE}}
\newcommand{\KeyGen}{\mathsf{KeyGen}}

\newcommand{\Enc}{\mathsf{Enc}}
\newcommand{\Dec}{\mathsf{Dec}}

\newcommand{\CVQC}{\mathsf{CVQC}}
\newcommand{\NDCVQC}{\mathsf{NDCVQC}}

\newcommand{\Recover}{\mathsf{Recover}}
\newcommand{\Exp}{\mathsf{Exp}}
\newcommand{\Ver}{\mathsf{Ver}}

\newcommand{\ct}{\mathsf{ct}}

\newcommand{\QMA}{\mathsf{QMA}}
\newcommand{\lang}{\calL}
\newcommand{\yes}{\mathsf{yes}}
\newcommand{\no}{\mathsf{no}}
\newcommand{\Accept}{\mathsf{Accept}}
\newcommand{\Reject}{\mathsf{Reject}}

\newcommand{\Ext}{\mathsf{Ext}}

\newcommand{\Setup}{\mathsf{Setup}}
\newcommand{\pubparams}{\mathsf{pp}}
\newcommand{\td}{\mathsf{td}}
\newcommand{\mode}{\mathsf{mode}}
\newcommand{\lossy}{\mathsf{recovery}}
\newcommand{\injective}{\mathsf{injective}}

\newcommand{\MP}{\mathsf{MP}}
\newcommand{\Invert}{\mathsf{Invert}}
\def\vecA{\mathbf{A}}

\def\vec0{\mathbf{0}}
\def\vect{\mathbf{t}}

\def\vecs{\mathbf{s}}
\def\vece{\mathbf{e}}
\def\vecu{\mathbf{u}}
\def\vecx{\mathbf{x}}
\def\vecy{\mathbf{y}}
\def\vecz{\mathbf{z}}
\newcommand{\ValEst}{\sf ValEst}

\newcommand{\brho}{\boldsymbol{\rho}}
\newcommand{\bsigma}{\boldsymbol{\sigma}}

\newcommand{\regrand}{{\cal R}}
\newcommand{\regtranscript}{\calA}
\newcommand{\regprover}{\calP}

\newcommand{\SPNP}{\mathsf{WPA}_{\mathsf{NP}}}

\newcommand{\Repair}{\mathsf{Repair}}

\newcommand{\ketbra}[1]{\mathinner{|{#1}\rangle\langle{#1}|}}

\newcommand{\MW}{\mathsf{MW}}

\newcommand{\eigenspace}{\mathcal{ES}}
\newcommand{\cmszmodereg}{{\scriptstyle +_{\top, \bot}}}
\newcommand{\topbot}{{\scriptscriptstyle\top/\bot}}
\newcommand{\reset}{\mathsf{reset}}
\newcommand{\jor}{\mathsf{jor}}

\newcommand{\Repairable}{\mathsf{Repairable}}

\newcommand{\rand}{\mathsf{rand}}

\newcommand{\test}{\mathsf{T}}
\newcommand{\honest}{\mathsf{C}}

\renewcommand{\vec}[1]{\mathbf{#1}}

\newcommand{\poksoundness}{1-\frac{1}{|x|^d\secp^d}}

\newcommand{\LWE}{\mathsf{LWE}}

\newcommand{\PoNISec}{\mathsf{PoNI\text{-}Sec}}
\newcommand{\cert}{\mathsf{cert}}
\newcommand{\verkey}{\mathsf{vk}}
\newcommand{\VerKeyGen}{\mathsf{VKGen}}
\newcommand{\PoNI}{\mathsf{PoNI}}
\newcommand{\PoNIEncSearch}{\PoNI\text{-}\mathsf{Enc}\text{-}\mathsf{S}}

\newcommand{\regH}{\mathcal{H}}
\newcommand{\regR}{\mathcal{R}}
\newcommand{\Del}{\mathsf{Del}}
\newcommand{\PKE}{\mathsf{PKE}}
\newcommand{\Sig}{\mathsf{Sig}}
\newcommand{\PVCDExp}{\mathsf{PVCD\text{-}Exp}}
\newcommand{\Sign}{\mathsf{Sign}}
\newcommand{\regA}{\mathcal{A}}
\newcommand{\regB}{\mathcal{B}}
\newcommand{\regC}{\mathcal{C}}
\newcommand{\cAuxClass}{k_C}

\title{\Large{How to Classically Verify a Quantum Cat without Killing it}}
\author{Yael Tauman Kalai\thanks{MIT. \email{tauman@mit.edu}} \and Dakshita Khurana\thanks{UIUC and NTT Research. \email{dakshita@illinois.edu}} \and Justin Raizes\thanks{NTT Research.\email{justin.raizes@ntt-research.com}}}
\date{}
\newtheorem{definition}{Definition}
\newtheorem{remark}{Remark}
\newtheorem{lemma}{Lemma}
\newtheorem{theorem}{Theorem}
\newtheorem{claim}{Claim}

\newtheorem{corollary}{Corollary}

\theoremstyle{definition}
\newtheorem{construction}{Construction}

\begin{document}
\maketitle

\begin{abstract}
    A quantum verifier decides membership in QMA using a (single) quantum witness state. This witness may be arbitrary, and no requirement is placed on it beyond its acceptance probability.
A protocol claiming to classically verify QMA should offer the same guarantee: any witness that would convince a quantum verifier should also suffice to convince a classical verifier.
Unfortunately, existing classical verification of quantum computation (CVQC) protocols behave quite differently from the original QMA verifier:
they require many copies of the witness and destroy every one, effectively demanding that the prover solve the hard problem of witness duplication before verification can begin.

We build a protocol to classically verify QMA languages that uses a single copy of the prover's QMA witness.
Whenever the original QMA verifier accepts that witness with overwhelming probability, our protocol has negligible completeness and soundness errors and returns the witness negligibly disturbed after interaction with an honest verifier.
The soundness of our $\CVQC$ is based on the post-quantum Learning With Errors (${\sf LWE}$) assumption with superpolynomial modulus-to-noise ratio.

As an intermediate step of independent interest, we introduce {\em witness-preserving} classical arguments for $\NP$ which enable provers that start out with a superposition over $\NP$ witnesses to convince a verifier without disturbing their superposition.

\end{abstract}

\clearpage
\setcounter{tocdepth}{2}
\tableofcontents
\clearpage

\pagenumbering{arabic}

\section{Introduction}
Imagine a quantum prover that would like to demonstrate possession of Schr{\"o}dinger’s proverbial cat without killing the cat.
Or picture a classical verifier seeking to assess whether a prover has sufficient quantum money to pay them.
In both cases, the challenge is: how can one prove the existence of a valuable quantum state to a classical verifier without destroying it in the process?

Over the last few years, building on Mahadev's initial breakthrough~\cite{Mahadev18}, we have seen a number of exciting developments in $\mathsf{CVQC}$ (classical verification of quantum computation), where a classical verifier leverages a quantum prover to verify membership in $\mathsf{QMA}$~\cite{Mahadev18,TCC:ACGH20,TCC:ChiChuYam20,C:BKLMMV22,FOCS:Zhang22,KLVY23,FOCS:MetNatZha24,GKNV25,C:BKMSW25,C:BarKhu25}. However, as we discuss below, all of these protocols fall short of achieving the gold standard for $\mathsf{QMA}$ verification.\\

\noindent{\bf What should classical verification of $\QMA$ mean?} The standard definition of $\QMA$ places no structural requirements on the witness state beyond its acceptance probability. The witness can be any quantum state: it need not be a tensor product, need not have any special form, and the verifier needs only one copy of it\footnote{Any special structure of the witness would require a more powerful setup,
such as the ability to clone the witnesses, which may be computationally intractable~\cite{STOC:CKP25}.}.

A CVQC protocol for QMA should faithfully inherit this witness model: the prover should be able to use a single copy of any state that satisfies the QMA verifier’s completeness condition, without additional structural promises or access to fresh copies. Ideally, classical verification would also preserve the quantum verifier’s acceptance probability. We distinguish these two goals below.

Unfortunately, existing CVQC protocols fall short even of the first goal, including when the supplied witness is accepted by the original QMA verifier with probability $1$.
Their single-copy executions
only provide an inverse-polynomial gap between completeness and soundness even when the original QMA promise gap was constant.
One must then perform either parallel repetition (requiring many independent copies of the witness) or sequential repetition (consuming and destroying fresh copies one at a time) to boost the completeness-soundness gap. Unfortunately, repetition and destruction can be catastrophic in cryptographic applications. For example, proving possession of a quantum money state would need the owner to (1) create many copies of their money state, i.e. to \emph{break} the quantum money scheme, and (2) end up destroying all their money in the process!\\

\noindent{\bf Witness Preservation.}
This brings us to a second, closely related shortcoming. In the standard $\QMA$ setting with a quantum verifier, verification can be made essentially non-destructive.
The Marriott–Watrous amplification theorem~\cite{CC:MW05} shows that $\QMA$ completeness and soundness errors can be driven to $2^{-\lambda}$ using a single copy of the witness. At such high acceptance probability, the verification measurement barely disturbs the witness and the prover retains it for future use.

No existing $\CVQC$ protocol offers anything comparable. Every known protocol irreversibly destroys the prover's witness(es) during execution, even if the classical verifier is convinced with overwhelming probability.
The prover cannot reuse a single witness for the repeated executions on which current $\CVQC$s rely, let alone subsequent computations or audits by other verifiers.
Thus, to match the power of the quantum verifier, the {\em right} $\CVQC$ protocol should not only use a single witness but also return it to the prover intact.
\\

\noindent{\bf This Work.}
We construct classical verification of QMA using a single copy of the original witness and with negligible computational soundness error, assuming the post-quantum hardness of LWE with superpolynomial modulus-to-noise ratio.
The prover retains its witness after an honest interaction -- disturbed only by the unavoidable collapse associated with the verifier’s decision.

Completeness depends on the quality of the witness.
If the original QMA verifier accepts the witness with overwhelming probability, then our classical verifier also accepts with overwhelming probability. In this case, the witness is negligibly disturbed.
More generally, if the witness would convince the original QMA verifier with probability at least \(2/3\), then a single copy suffices to convince our classical verifier with probability at least \(1/3-\mathrm{negl}(\lambda)\).
Thus, our protocol achieves single-copy verification with constant gap for arbitrary valid witnesses, and additionally preserves high-quality witnesses.

\subsection{Results}
Our main result is a $\CVQC$ which uses one copy of \emph{any} witness which would convince the original $\QMA$ verifier, and furthermore preserves it if the witness is high quality.

\begin{theorem}[Informal]\label{thm:main:informal}
Assuming the post-quantum hardness of learning with errors (LWE) with superpolynomial modulus-to-noise ratio, there exists a $\CVQC$ protocol for any language $\calL\in\QMA_{a, b}$, associated with verifier $V_\QMA$ and with completeness $a$ and soundness error $b$, with the following properties:
\begin{itemize}
    \item \textbf{One-Copy Soundness and Completeness.}
    It has $\negl$ soundness error. Using \emph{one copy} of \emph{any} $\ket{w}$ which satisfies $\Pr[V_\QMA(x, \ket{w}) = 1] \geq a$, its completeness satisfies
    \[
        \Pr[\mathrm{CVQC\ accepts}] \geq \frac{a-b}{2-a-b}-\negl.
    \]
    In particular, this is $1/3-\negl$ for promise thresholds $a=2/3$ and $b=1/3$.
    Furthermore it has $1-\negl$ completeness if $\Pr[V_\QMA(x, \ket{w}) = 1] = 1-\negl$.

    \item \textbf{Witness Preserving for High Quality Witnesses.}
    If
    \[
        \Pr[V_\QMA(x,\ket{w})=1]=1-\negl,
    \]
    then, after interaction with an honest verifier, the prover is left with a state that is statistically close to $\ket{w}$.

    \item \textbf{Argument of Knowledge.}
    From any QPT prover that convinces our classical verifier with noticeable probability, an efficient extractor can extract a quantum witness up to complex conjugation: it outputs a state and a bit indicating whether to test it with $V_\QMA$ or its complex-conjugate verifier $\overline{V_\QMA}$. The indicated verifier accepts with noticeable probability, counting extraction failures as rejection.

     \end{itemize}
\end{theorem}

For witness states with real amplitudes, complex conjugation leaves the state unchanged, so allowing conjugation makes no difference. Thus when the original acceptance operator is real in the computational basis, the orientation bit can be discarded and the extraction guarantee holds for the original relation. Any $\QMA$ verification problem can be reformulated to use witness states with real amplitudes (but this may require changing the witness representation).

\begin{remark}
    If $V_\QMA$ would accept some witness $\ket{w}$ with overwhelming probability, then our CVQC will do the same. It does not matter whether $\ket{w}$ contains many copies of a sub-witness $\ket{w'}^{\otimes n}$ or if it is a single copy of some ``minimal'', yet high quality, witness.
    For example, one could use Marriott and Watrous's QMA verifier~\cite{CC:MW05}, which they show accepts certain short witnesses with probability $1-2^{-\secpar}$.

\end{remark}

Being simultaneously an argument of knowledge and witness preserving may at first seem contradictory;
if it is possible to extract the witness from the prover \emph{and} have the prover keep their witness, have we not cloned the witness, which could in general be hard for QMA~\cite{STOC:CKP25}?
Indeed, \cite{VidickZhang21} formalizes such an intuition for non-destructive \emph{proofs} of knowledge (i.e., with statistical soundness).

Our results demonstrate that such an issue does not apply if we only ask for an \emph{argument} of knowledge (i.e., with computational soundness).
Our approach dodges this issue by ensuring that extraction and preservation are both possible at the start of the protocol, but do not happen \emph{simultaneously}. In a real execution, the prover's witness is preserved. However, when extracting an adversarial prover's witness, the extractor interacts differently with the prover in a mode which disables witness preservation.
An adversarial prover cannot distinguish which mode is occurring, making our $\CVQC$ \emph{both} witness preserving and an argument of knowledge.
See \Cref{sec:overview-general-approach} for more details.

\paragraph{A Gentle Compiler.}
To obtain our main result, we build a compiler equipping any non-adaptive CVQC with a useful property.\footnote{A non-adaptive CVQC is one where the verifier's messages are independent of the prover's. This is a broad class containing all CVQC protocols we are aware of.}
Our key technical contribution is to compile such a CVQC so that {\bf the \emph{only} disturbance to the prover's witness is the measurement of the verifier's acceptance bit.}

\begin{theorem}[Informal]\label{informal-thm:witness-preserving}\label{informal-thm:gentle}
    Assuming the post-quantum hardness of $\LWE$ with superpolynomial modulus-to-noise ratio, there is an efficient compiler that converts any non-adaptive $\CVQC$ protocol where the prover runs in strict polynomial time into a new $\CVQC$ protocol with the same completeness and soundness that is
    \begin{itemize}
        \item \textbf{Gentle.} Let $\Pi_r$ be the projector onto states which the verifier would accept using randomness $r$. At the end of the protocol, the prover is left with a state that is statistically close to the result of measuring its input witness $\ket{w}$ with respect to the PVM $\{\Pi_r, I- \Pi_r\}$.
    \end{itemize}

    \noindent As a special case, if the $\CVQC$ has completeness $1-\negl$, then compiled protocol is
    \begin{itemize}
        \item \textbf{Witness Preserving.} At the end of the protocol, the prover is left with a state that is statistically close to its input witness $\ket{w}$.
    \end{itemize}
\end{theorem}

Our gentle compiler removes the damage from all other measurements that the prover may do during the original CVQC.
By limiting the damage to the witness, it becomes possible to repeatedly test the same witness to amplify completeness and soundness.

To obtain our main result, we combine the gentle compiler with Mahadev's measurement protocol~\cite{Mahadev18} and recent advances in frustration-free Hamiltonians~\cite{MN25}, then amplify sequentially. The new Hamiltonians and the measurement protocol together give a 1-copy CVQC with $1-\negl$ completeness, but which has a very high soundness error; false instances $x$ are accepted with up to $1-1/\poly[|x|]$ probability!
We gently compile this high completeness CVQC into a witness-preserving CVQC, then use the witness preservation to sequentially amplify to negligible soundness error without requiring that the prover produce additional witness copies.

To construct our compiler, we develop two techniques that we believe are of independent interest.

\paragraph{Witness-preserving arguments for $\NP$.} We construct an interactive argument for $\NP$ with the guarantee that if the prover starts with a superposition over possible witnesses $\sum \alpha_w\ket{w}$ then this superposition is maintained at the end of the interactive argument.  We refer to such an argument system as {\em witness preserving}.

 \begin{theorem}[Informal]
     There exists a witness preserving interactive argument for $\NP$ assuming the post-quantum ${\sf LWE}$ assumption with superpolynomial modulus-to-noise ratio.
 \end{theorem}
We refer the reader to \Cref{sec:state-pres-np} for details.
The construction uses the following primitive, which is also used by our compilers.

\paragraph{Dual-mode trapdoor function family with state recovery.}  We define and construct a  family of randomized functions that can be sampled in two modes: $\lossy$ mode or $\injective$ mode.  It has the guarantee that if a function is sampled in the $\lossy$ mode then it is {\em non-collapsing} in the sense that applying $f$ to a state $\sum_{w} \alpha_w\ket{w}$ and measuring the output~$y$, allows one to reconstruct the state $\sum_{w} \alpha_w\ket{w}$ given a trapdoor.  On the other hand, except with negligible probability over setup, a function sampled in the $\injective$ mode is injective and hence the state is collapsed to a single~$w$, which can be computed given the trapdoor.

Intuitively, these function families allow one to classically commit to a quantum state, with the option of recovering the committed state, with the receiver's help, instead of opening it.

 \begin{theorem}[Informal]
     There exists a dual-mode trapdoor function family with state recovery under the post-quantum LWE assumption with superpolynomial modulus-to-noise ratio.
 \end{theorem}

The construction
is basically the same as the construction of dual-mode claw-free trapdoor functions from \cite{BCMVV18}, with the addition of a new $\Recover$ algorithm.
We refer the reader to \Cref{sec:dual-trapdoor-recovery} for details.

\paragraph{Witness Preservation Against Malicious Verifiers.}
Finally, we mention the implications of our results to a related interesting scenario where the verifier is malicious and attempts to destroy the prover's witness. In other words,
\begin{quote}
    \begin{center}
        \emph{Can a quantum prover convince a verifier of a $\QMA$ statement without losing their witness even if the verifier deviates from the protocol?}
    \end{center}
\end{quote}

A possible route is to compile any witness-preserving $\CVQC$ (not necessarily non-adaptive) generically into one with \emph{malicious-verifier} witness preservation in a timing model, assuming time-lock puzzles and zero-knowledge arguments for $\NP$.
In this approach, the verifier would send the prover a time-lock puzzle containing its randomness; then for each message it sends in the $\CVQC$, it would prove that this message is consistent with the time-locked randomness in zero-knowledge.\footnote{This requires unbiased random choices and computational hiding of their unrevealed portions during verification. The timing bound must cover the adversary's computation through the verification decision, including the zero-knowledge arguments and any other auxiliary protocols. Coin tossing-in-the-well should meet these requirements in the composed protocol.}
If these requirements can be met, then when the verifier deviates from the protocol (e.g. by aborting), the prover could solve the time-lock puzzle and complete the protocol locally to recover its witness. To retain soundness, one would need to show that the prover learns no usable information about the unrevealed portions of the verifier's randomness beyond what the ordinary transcript reveals before the verification decision.

Another possible approach uses classical witness encryption for $\QMA$ (for which there are no plain-model constructions, but there are constructions in the classical oracle model e.g. due to \cite{ITCS:BarMal22}). The verifier encrypts a random message $m$ under the $\QMA$ statement to be proven, then the prover decrypts it and sends back $m'$, whereupon the verifier checks that $m' = m$.
If the witness encryption is correctly generated, then decryption is almost deterministic and disturbs the prover's witness negligibly. A zero-knowledge argument could establish consistency with some encryption coins, but ensuring that these coins are honestly sampled requires a separate argument. We leave a deeper exploration of this question to future work.

\subsection{Related Works}

\paragraph{Classical Proofs of Quantum Knowledge~\cite{VidickZhang21}.}
Vidick and Zhang observed that proofs of quantum knowledge with classical communication are generally destructive unless the witness can be (inefficiently) cloned, and so an honest prover in a sequentially repeated protocol needs to use multiple copies of the witness in order to succeed. Their analysis applies to perfectly non-destructive, statistically sound proofs and they provide some discussion about extending the results to protocols which only damage the witness by a small amount.\footnote{They also observed that in the specific case of QMA, an unbounded adversary could manufacture as many witness states as they wish. However, we note that the cloning implication still restricts QMA verification if there are at least two valid witnesses, since not all superpositions can be cloned.}
On the other hand, in this work we build arguments of knowledge that are non-destructive.

\paragraph{Proofs of No Intrusion~\cite{eprint:GR25}.}
Recently, Goyal and Raizes investigated a new primitive they called ``proofs of no-intrusion'' (PoNI). A PoNI for encryption is essentially a non-destructive classical argument that a ciphertext has not been ``stolen'', in the sense that no external party can decrypt the ciphertext even given the key. Our results enable these by combining our state-preserving argument for NP with a public-key encryption scheme with publicly verifiable certified deletion (known from just public-key encryption~\cite{TCC:BKMPW23,TCC:KitNisYam23}). See \Cref{app:PoNI} for more details.

However, their main new technical tool achieves something incomparable to our results. Goyal and Raizes introduce PoNIs for coset states, which is essentially a non-destructive classical argument of quantum knowledge specifically for coset states~\cite{C:CLLZ21}. Although conceptually related, PoNIs for coset states are incomparable to our results. Our results allow non-destructive arguments for the general class of $\QMA$, where the statement being proven is public. On the other hand, PoNIs for coset states allow the prover to blindly prove a specific statement about coset states, where the prover \emph{does not know the statement being proven.}

\paragraph{Disambiguation.} \cite{FOCS:LomMaSpo22} also used the term ``state preserving argument'' in the context of arguments of knowledge. There, ``state preserving'' refers to the ability to extract some knowledge from an adversarial prover without noticeably disturbing the prover's internal state. This takes place in a sandbox where the communications with the adversarial prover can be purified or rewound.
In contrast, our usage of witness-preserving refers to preserving an honest prover's input even in a real execution.

\subsection{Applications}

Witness-preserving $\CVQC$ is a powerful and flexible primitive that reduces the requirements both for future computation and for even running $\CVQC$ in the first place.

\begin{itemize}
    \item \textbf{Classically-Verifiable Unclonable Credentials.}
    \cite{TCC:GoyMalRai24,AC:JawKhu24} proposed the idea of an unclonable credential: an unclonable quantum state which can be publicly verified to prove identity or qualifications, for example to access a secure facility.

    Our results provide a route to verifying such credentials \emph{classically}. This substantially lowers the hardware requirements for deployment: the verifying door, server, or checkpoint can be entirely classical.
    Moreover, verification by an honest verifier does not consume an honestly generated credential that its original verifier accepts with overwhelming probability. A user can present the same credential repeatedly, and the device storing the credential need only hold one copy at a time.

    \item \textbf{Proofs of No-Intrusion.}
    As mentioned previously, our techniques also enable proofs of no-intrusion, which were recently proposed by \cite{eprint:GR25}.
    Compared with the original approach, our approach allows the verification key to be published so that anyone may act as the verifier.

    \item \textbf{Cryptographic Sub-Protocols.}
    A common paradigm in cryptographic protocols is to prove knowledge of an ``honest'' internal state (for example knowledge of a valid opening to a commitment).
    However, quantum protocols often create internal states that are unique to a particular execution, for example after measurements.

    Existing verification techniques would both require extra copies (which don't exist) or consume the state, preventing the protocol from continuing.
    Our techniques allow a party to give such a proof while preserving a high-quality state for use in continuing the protocol.
\end{itemize}

\subsection{Open Problems}
We conclude this section with a set of open directions:
\begin{enumerate}
    \item {\bf Round Complexity.}
    Can we achieve single-copy $\CVQC$ for $\QMA$ with negligible errors in constant {\em (or even two)} rounds of interaction?
    Since high-error protocols can only be composed sequentially when there is a single copy of the witness, addressing this question likely requires additional techniques beyond those developed in this work.

    \item {\bf Succinctness.}
    Can we achieve single-copy $\CVQC$ for $\QMA$ with negligible errors with smaller communication than the instance size?
    Our compiler from $1-\negl$ completeness to witness preservation can make each message succinct using techniques from \cite{GKNV25}.
    However, our compiler for amplifying soundness using one copy of the witness uses a number of sequential repetitions which grows with the instance size, preventing it from being succinct even with short messages.

    \item {\bf Relationship with Privacy-Preserving Proofs.}
    Witness preservation in $\CVQC$ seems to imply a form of witness hiding: intuitively, because the prover must keep the witness intact and quantum states cannot be cloned, an honest verifier cannot end the protocol holding the witness. It would be interesting to study connections and tradeoffs in the relationship between witness preservation and (variants of) zero-knowledge.

    \item {\bf Malicious-Verifier Witness Preservation.}
    Can the approaches sketched above be made rigorous, or can a quantum prover preserve its witness against a malicious verifier under weaker assumptions? This question is even interesting using \emph{quantum communication}, since the verifier is not guaranteed to cooperate with any request to uncompute the protocol.

    \item \textbf{Negligible Acceptance Overhead for General Witnesses.}
    Our protocol guarantees that if a quantum verifier accepts a state with overwhelming probability, then so does a classical verifier. Can we ensure that a classical verifier accepts witnesses with almost the same probability as the quantum verifier, even if that probability is much smaller than 1? Our compiler focuses on preserving the witness, which enables amplification techniques. If the witness could not have been amplified in the first place, this is not helpful. Achieving parity between classical and quantum verifiers seems to require a fundamentally new approach to CVQC.

    \end{enumerate}

\section{Technical Overview}

In what follows, we describe how to transform any $\QMA$ verifier into a $\CVQC$ protocol that uses a single copy of the $\QMA$ witness and is witness-preserving and has nearly perfect completeness and soundness.
We focus on using witnesses which the original QMA verifier would accept with overwhelming probability.

Our compiler works in three stages.

\begin{itemize}
    \item \textbf{Stage 1: A Gentle CVQC.}
    As our main technical contribution, we compile any non-adaptive $\CVQC$ into one where the prover's witness is only disturbed by the binary-outcome measurement of whether the verifier would accept.

    As a special case, we achieve \textbf{witness preservation}: the prover's witness is nearly untouched if the classical verifier would accept it with overwhelming probability.

    \item \textbf{Stage 2: Witness Preservation, Low Soundness.}
    Then, we gently compile a specifically high completeness $\CVQC$ protocol to obtain a witness preserving $\CVQC$. The combination of Mahadev's measurement protocol with Ma and Natarajan's frustration-free XZ Hamiltonians produces a $\CVQC$ with $1-\negl$ completeness, but which potentially accepts false QMA statements $x$ with up to $1-1/\poly[|x|]$ probability.

    \item \textbf{Stage 3: Sequential Amplification.}
    Finally, we amplify soundness by repeating the protocol sequentially. Since the starting witness is preserved after each repetition, the additional repetitions do not require additional copies of the witness.

\end{itemize}
The bulk of the technical ideas appear in stage 1, which is the focus of this overview.

\subsection{The Gentle Compiler}\label{sec:overview-high-completeness}\label{sec:overview-general-approach}
We begin with the gentle compiler, which limits the damage to the state to just measuring the verifier's decision bit.

\paragraph{Classical Verification in Superposition.}
Ultimately, the goal in $\CVQC$ is to convert a quantum statement into a classical one so that the classical verifier can check it.
The problem with existing approaches to $\CVQC$ is that measuring said transcript collapses the prover's state.
If the prover could keep the transcript ``in superposition'' throughout the protocol, then at the end of the protocol (after the verifier measures the verdict), the prover could uncompute the transcript and recover the witness.
This yields something like the following very high-level approach.
\begin{enumerate}
    \item Interactively, the prover \emph{coherently} computes its answers $a_i$ to each of the verifier's $\ell$ queries $q_{i}$ during the $\CVQC$ and does not measure them. At the end, the prover holds a superposition over answers $a_{[\ell]} = (a_1, \dots, a_{\ell})$:
    \[
        \sum_{a_{[\ell]}} \alpha_{a_{[\ell]}} \ket{a_{[\ell]}}
    \]Since the verifier's queries are non-adaptive, it doesn't need to see the prover's answers to send the next query.

    \item The verifier sends its (classical) secret state from the $\CVQC$
    to the prover\footnote{Sending the verifier's internal $\CVQC$ state is necessary so that the prover knows the $\NP$ statement it is proving. If the $\CVQC$ is publicly verifiable, this is not necessary, but publicly verifiable $\CVQC$ remains an open problem.} and the prover uses a special argument for $\NP$ to prove that it holds an answer sequence $a_{[\ell]}$ which the verifier would accept, without collapsing its superposition over $a_{[\ell]}$.

    Since the completeness is almost perfect, measuring whether $a_{[\ell]}$ is an accepting transcript disturbs the state negligibly.

    \item Finally, the prover uncomputes its answers $\sum_{a_{[\ell]}} \alpha_{a_{[\ell]}} \ket{a_{[\ell]}}$ to recover the witness.\footnote{Uncomputation requires that the prover is strict polynomial-time, since otherwise measuring the runtime of the prover may disturb the witness.}
\end{enumerate}
Of course, this approach is not sound as written because the prover can change its answers $a_{[\ell]}$ after seeing future queries or even the verifier's internal state. So, the question is:\\

{\em How do we bind the prover to $a_{[\ell]}$  without collapsing it or otherwise preventing the prover from uncomputing its answers and recovering the witness?}

\paragraph{Dual-Mode Trapdoor Functions with State Recovery.}
Our solution to this conundrum is an observation of a novel property of \cite{BCMVV18,Mahadev18}'s trapdoor claw-free functions which we abstract out as a new primitive: dual-mode trapdoor functions with state recovery.\footnote{See \Cref{appendix-NTCF-construction} for a discussion on using general dual-mode trapdoor claw-free functions.}
This is a family of randomized functions and trapdoors $(f_\pubparams, \td_\pubparams)$. The tuples can be generated in two indistinguishable modes: $\lossy$ or $\injective$.

In $\injective$ mode, two different inputs $x\neq x'$ cannot collide: $f_\pubparams(x;r) \neq f_\pubparams(x';r')$ for all randomness $r$ and $r'$. Given any $y = f_\pubparams(x;r)$, the trapdoor $\td_\pubparams$ can be used to find $x$.
This mode will be important for proving soundness of our protocol.

In $\lossy$ mode, $f_\pubparams$ is $2^n$-to-$1$, where $n$ is the input size, and $\td_\pubparams$ can be used to invert a preimage $(x; r_x)$ of $y = f_\pubparams(x;r_x)$ for \emph{any} $x\in \{0,1\}^n$, where $r_x$ is the randomness. Similarly to Mahadev's measurement protocol, a quantum party can ``commit'' to a computational basis measurement of a quantum state $\sum_{x} \alpha_x\ket{x}$ by preparing a superposition over all randomness $r$ then coherently evaluating $f_\pubparams(x;r)$ and measuring the result $y$. This results in the state
\[
    \sum_{x} \alpha_{x} \ket{x, r_x}
\]where $f_\pubparams(x, r_x) = y$ for all $x\in \{0,1\}^n$.
Here is where the recovery comes in: if the trapdoor $\td_\pubparams$ were revealed, the quantum party can use it to coherently compute $r_x$ controlled on $x$. By subtracting that from the randomness register, they disentangle their state and recover $\sum_{x}\alpha_x \ket{x}$.
This strengthens conventional TCFs which require merely noticeable availability of a sibling preimage.
To obtain negligible disturbance, we need a matching preimage to exist with overwhelming, rather than merely noticeable, probability. \Cref{sec:dual-trapdoor-recovery,appendix-NTCF-construction} establish this property and bound the recovery error.

Using the $\lossy$ mode, we can update the protocol above as follows.
\begin{enumerate}
    \item The verifier generates a pair $(f_\pubparams, \td_\pubparams)$ in the  $\lossy$ mode, and sends $\pubparams$ to the prover. Interactively, the prover ``commits'' to each answer $a_i$ by evaluating $f_\pubparams$ as discussed previously. At the end, the verifier knows $y_{[\ell]} = (y_1, \dots, y_\ell)$ and the prover holds
    \[
        \sum_{a_{[\ell]}} \alpha_{a_{[\ell]}} \ket{a_{[\ell]}} \otimes \ket{r_{a_{[\ell]}}}
    \]where $y_i = f_\pubparams(a_i;r_{a_i})$ for each $i\in [\ell]$.

    \item The verifier reveals its secret state for the underlying $\CVQC$ and the prover uses the ``witness-preserving'' argument for $\NP$ (which we elaborate on below) to prove that there exist preimages $(a_i;r_{a_i})$ of $y_{[\ell]}$ such that the verifier would accept $a_{[\ell]}$ in the $\CVQC$.

    \item The verifier sends $\td_\pubparams$ to the prover, who uses it to uncompute each $\ket{r_{a_{[\ell]}}}$. Since the superposition over answers $a_{[\ell]}$ is no longer entangled with any external state, the prover can uncompute its answers to recover the witness.
\end{enumerate}
If the witness-preserving argument for $\NP$ does not disturb the superposition $\sum_{a_{[\ell]}} \alpha_{a_{[\ell]}} \ket{a_{[\ell]}} \otimes \ket{r_{a_{[\ell]}}}$, then the $\lossy$ mode trapdoor allows recovering the prover's original witness up to negligible disturbance.

\paragraph{Witness-Preserving Arguments for $\NP$.}
Roughly, we need the prover to be able to use a superposition of $\NP$ witnesses $\sum_{i} \alpha_{i} \ket{w_{\NP,i}}$ to prove the $\NP$ statement, and at the end of the argument they should have a state close to their original witness superposition.

A very simple construction is possible using witness encryption for $\NP$. The verifier encrypts a random message $m$ under the statement to be proven. Then, the prover coherently decrypts $m'$ using their witness superposition. The verifier accepts if $m=m'$. Since decryption is almost deterministic using a valid witness, measuring $m'$ disturbs the prover's superposition negligibly.
Unfortunately, witness encryption is a relatively strong assumption, which we would like to avoid (and do avoid).

The next idea is to use a \emph{statistically} witness-indistinguishable (${\sf WI}$) argument. Intuitively, since the verifier's view is statistically independent of which witness was used, the prover avoids measuring their state.
However, some care is needed here. During the argument, the prover may also entangle their witness with the randomness used. To recover the witness superposition, it is crucial that the prover be able to unentangle their state later.\footnote{It is not hard to come up with examples where unentangling is computationally hard, for example using a claw-free lossy function \emph{without} a trapdoor.}

We show that the well-known 3-coloring ${\sf WI}$ argument, when instantiated using dual-mode trapdoor functions with state recovery, is a witness-preserving argument for $\NP$.
The NP argument must preserve the coherent witness and allow its auxiliary workspace to be undone; \Cref{sec:state-pres-np,sec:gentle-analysis} account for the trapdoor-function randomness and the 3-coloring randomness to establish these guarantees.

\paragraph{Soundness.}
The proof of soundness reduces to the soundness of the underlying $\CVQC$ by extracting an accepting transcript $a_{[\ell]}$ from the prover.
As a result, the new $\CVQC$ is an argument of knowledge for a specified relation $\calR$ if the underlying $\CVQC$ has the required testable knowledge guarantee for that same relation and witness register.
If $(f_\pubparams, \td_\pubparams)$ is generated in $\injective$ mode, then there is only one preimage $a_i$ of each $y_i$ in step 1. The reduction can use $\td_\pubparams$ to extract each answer $a_i$ from $y_i$ as the prover sends it, then forward $a_i$ to the $\CVQC$ verifier.

The subtle part of the proof is in showing that the extracted $a_{[\ell]}$ will be accepted by the verifier.
Intuitively, the indistinguishability of $\injective$ and $\lossy$ mode ensures that the adversarial prover cannot \emph{knowingly} change its behavior between the extractor and a real execution. However, the adversarial prover does not know $\td_\pubparams$, so anything extracted using it might change when we switch modes.\footnote{As a concrete example, suppose that the function was specified by an Fully Homomorphic Encryption (FHE) ciphertext $\ct = \Enc(b)$. If the mode is $\lossy$, $b=0$, and if the mode is $\injective$, $b=1$. An evaluation of $x$ is just a homomorphic multiplication of $x\cdot b$. Although an adversarial evaluator might not be able to detect a switch from $b=0$ to $b=1$, homomorphic evaluation allows them to encrypt different messages in the two cases.}

To prevent this possibility, we break the proof into two parts. First, we observe that in $\injective$ mode, the \emph{existence} of a preimage $a_{[\ell]}$ of $y_{[\ell]}$ which the $\CVQC$ verifier would accept, is sufficient to extract an accepting $a_{[\ell]}$, since said preimage is unique.

Then, we use the soundness of the $\NP$ argument to establish that this holds with almost the same probability that the adversarial prover would convince the non-destructive verifier.
If the $\NP$ argument verifier would accept (in $\lossy$ mode), there exists such a preimage, although it is not necessarily the unique preimage.
Since the $\NP$ argument can be checked \emph{without} $\td_\pubparams$, we can rely on the indistinguishability of $\lossy$ and $\injective$ modes to show that the probability of the $\NP$ argument verifier accepting is almost the same in $\injective$ mode as in $\lossy$ mode.

\paragraph{Dodging the Cloning Bullet.}
As mentioned previously, at first it may seem that the ability to both extract and retain the prover's witness implies the ability to clone it, which would violate the no-cloning theorem. Now that we have explained the core of the protocol, it becomes clearer why our approach does not imply cloning.

During our proof for (knowledge) soundness, we reduce to the knowledge soundness of the original $\CVQC$ by using an $\injective$ mode $f_\pubparams$ to extract the prover's answers for the original $\CVQC$. However, $\injective$ mode does not have the nice recovery property we used to preserve the witness. In fact, measuring an evaluation of an injective mode $f_\pubparams$ \emph{permanently} collapses the witness.
The only way to extract a witness is to take it away from the prover!

We also mention how our protocol avoids Vidick and Zhang's more formal cloning implication for \emph{proofs} of knowledge~\cite{VidickZhang21}.
Their technique relies on the non-destructivity to query the prover on every possible verifier message and learn its classical response. The resulting table of responses acts as an inefficient clone of the prover.
At that point, the extractor can be applied to each copy of the prover, extracting the prover's witness twice.
In our case, the extraction and witness recovery modes have two disjoint, but indistinguishable sets of verifier messages. The prover's behavior in the witness recovery mode could indeed be cloned, but it would not be possible to extract a witness from the cloned prover. In fact, the cloned prover would be able to distinguish between an $\injective$ and $\lossy$ mode $f_\pubparams$ simply because it does not know how to respond when the verifier sends $f_\pubparams$ in $\injective$ mode.

\subsection{Instantiating the Base $\CVQC$}
\label{sec:over-base}

To obtain our main result, we need a ``base'' $\CVQC$ protocol where the honest prover uses a single copy of the $\QMA$ witness. It needs $1-\negl$ completeness, but may have arbitrary soundness error $s$ so long as $1-s \geq 1/\poly[\secp,|x|]$.
We describe two general approaches to obtaining $\CVQC$, then specify a specific instantiation which does yield $1-\negl$ completeness.
We provide more detail on the completeness-soundness gaps for general approaches in
\Cref{app:c-s-gaps}.

\paragraph{Compiling Non-Local Games.}
One way to obtain such a $\CVQC$ protocol is to compile a non-local game using the \cite{KLVY23} compiler, as follows:
\begin{enumerate}
    \item Take any ${\sf MIP}^*$ protocol for $\calL$, with arbitrary completeness~$c$ and soundness~$s$, such that $c-s\geq 1/{\sf poly}(|x|,\secp)$, that uses a {\em single} copy of a (repairable) witness.
    \item Apply the compiler from \cite{KLVY23} to convert this ${\sf MIP}^*$ protocol into a $\CVQC$.
\end{enumerate}

Natarajan and Zhang~\cite{NZ23FOCS} give a non-local game for verifying BQP using one copy of the history state, and show that it can be compiled into a $\CVQC$.\footnote{The soundness of the \cite{KLVY23} compiler for general protocols is still under investigation. It was proven to be equivalent to the quantum value (assuming the existence of a secure quantum $\FHE$) for any game where the optimal strategy is a finite one \cite{NZ23FOCS,KMPSW25STOC,BaroniEtAl25Asymptotic,baroni2025}.  Unfortunately, it is not clear whether for every $\QMA$ language $\calL=(\calL_{\yes},\calL_{\no})$ and every ${\sf MIP}^*$ for $\calL$, it holds that for every $x\in\calL_{\no}$ the optimal strategy for proving that $x\in \calL_{\yes}$ is finite.}
Their result can be reinterpreted as $\CVQC$ for QMA by treating the history state as a QMA witness. Unfortunately, the completeness of their protocol is polynomially bounded away from $1$.

\paragraph{Mahadev's Measurement Protocol.}
Another method is to use Mahadev's protocol~\cite{Mahadev18} as the underlying $\CVQC$ protocol.  This requires some care since Mahadev's $\CVQC$ consists of many repetitions of an underlying $\CVQC$ (which~\cite{Mahadev18} constructs). As a result, the prover needs  many copies of the witness, one copy for each repetition.

It is tempting to use only one copy of Mahadev's protocol, however, a single copy of her original protocol has worse completeness than soundness!  To demystify this, we note that what makes this protocol useful is that it has a special structure: The verifier $V$ is defined via two verification algorithms $(V_\test, V_{\honest})$ which implements with probability $1/2$ each.
$V_\test$ is a test which has completeness ~$1$, while
$V_{\honest}$ checks the validity of the witness according to a ``ZX'' Hamiltonian $H$ whose terms are diagonal in the computational and Hadamard basis.
The guarantee is that for every cheating prover $P^*$ that convinces $V_\test$ to accept with probability close to~$1$, $V_{\honest}$ passes with probability close to the probability of passing the Hamiltonian $H$. Thus the soundness is bounded by $H$'s soundness $s$ when $V_\test$ acccepts with probability close to $1$.
We denote this notion of soundness by {\em testable soundness}.

In her original protocol, Mahadev uses the ZX Hamiltonian from Morimae--Fitzsimons~\cite{MF16}, which has completeness $c$ and soundness error $s$ which are both close to $1/2$.

\paragraph{High Completeness $\CVQC$.}
A high completeness $\CVQC$ can be obtained using Mahadev's approach by replacing the Morimae--Fitzsimons ZX Hamiltonian by a higher completeness Hamiltonian. Recently, Ma and Natarajan~\cite{MN25} showed just such a ZX Hamiltonian. However, Mahadev's approach also (testably) inherits the soundness error from the Hamiltonian, which is very high in Ma and Natarajan's case: false statements $x\in \lang_\no$ can be accepted with probability up to $1-1/\poly[|x|]$. Ma and Natarajan's Hamiltonian is a variation of Kitaev's clock construction, which allows generic conversion from any QMA verifier and its witnesses.
\\

We ensure that our gentle compiler in \Cref{informal-thm:gentle} also compiles $\CVQC$ protocols with testable soundness~$s$. We refer the reader to \Cref{sec:gentle,sec:wit-pres-cvqc} (and in particular to \Cref{thm:gentle,thm:wit-pres-cvqc} for the formal theorems).

\subsection{Amplifying Soundness}\label{sec:techover:amplifying}
Let us now focus on achieving negligible errors with a single copy of the witness.
To do this, we will have the prover and verifier sequentially repeat an witness-preserving $\CVQC$ with completeness $c$ and testable soundness $s < c$ a total of $N$ times, where $N$ is a large enough polynomial in $(\secp, |x|, (c-s)^{-1})$. In our case, $c \approx 1$ and does not need to be amplified.

In our sequentially repeated protocol, the verifier  will  accept if and only if all the test rounds accept, and at least a $\left(\frac{c+s}{2}\right)$ fraction of the check rounds accept. Since the witness is negligibly disturbed after each repetition, an honest prover who starts with a valid witness convinces the verifier in all executions except with negligible probability.

For our concrete construction, we establish in \Cref{sec:ma-natarajan} a base extractor that recovers a witness up to complex conjugation, together with a bit indicating whether to use $V_\QMA$ or $\overline{V_\QMA}$; denote this verification relation by $\calR_{\pm}$. To prove the repeated protocol's knowledge guarantee, we build an extractor whose output passes $\calR_{\pm}$ with noticeable probability whenever the prover generates accepting transcripts with noticeable probability.
Our extractor will simply pick a random execution $j \gets [N]$ and run the extractor of the underlying $\CVQC$ on this execution: we will prove that its output passes $\calR_{\pm}$ with noticeable probability.

To prove correctness of this extractor, we will rely on the extractability of the underlying $\CVQC$: namely, given any prover $P^*$ that convinces $V_\test$ to accept with probability close to~$1$, and convinces $V_{\honest}$ to accept with probability noticeably larger than $s$, the underlying extractor outputs a state and orientation bit that pass $\calR_{\pm}$ with noticeable probability.

Then our goal is to simply prove that with noticeable probability, the session $j$ picked by the outer extractor satisfies both the constraints above. We prove this in two parts: first, we
prove that for accepting transcripts, nearly all prefixes of the transcript satisfy the following: the prover, conditioned on this prefix, will pass $V_{\test}$ with high probability in the upcoming round. We call such a transcript a $\mathsf{Good}$ transcript.

Next, we prove that, outside a negligible exceptional event, every $\mathsf{Good}$ transcript has at least $\Delta N/12$ qualifying sessions, where $\Delta=c-s$: at each such prefix, the prover passes $V_\test$ with probability at least $1-\eta$ and $V_\honest$ with probability greater than $s+\Delta/4$. If there were fewer, Azuma's inequality and a bound on the number of check choices would put the number of accepting check rounds below the acceptance threshold.

Except with negligible probability, an accepting execution contains \(\Omega((c-s)N)\) rounds at which both extraction conditions hold. The extractor chooses a round uniformly and runs the protocol normally up to that round. Averaging over these unconditioned prefixes gives noticeable extraction probability whenever the prover is accepted with noticeable probability. The extractor does not condition on eventual acceptance or require another copy of the witness.

The proof of soundness proceeds similarly, by noting that for any prover that outputs accepting proofs with noticeable probability, there must exist a round where the prover passes the test round with probability close to $1$, and passes the check round with probability greater than $s$, which gives us a contradiction to the soundness of the base $\CVQC$.

\section{Preliminaries}

\subsection{Quantum Computing}

Let $\calH$ and $\calA$ be  Hilbert spaces.
A quantum algorithm $V$ is a unitary circuit $U_V$ together with a (without loss of generality, computational basis) measurement $\Pi_V = \{\Pi^x_V\}_{x\in X}$. To evaluate it on a state $\ket{\psi}\in \calH$, first append $\ket{0}\in \calA$ to it, then compute $U_V(\ket{\psi}\otimes \ket{0})$ and finally measure the resulting state with respect to $\Pi$. The ancilla registers are discarded afterwards.
The probability of outcome $x$ is
\[
    \|\Pi_V^x U_V (\ket{\psi} \otimes \ket{0})\|^2
\]
In the case where $V$ outputs a bit $x\in \{\Accept, \Reject\}$, we overload $V$ to also mean the operator representing an accepting outcome:
\[
    V \coloneqq (I \otimes \bra{0}) U_V^\dagger \Pi_V^\Accept U_V (I \otimes \ket{0})
\]
Then the probability of accepting a state $\ket{\psi}$ is
\[
    \Tr[V \ket{\psi}] = \bra{\psi} V \ket{\psi}.
\]
Any eigenvector of $V$ with eigenvalue $p$ is accepted by $V$ with probability $p$. A general quantum channel may not be described in terms of a unitary, in which case we refer to a quantum algorithm that implements the channel, i.e. a unitary circuit $U_{\Phi}$ and a measurement $\Pi_{\Phi}$, as a unitary dilation of the channel.

\begin{definition}\label{def:eigen}
    Given an algorithm $V$, we define the Hilbert space $\eigenspace_{\geq a}(V)$ to be the span of eigenvectors of $V$ with eigenvalue $\geq a$. $\eigenspace_{\leq b}(V)$ is defined similarly.
\end{definition}

The trace distance between two quantum states $\rho$ and $\sigma$ is $\frac{1}{2}\|\rho - \sigma\|_1$, where $\|\cdot\|_1$ denotes the trace norm. The distance between two quantum channels is measured by how far apart they map the same state. In other words,

\begin{definition}[Diamond Distance]
    The diamond distance, which is induced by the completely bounded trace norm,
    between two channels $\Phi_1$ and $\Phi_2$ on $n$ qubits is given by
    \[
        \frac{1}{2}\|\Phi_1 - \Phi_2\|_{\diamond}
        = \max_{\rho} \frac{1}{2} \|(\Phi_1 \otimes I_n)\rho - (\Phi_2 \otimes I_n)\rho \|_1
    \]
    where $\|\cdot\|_1$ denotes the trace norm.
\end{definition}

\begin{lemma}[Mixed to Pure]\label{lem:mixed-to-pure}
    Let $V$ be a binary-outcome measurement (not necessarily projective). Then any $\rho$ such that
    \[
        \Pr[\Accept \gets V(\rho)] \geq 1-\gamma
    \]
    can be decomposed as $\rho = q_\mathsf{good} \rho_\mathsf{good}  + q_\mathsf{bad} \rho_\mathsf{bad}$, a mixture over two orthogonal mixed states where
    \[
        q_\mathsf{bad} \leq \sqrt{\gamma}
    \]
    and $\rho_\mathsf{good}$ admits a convex decomposition into pure states $\ket{\psi}$ such that
    \[
        \Pr[\Accept \gets V(\ket{\psi})] \geq 1 - \sqrt{\gamma}
    \]
\end{lemma}
\begin{proof}
    Without loss of generality, we may write
    \[
        \rho = \sum_{i} q_i \ketbra{\psi_i}
    \]
    for some orthogonal basis $\{\ket{\psi_i}\}_i$.
    Let \[\calI_\mathsf{good}=\big\{i:~\Pr[\Accept \gets V(\ket{\psi_i}] \geq 1- \sqrt{\gamma}\big\}\] and define $\calI_\mathsf{bad}$ to be its complement.
    Define
    \[
q_{\mathsf{good}} = \sum_{i\in \calI_\mathsf{good}} q_i~~\mbox{ and }~~\rho_{\mathsf{good}} = q_{\mathsf{good}}^{-1}\sum_{i\in \calI_\mathsf{good}} q_i \ketbra{\psi_i}.
\]
Define $q_{\mathsf{bad}}$ and $\rho_{\mathsf{bad}}$ analogously. By definition, $\rho_{\mathsf{good}}$ satisfies the second half of the claim.

    Thus,
    \begin{align*}
        &1-\gamma\leq\\
        &\Pr[\Accept \gets V(\rho)]=\\
        &\sum_{i} q_i \Pr[\Accept \gets V(\ket{\psi_i})]
        =\\ &\sum_{i\in \calI_\mathsf{good}} q_i \Pr[\Accept \gets V(\ket{\psi_i})] + \sum_{i \in \calI_\mathsf{bad}} q_i \Pr[\Accept \gets V(\ket{\psi_i})]\leq
        \\
        & 1 - \sqrt{\gamma} \sum_{i\in \calI_\mathsf{bad}} q_i
    \end{align*}
   Therefore,
   \[
        q_\bad=\sum_{i\in \calI_\mathsf{bad}} q_i \leq \gamma/\sqrt{\gamma} = \sqrt{\gamma},
   \]
   as desired.
\end{proof}

\subsection{Tail Bounds}
Our analysis of sequential repetition will use tail bounds for certain non-independent processes called Martingales.
A martingale is a stochastic process in which the expected value of the next observation, given all prior observations, is equal to the most recent value.
We formally define such a process and tail bounds for this process below (focusing on the simplified finite case, which suffices for our setting).

\begin{definition}[Martingales]
A sequence of random variables $(S_m)_{m \geq 0}$ from a finite universe is called a \textbf{martingale} if
$\mathbb{E}[S_{m+1} \mid S_1, \ldots, S_m] = S_m$ for all $m \geq 0$.
\end{definition}

\begin{theorem}[Azuma–Hoeffding inequality for Martingales with bounded differences]
\label{thm:azuma}
Let $(S_m)_{m \geq 0}$ be a martingale
and let $c_1, c_2, \ldots, c_m$ be constants such that for all $1 \leq i \leq m$,
$|S_i - S_{i-1}| \leq c_i
$. Then for any $t > 0$,
        \[\Pr[ S_m - S_0 \geq t ] \leq \exp\left( \frac{-t^2}{2\sum_{k=1}^{m} c_k^2} \right).\]
\end{theorem}

\subsection{QMA}

\begin{definition}\label{def:QMA}
    Let $a, b: \bbN \rightarrow [0,1]$.
    A promise language $\lang=(\lang_\yes,\lang_\no)$ belongs to $\QMA_{a,b}$ if there exist a polynomial $p$ and a family $\{V_x\}_{x\in\lang_\yes\cup\lang_\no}$ of quantum polynomial-time algorithms such that, for every instance $x$ of length $n=|x|$, $V_x$ takes a quantum state on $p(n)$ qubits, outputs a decision bit, and satisfies the following properties.
    \begin{itemize}
        \item \textbf{Efficiency.} There exists a quantum polynomial-time algorithm that takes as input a classical string $x\in\{0,1\}^*$ and outputs a description of $V_x$.

        \item \textbf{Completeness.} For every $x\in \lang_\yes$, there exists a quantum state $\ket{w}$
        \[
            \Pr[\Accept \gets V_x(\ket{w})] \geq a(n)
        \]
        Such states $\ket{w}$ are called \emph{witnesses}.

        \item \textbf{Soundness.} For every $x\in \lang_\no$ and every state $\ket{\psi}$,
        \[
            \Pr[\Accept \gets V_x(\ket{\psi})] < b(n)
        \]
    \end{itemize}
    In general, there may be many verifiers $V=\{V_x\}_{x\in \lang_\yes \cup \lang_\no}$ satisfying these conditions. We say $\{V_x\}_{x\in \lang_\yes \cup \lang_\no}$ decides $\lang \in \QMA_{a,b}$ if it satisfies the above conditions. We say $\ket{w}$ is a witness for $x$ with respect to $V$ if it satisfies the completeness condition using $V_x$.
\end{definition}

In general, the set of witnesses is not a subspace because a witness may be the superposition over a state which is accepted with probability $1$ and a state which is accepted with probability $0$.
Marriott and Watrous \cite{CC:MW05} showed that eigenstates of $V_x$ with eigenvalues $\geq a$ can be amplified to exponentially small completeness error, and moreover that their procedure rejects eigenstates with eigenvalues $\leq b$ with exponentially small soundness error.

\begin{theorem}[\cite{CC:MW05}]\label{thm:MW}
    Let $a, b: \bbN \rightarrow [0,1]$ such that $a(n) - b(n) = 1/\poly[n]$ for all $n\in \bbN$.
    There exists a quantum polynomial-time algorithm  $\MW_{V_x}(1^\lambda, \cdot)$ such that for every promise language $(\lang_{\yes},\lang_\no)$ in $\QMA_{a,b}$, every $x\in \lang_\yes \cup \lang_\no$ with instance size $n\in \bbN$, and every $\lambda \in \bbN$, the following properties hold.
    \begin{itemize}
        \item \textbf{Completeness.} For all $\ket{w} \in \eigenspace_{\geq a}(V_x)$,
        \[
            \Pr[\Accept \gets \MW_{V_x}(1^\lambda, \ket{w})] \geq 1-2^{-\lambda}
        \]
        \item \textbf{Soundness.} For all $\ket{\psi} \in \eigenspace_{\leq b}(V_x)$,\footnote{Note that if $x \in \lang_\no$, then $\eigenspace_{\leq b}(V_x)$ is the entire Hilbert space on $p(n)$ qubits.}
        \[
            \Pr[\Accept \gets \MW_{V_x}(1^\lambda, \ket{\psi})] \leq 2^{-\lambda}
        \]
    \end{itemize}
\end{theorem}

\subsection{Classical Verification of Quantum Computation}\label{sec:cvqc}

In all $\CVQC$ definitions and guarantees below, fix a polynomial $q$ and let instance quantifiers range over $n=|x|\leq q(\secp)$. If lower-bound padding is used, $q$ is fixed to include the padded length. Negligible bounds are in $\secp$ and may depend on $q$ and the fixed $\QMA$ verifier.

\begin{definition}\label{def:CVQC}
    A Classical Verification of Quantum Computation $(\CVQC)$ protocol for a $\QMA_{a,b}$ language $\lang = (\calL_\yes,\calL_\no)$  associated with a verifier $V_\QMA = \{V_{\QMA, x}\}_{x\in \lang_\yes \cup \lang_\no}$ deciding $\lang\in \QMA_{a,b}$,  is an interactive protocol between a quantum prover $P$ and a classical verifier $V$.

    It satisfies the following for some parameters $c = c(\secpar,|x|)$ and $s = s(\secpar,|x|)$.
    \begin{itemize}
        \item {\bf Completeness $c$.}
        For every $\secp\in\mathbb{N}$, every $x\in\calL_\yes$ of size at most $q(\secp)$, and
        every witness $\ket{w}$ such that
        \[
            \Pr[V_{\QMA,x}(\ket{w})=1]\geq a,
        \]
       it holds that

        \[
            \Pr[(P(\ket{w}),V)(1^\secp,x)=1]\geq c(\secp,|x|).
        \]
    \begin{remark}
        Many of the $\CVQC$ protocols in the literature (such as \cite{Mahadev18}) need many copies of the witness to obtain completeness.  The focus of this work is on constructing $\CVQC$ protocols where the prover uses only a single copy of the witness, and hence we define completeness in this restrictive manner.
    \end{remark}

        \item {\bf Computational soundness $s$.}  For every $\QPT$ cheating prover~$P^*$, every $\secp\in\mathbb{N}$ and every $x\in\calL_\no$
        it holds that for every state $\ket{\psi}$,

        \[
        \Pr[(P^*(\ket{\psi}),V)(1^\secp,x)=1]\leq s(\secp,|x|).
        \]

        \item {\bf Efficiency.}  For every $\secp$ and every $x\in\calL_\yes$ with a corresponding witness $\ket{w}$,
        the honest prover $P(1^\secp,x,\ket{w})$ is an expected QPT algorithm with expected runtime ${\sf poly}(\secp, |x|)$
        and the verifier $V(1^\secp,x)$ is a classical probabilistic polynomial time circuit of size ${\sf poly}(\secp, |x|)$.
    \end{itemize}
\end{definition}

A CVQC protocol can additionally be a {\em classical proof of quantum knowledge}~\cite{VidickZhang21},  which is a generalization of the classical concept of a proof of knowledge to the quantum setting.

\begin{definition}[$\eta$-Argument of  Knowledge]
\label{def:pok}
    Let $\calR$ be a $\QPT$ algorithm outputting a decision bit, which we call the ``$\QMA$ relation''.
    An interactive protocol between a classical verifier and a quantum prover is an $\eta$-argument of (quantum) knowledge for $\calR$ if there exists a single $\QPT$ extractor $\cE$ such that, for every dishonest $\QPT$ prover $P^*$, there is a negligible function $\nu_{P^*}(\secp)$, uniform over all allowed instances $x$ and polynomial-size initial advice states $\rho$, for which
    \[
        p_{\rm ext}(\secp,x,\rho)
        \geq \eta(\secp,n)\,p_{\rm acc}(\secp,x,\rho)-\nu_{P^*}(\secp).
    \]
    Here $p_{\rm acc}$ is the verifier's acceptance probability and
    \[
        p_{\rm ext}(\secp,x,\rho)
        :=\Pr[\Accept\gets\calR(1^\secp,x,\sigma):
        \sigma\gets\cE^{P^*}(\rho,1^\secp,x)].
    \]
    Here $p_{\rm ext}$ is the probability that the QMA verifier $\calR$ accepts the extractor's output, counting extraction failures as rejection.
    The extractor has coherent access to the prover's round unitaries and their inverses, with one initialized private advice state.
\end{definition}

This guarantee does not assert that the extracted state meets the original $\QMA$ completeness threshold, either on each successful extraction or on average. This definition is weaker than the one in the classical setting, which guarantees that if $P^*$ convinces the verifier w.p.~$\frac{1}{p(\secp)}$ then the extractor extracts a valid witness with overwhelming probability (after running in time ${\sf poly}(\secp,p(\secp))$). With only one supplied state, a superposition $\alpha\ket{w}+\beta\ket{\sf junk}$ illustrates why an extraction attempt may fail, for example by projecting onto the junk component. This motivates allowing extraction failure; it does not show that a successfully extracted witness must have low acceptance probability.

\paragraph{Testable Soundness and Proof of Knowledge.} For use in our technical sections, we now define {\em testable} soundness and proof of knowledge. Intuitively, amplification can be done sequentially repeating a base CVQC protocol with high completeness and soundness errors; for certain protocols, completeness error is higher than the soundness error. Fortunately, the usefulness of these protocol comes from a special structure: the verifier $V$ is specified by a pair of verification algorithms $(V_\test, V_\honest)$. On any execution, the verifier chooses between them uniformly at random: it runs a test phase $V_\test$ with probability $1/2$ and a check phase $V_\honest$ with probability $1/2$. Crucially, $V_\test$ has \emph{completeness~$1$}, whereas $V_\honest$ has a \emph{lower completeness} $c$\footnote{This arises because the original $\QMA$ witness is transformed into a Morimae--Fitzsimons~\cite{MF16} $\QMA$ witness, which has low completeness.}.

The protocol's security guarantee is as follows: for every $x\in\mathcal{L}_\no$ and any cheating prover $P^*$ that succeeds at convincing $V_\test$ with probability close to $1$, the probability of convincing $V_\honest$ is bounded above by $s < c$. We define \emph{testable soundness} below, parameterized by $s$---the upper bound on the probability that $V_{\honest}$ accepts $x\in\mathcal{L}_\no$.

\begin{definition}[$(\eta,s)$-Testable Soundness]
\label{def:2part}
Let $\eta=\eta(\secp,|x|)\in(0,1]$ be an efficiently computable function satisfying $\eta\geq 1/{\sf poly}(\secp,|x|)$. A $\CVQC$ protocol for language $\lang$ has {\em testable soundness}~$(\eta,s)$, where $s=s(\secp,|x|)$, if the following holds. The verifier uniformly randomly picks one out of two verification algorithms: a ``test'' algorithm $V_\test$ and a ``check'' algorithm $V_\honest$, where $V_\test$ has perfect completeness,
and for every $\QPT$ cheating prover~$P^*$, every $\secp\in\mathbb{N}$ and every $x\in\calL_\no$
it holds that for every (mixed) state $\brho$ consisting of at most ${\sf poly} (\secp, |x|)$ qubits, if
        \begin{equation}\label{eqn:test-sound}
        \Pr[(P^*(\brho),V_{\test})(1^\secp,x)=1] \geq 1-\eta(\secp,|x|)
        \end{equation}
then
        \begin{equation}\label{eqn:check-sound}
        \Pr[(P^*(\brho),V_{\honest})(1^\secp,x)=1] \leq s(\secp,|x|)
        \end{equation}

When the test tolerance $\eta$ is implicit, we refer to this as $s$-testable soundness.
\end{definition}

The following definition generalizes the definition of testable soundness to an argument-of-knowledge setting, where the extracted state passes a test of the specified relation with a guaranteed probability when the prover succeeds in both phases.
The key property is as follows: whenever a prover succeeds in both (i) convincing $V_\test$ with high probability and (ii) convincing $V_\honest$ with at least some threshold probability $s$, there exists a \emph{quantum polynomial-time extractor} $\mathcal{E}$ that outputs a state $\sigma$ whose joint probability of extraction and acceptance by the underlying $\QMA$ relation $\mathcal{R}$ is at least $\delta$.

\begin{definition}[$(d,s,\delta)$-Testable Argument of Knowledge]
\label{def:2partpok}
Let $\calR$ be a $\QPT$ algorithm outputting a decision bit, which we call the ``$\QMA$ relation''.
A $\CVQC$ protocol for language $\lang$ is a $(d, s, \delta)$
{\em testable argument of knowledge} for $\calR$ if the following holds, where $d$ is a constant, $s=s(\secp,|x|)\in[0,1]$ is a probability-valued function, and $\delta$ is a function of $(\secp, |x|)$ that takes values between $0$ and $1$. The verifier uniformly randomly picks one out of two verification algorithms: a ``test'' algorithm $V_\test$ and a ``check'' algorithm $V_\honest$, where $V_\test$ has perfect completeness, and there is a $\QPT$ extractor $\cE$ such that for every $\QPT$ prover~$P^*$, every $\secp\in\mathbb{N}$ and every $x\in\calL_\yes \cup \calL_\no$
it holds that for every state $\brho$ consisting of at most ${\sf poly}(\secp,|x|)$ qubits, if
        \begin{equation}\label{eqn:test-pok}
        \Pr[(P^*(\brho),V_{\test})(1^\secp,x)=1] \geq \poksoundness \end{equation}
and \begin{equation}\label{eqn:check-pok}
        \Pr[(P^*(\brho),V_{\honest})(1^\secp,x)=1] \geq s(\secp,|x|)
        \end{equation}
then
    \[
    \Pr[\mathsf{Accept} \gets \calR(1^\secp, x, \sigma): \sigma \gets \cE^{P^*}(\brho, 1^\secp, x)] \geq \delta(\secp,|x|)
    \]
\end{definition}

We note that {\em a single repetition of} Mahadev's $\CVQC$ protocol~\cite{Mahadev18} has completeness $c$ w.r.t. $V_{\honest}$ and satisfies $s$-testable soundness with $(c-s)>1/{\sf poly}(\secp,|x|)$. Vidick and Zhang~\cite[Theorem 6.5]{VidickZhang21} obtain knowledge guarantees for relations satisfying their hypotheses (Q.1) and (Q.2), including the link from low Hamiltonian energy to acceptance by the specified relation. For our concrete construction, we use their Hamiltonian-state extractor~\cite[Claim 6.4 and Section 6.3]{VidickZhang21} and prove the required decoding bound in \Cref{sec:ma-natarajan}. This establishes testable knowledge for the original-or-conjugate relation $\calR_{\pm}$ of \Cref{thm:wit-pres-cvqc}.

\paragraph{Non-adaptive Verifiers.}
Finally, we will only compile protocols that admit non-adaptive verifiers, which we define below. First, we provide some notation.

If the $\CVQC$ protocol consists of $\ell$ back-and-forth rounds, we denote by $(q_1,\ldots,q_\ell)$ the verifier's messages, and by $(a_1,\ldots,a_\ell)$ the prover's messages, where in round $i$ the verifier sends query $q_i$ and the verifier responds with answer $a_i$.

\begin{definition}\label{def:non-adaptive-cvqc}
    An $\ell$-round $\CVQC$ protocol is said to have a non-adaptive verifier if the verifier's queries do not depend on the provers answers; namely, a non-adaptive verifier can be partitioned into two ${\sf PPT}$ algorithm $V=(V_1,V_2)$, where $V_1$ is the query sampler
    \[
    (\st, q_1,\ldots,q_\ell)\gets V_1(x,1^\secp)
    \]
    and $V_2$ is the verdict function
    \[    V_2(\st, a_1,\ldots,a_\ell)\in\{0,1\}.
    \]
\end{definition}
We note that all $\CVQC$ protocols that we are aware of have a non-adaptive verifier ~\cite{Mahadev18,TCC:ACGH20,TCC:ChiChuYam20,C:BKLMMV22,FOCS:Zhang22,KLVY23,FOCS:MetNatZha24,GKNV25,C:BKMSW25,C:BarKhu25}\footnote{\cite{C:BarKhu25}'s verifier messages technically depend on the prover's messages. However, that part of the protocol is independent of the statement being proven, so we can view it as a non-adaptive $\CVQC$ in the (classical) preprocessing model, which is also sufficient for our compilers.} (though they don't all satisfy our desired completeness guarantee that completeness holds with a single witness).

\section{Dual-Mode  Trapdoor Functions with State Recovery}\label{sec:dual-trapdoor-recovery}

In this section we define the notion of a {\em dual-mode trapdoor function family with state recovery}.  Intuitively, this is a family of (randomized) functions that have two modes: an {\em injective mode} and a {\em recovery mode}, and it is hard to distinguish between the two.  Except with negligible probability over setup, functions in the injective mode are injective and are associated with a trapdoor which can be used to invert the function efficiently. Functions in the recovery mode are also associated with a trapdoor, where this trapdoor is used to recover the state (this roughly corresponds to the lossy/two-to-one mode in the existing literature on dual-mode TCFs).  More specifically, functions in recovery mode have the property that if we compute the function coherently on a quantum state and measure the output value then we can use the trapdoor to recover a state at negligible trace distance from the input, except with negligible probability. The key difference over ``standard'' definitions of trapdoor claw-free functions is that in the lossy mode, we require inputs $x$ to have siblings with high probability.

\begin{definition}[Dual-Mode Trapdoor Function Family with State Recovery.]\label{def:dual-trapdoor-recovery}
    A \textbf{dual-mode trapdoor function family with state recovery} is a family of (potentially randomized) functions
    \[
        \calF=\{f_{\pubparams}:\calX\times\calR\rightarrow\calY\}_{\pubparams}
    \]satisfying the following properties.\footnote{Formally, $\calX,\calR,\calY$ all depend on the security parameter with respect to which $\pubparams$ was generated. This is omitted from the notation for simplicity.}
    \begin{itemize}
        \item \textbf{Efficient Generation.} There exists a probabilistic polynomial-time algorithm $\Setup$ which, on input $1^\secpar$ and $\mode\in\{\lossy,\injective\}$, outputs
        \[
        (\pubparams,\td)\gets\Setup(1^\secpar,\mode).
        \]
        For each mode there is an event $\mathsf{Good}_{\mode}$ depending only on setup, of probability at least $1-\nu_{\mode}(\secpar)$ for a negligible function $\nu_{\mode}$. The correctness guarantees below hold simultaneously for all allowed inputs on this event.

        \item \textbf{Dual-Mode Indistinguishability.}
        \[
            \{\pubparams:(\pubparams,\td)\gets\Setup(1^\secpar,\injective)\}
            \approx_c
            \{\pubparams:(\pubparams,\td)\gets\Setup(1^\secpar,\lossy)\},
        \]
        where $\approx_c$ denotes computational indistinguishability against $\QPT$ algorithms. These are the unconditioned setup distributions (not conditioned on good events).

        \item \textbf{Efficient Evaluation and Range Superposition.}
        There is a polynomial-time algorithm which, given any generated $\pubparams$ and any $(x,r)\in\calX\times\calR$, outputs $f_\pubparams(x;r)$. Evaluation randomness has mass function $p$ on $\calR$. The positive-amplitude state
        \[
        \ket{\rand_\pubparams}=\sum_{r\in\calR}\sqrt{p(r)}\ket r
        \]
        can be prepared to trace-distance error $\epsilon$ in time polynomial in $\secpar$ and $\log(1/\epsilon)$.

        \item \textbf{Injective Mode: Trapdoor.}
        There is a deterministic polynomial-time algorithm $\Ext$ such that for every key satisfying $\mathsf{Good}_{\injective}$ and simultaneously for every $(x,r)\in\calX\times\calR$,
        \[
        \Ext(\td,f_\pubparams(x;r))=x.
        \]
        In particular, outputs for distinct inputs have disjoint supports. The algorithm may return $\bot$ on other outputs.

        \item \textbf{Recovery Mode: Trapdoor.}
        Consider the following experiment $\mathsf{Exp}(\pubparams,\ket\psi)$, where
        $\ket\psi=\sum_x\alpha_x\ket{\phi_x}_A\otimes\ket x_B$ is an arbitrary pure state, including a reference register $A$.
        \begin{enumerate}
            \item Prepare the evaluation randomness in register $C$ and evaluate $f_\pubparams$ coherently to obtain, for ideal preparation,
            \[
            \sum_x\alpha_x\ket{\phi_x,x}_{A,B}
            \otimes\sum_r\sqrt{p(r)}\ket{r,f_\pubparams(x;r)}.
            \]
            \item Measure the output register to obtain $y$, and retain the residual state on $A,B,C$.
        \end{enumerate}
        There is a $\QPT$ algorithm $\Recover$ acting on $B,C$, with classical inputs $\td,y$, and negligible functions $\tau,\beta$, such that for every key satisfying $\mathsf{Good}_{\lossy}$ and every such input,
        \[
        \Pr_y\left[
          \frac12\|\rho^y_{A,B}-\ketbra\psi\|_1>\tau(\secpar)
        \right]\leq\beta(\secpar).
        \]
        Here $\rho^y_{A,B}$ is the recovered state conditioned on $y$; $A$ is untouched. The bounds are uniform over the key and the input, and include implementation error.

        We say that the family has $(1-\delta)$-state recovery if $\tau\leq\delta$, with negligible conditional failure $\beta$ and setup failure $\nu_{\lossy}$ kept separate. If $y$ and the ancillary registers are discarded, write
        $\Recover_\td\circ\Exp_\pubparams$ for the resulting channel on $B$. Its averaged error satisfies
        \[
        d:=\frac12\|\Recover_\td\circ\Exp_\pubparams-\mathrm{Id}\|_\diamond
        \leq\tau+\beta.
        \]
        This follows by convexity for each input with reference. The conditional guarantee and this averaged bound use different error parameters. Averaging also over setup adds at most $\nu_{\lossy}$.
    \end{itemize}
\end{definition}

Because these function families are evaluated in the computational basis, they support coherently implemented classical computation on the encoded data followed by recovery. More concretely, imagine that one had already evaluated $f_\pubparams$ on a state $\sum_{x}\alpha_x \ket{x}$, but wanted to measure some classical predicate of $x$. Computing on $x$ and measuring the result, and then performing $\Recover$ produces a state which is close to the state obtained if the computation had been performed directly on $\sum_{x}\alpha_x \ket{x}$.  This is formalized in the following lemma.

\begin{lemma}\label{lem:eval-measure-recover}
    Let $U=\sum_{x\in\calX}U_x\otimes\ketbra{x}$ be a unitary on registers $\calA,\calB$. For every recovery key satisfying $\mathsf{Good}_{\lossy}$, let $d$ be the averaged half-diamond error in \Cref{def:dual-trapdoor-recovery}. Then
    \[
    \frac12\left\|
    (I_\calA\otimes\Recover_\td)\circ(U\otimes I)\circ
    (I_\calA\otimes\Exp_\pubparams)-\mathcal U
    \right\|_\diamond\leq d,
    \]
    where $\mathcal U(\rho)=U\rho U^\dagger$, the ancillas of $\Exp$ are initialized internally, and $\Recover$ discards these ancillas and the measured output. Thus both channels have output registers $\calA,\calB$. With $(1-\delta)$-state recovery and conditional failure $\beta$, one may take $d\leq\delta+\beta$. Averaging over setup adds at most $\nu_{\lossy}$. The same bound holds after measuring any output register in both experiments.
\end{lemma}
\begin{proof}
    The controlled unitary $U$ commutes with coherent evaluation and with measurement of the function output: both computations are controlled by the computational basis of $\calB$ and otherwise act on disjoint registers. Move $U$ before $\Exp$, then apply the recovery bound with $\calA$ included in the reference. Composition with a channel, including a measurement, cannot increase diamond norm. Finally, on a bad setup key the half-diamond distance is at most one.
\end{proof}

\begin{theorem}\label{thm:dual-mode-state-recovery}

    For every polynomially bounded $N=N(\secpar)$, there exists a dual-mode trapdoor function family with $(1-2^{-\secpar})$-state recovery for $\calX=\{0,1\}^N$, assuming post-quantum decisional $\mathsf{LWE}_{n,m,q,D_{\mathbb Z,\sigma}}$ with the parameters in \Cref{appendix-NTCF-construction}. In the notation of \Cref{def:dual-trapdoor-recovery}, $\beta$ is negligible and $\nu_{\mode}\leq2^{-8\kappa}$ for both modes, with the internal parameter $\kappa$ specified there.
\end{theorem}

The construction is very similar to the construction of TCFs in \cite{BCMVV18,Mahadev18}, except for the new $\Recover$ algorithm.
At a high level, the result of evaluating
\cite{BCMVV18,Mahadev18}'s
trapdoor claw-free function in superposition and measuring the result $y$ is close to
\[
    \sum_{x} \alpha_x \ket{x, r_x}
\]where $f_\pubparams(x;r_x) = y$ for all $x$.
By using the trapdoor to invert $r_x$ from $y$ (controlled on $x$), the second register can be returned to $0$, recovering the original state.
Since we cannot destroy the witness, we need to ensure that the trapdoor can (with overwhelming probability) be used to recover preimages $r_x$ given $y$. In the following subsection, we show that this is achievable under the learning with errors (LWE) assumption.

\paragraph{On TCF-Based Approaches.}
The strategy outlined above matches dual-mode trapdoor claw-free functions, at a high level. However, it requires a strengthened property from these which is not standard. Specifically, we require that every input $(b, x_b)$ has a sibling $(1-b, x_{1-b})$ with overwhelming probability. Existing definitions (e.g. ~\cite{C:BarKhu25}) only require this to hold with noticeable probability, which can compromise recovery.

\paragraph{Other Constructions.}
In addition to approach outlined above, other approaches may be possible. The construction below uses at most one randomness preimage in each branch and nearly identical branch output distributions. The trapdoor computes every preimage that occurs, even though (rarely) an output not have a preimage in both branches.
But recovery may be possible even if $y$ has more complicated collisions. In fact, it is not strictly necessary to uncompute the preimages to recover the state. If the committed state can be ``opened'' to either a computational basis or Hadamard basis measurement, then recovery can use the observable-based extractor from \cite{VidickZhang21}.

\subsection{Construction}
\label{appendix-NTCF-construction}

We give a version with domain $\calX=\{0,1\}$. To evaluate on $N$ bits, we use the same public parameter and independent evaluation randomness for each bit. We account for this extension at the end of the proof.

The construction is lattice-based, and makes use of the following consequence of \cite{MicciancioP11}.

\begin{theorem}[Theorems 5.1 and 5.4 in \cite{MicciancioP11}]\label{thm:MP11}
Let $n\geq1$, $q\geq2$, and $m=\lceil c_0n\log_2q\rceil$ for a sufficiently large constant $c_0$. There is an efficient randomized algorithm $\mathsf{TrapGen}_\MP(1^n,1^m,q)$ that returns $\vecA\in\mathbb Z_q^{m\times n}$ and a trapdoor $\vect_\vecA$, whose matrix marginal is negligibly close to uniform in $n$. There is a deterministic polynomial-time algorithm $\Invert_\MP$ and an event $\mathsf{Good}_{\MP}$ depending only on $(\vecA,\vect_\vecA)$, with
\[
 \Pr[\neg\mathsf{Good}_{\MP}]\leq\nu_{\MP}(n):=2e^{-\pi n},
\]
such that on this event, simultaneously for every $\vecx\in\mathbb Z_q^n$ and every integer vector $\vecz$ with $\|\vecz\|_2\leq R:=q/(C\sqrt{n\log_2q})$,
\[
 \Invert_\MP(\vecA,\vect_\vecA,\vecA\vecx+\vecz)=(\vecx,\vecz).
\]
Here $C$ is a sufficiently large universal constant and arithmetic in the last argument is modulo $q$.
\end{theorem}

\begin{remark}
We use the statistical instantiation in Section 5.2 of \cite{MicciancioP11} with independent trapdoor entries equal to $0$ with probability $1/2$ and to each of $\pm1$ with probability $1/4$. Write $w=n\lceil\log_2q\rceil$ and $\bar m=m-w$. This distribution is $0$-subgaussian with parameter $\sqrt{2\pi}$. Lemma 2.9 there, with tail parameter $\sqrt n$, bounds the trapdoor's largest singular value by $C_1\sqrt{2\pi}(\sqrt{\bar m}+\sqrt w+\sqrt n)$ except with probability $2e^{-\pi n}$. This defines $\mathsf{Good}_{\MP}$. The same instantiation has matrix distance at most $(w/2)\sqrt{q^n/2^{\bar m}}$ from uniform. A sufficiently large $c_0$ makes this negligible. The fixed-trapdoor guarantee of Theorem 5.4, using the gadget decoder, is uniform over the stated error ball; enlarging $C$ makes its strict radius inequality hold on that closed ball. Setup is not conditioned on this event.
\end{remark}

\begin{construction}[Dual-Mode Trapdoor Functions with State Recovery]
\label{constr:fns-with-state-recovery}
The construction for input domain $\calX=\{0,1\}$ is as follows. Integer errors are represented by their centered representatives modulo $q$.
\begin{itemize}
\item {\bf $\Setup$:} On input $1^\secpar$ and $\mode\in\{\lossy,\injective\}$, use an internal parameter $\kappa\geq\secpar+\kappa_0$ and the following parameters (for sufficiently large $\kappa$):
\begin{itemize}
 \item A prime $2^{8\kappa}\leq q<2^{8\kappa+1}$, $n=\kappa^2$, and $m=\lceil c_0n\log_2q\rceil$ as in \Cref{thm:MP11}.
 \item Gaussian widths $\sigma=n$ and $\sigma'=2^{3\kappa}mB$, where
 \[
 t=\kappa+\lceil\log_2(16m)\rceil+10,\qquad
 B=\lceil t\sigma\rceil,\qquad B'=\lceil t\sigma'\rceil.
 \]
 \item Normalized truncated discrete Gaussians
 \begin{align*}
 \chi_B(a)&=\frac{e^{-\pi a^2/\sigma^2}}{\sum_{|j|\leq B}e^{-\pi j^2/\sigma^2}}
 \qquad(|a|\leq B),\\
 \chi_{B'}(a)&=\frac{e^{-\pi a^2/(\sigma')^2}}{\sum_{|j|\leq B'}e^{-\pi j^2/(\sigma')^2}}
 \qquad(|a|\leq B'),
 \end{align*}
 extended by zero outside these intervals. The widths are $\sigma,\sigma'$; the cutoffs are $B,B'$.
\end{itemize}
It does the following:
\begin{enumerate}
 \item Generate $(\vecA,\vect_\vecA)\gets\mathsf{TrapGen}_\MP(1^n,1^m,q)$.
 \item If $\mode=\lossy$, sample $\vecs\gets\mathbb Z_q^n$ uniformly and $\vece\gets\chi_B^m$, and set $\vecu=\vecA\vecs+\vece$.
 \item If $\mode=\injective$, sample $\vecu\gets\mathbb Z_q^m$ uniformly.
 \item Output $\pubparams=(\vecA,\vecu)$ and $\td=(\vecA,\vecu,\vect_\vecA)$.
\end{enumerate}

\item {\bf Function $f_\pubparams$:} Let $C=[-B',B']^m\cap\mathbb Z^m$. The randomness domain is $\calR=\mathbb Z_q^n\times C$, and the range is $\mathbb Z_q^m$. Define
\[
 f_{\vecA,\vecu}(b;(\vecx,\vecz))=\vecA\vecx+b\vecu+\vecz\pmod q.
\]
The evaluation randomness consists of uniform $\vecx\in\mathbb Z_q^n$ and independent $\vecz\gets\chi_{B'}^m$. Its ideal preparation has positive amplitudes.

\item {\bf Injective mode $\Ext$ algorithm:} On input $(\td,\vecy)$:
\begin{enumerate}
 \item Parse $\td=(\vecA,\vecu,\vect_\vecA)$.
 \item For each $b\in\{0,1\}$, run $\Invert_\MP(\vecA,\vect_\vecA,\vecy-b\vecu)$. Retain a returned pair $(\vecx_b,\vecz_b)$ only if
 \[
 \vecy-b\vecu=\vecA\vecx_b+\vecz_b\pmod q
 \quad\text{and}\quad\vecz_b\in C.
 \]
 \item Return the unique bit with a retained candidate, or $\bot$ if there is not exactly one.
\end{enumerate}

\item {\bf $\Recover$ algorithm:} Given $\td$, the measured output $\vecy$, and the retained input and randomness registers:
\begin{enumerate}
 \item Controlled on the input bit $b$, reversibly compute $\Invert_\MP(\vecA,\vect_\vecA,\vecy-b\vecu)$ in a fresh workspace.
 \item XOR the binary encoding of the returned pair $(\vecx_b,\vecz_b)$ into the retained randomness register. On every branch that occurs in an honest evaluation with a good key, this returns that register to zero.
 \item Reverse the decoding computation to return its workspace to zero, and output the input register. Discard the zero registers and $\vecy$.
\end{enumerate}
Use a fixed polynomial-time, deterministic implementation of the decoder, with a fixed failure encoding on other inputs. Its reversible simulation and the XOR operations are permutations of computational basis states, implemented without relative phases. No measurement of $b$ or of the decoding workspace is performed. For $N$ bits the same operations are performed bit by bit, in polynomial time.
\end{itemize}
\end{construction}

\paragraph{Parameters, setup, and indistinguishability.}
Fix a universal constant $\kappa_0\geq10$ large enough for the inequalities below, and take $\kappa\geq\kappa_0$. All parameters have polynomial bit length. Moreover,
\[
 \sqrt m(B+B')\leq\frac{q}{C\sqrt{n\log_2q}}=R,
 \qquad 2(B+B')<q,
\]
for all sufficiently large $\kappa$: the left-hand side of the first inequality is $2^{3\kappa}\mathsf{poly}(\kappa)$, whereas its right-hand side is $2^{8\kappa}/\mathsf{poly}(\kappa)$. Thus the decoding radius covers $C$ and $C\pm\vece$ for every allowed setup error. A recovery key is good when $\mathsf{Good}_{\MP}$ holds. An injective key is good when this event holds and
\[
 \vecu\notin\vecA\mathbb Z_q^n+(C-C),
\]
where $C-C=\{\vecz-\vecz':\vecz,\vecz'\in C\}$.
For each fixed $\vecA$, uniform $\vecu$ violates the latter condition with probability at most
\[
 \omega_\kappa:=q^{n-m}|C-C|
 \leq q^{n-m}(4B'+1)^m=2^{-\Omega(\kappa m)}.
\]
Indeed, $\log_2q=8\kappa+O(1)$ and $\log_2(4B'+1)=3\kappa+O(\log\kappa)$, whereas $m=\Theta(n\kappa)$. On a good injective key the two branches have disjoint images, and \Cref{thm:MP11} makes $\Ext$ correct simultaneously for all inputs and all evaluation randomness.

Write $D_{\mathbb Z,\sigma}(a)\propto e^{-\pi a^2/\sigma^2}$ for the untruncated discrete Gaussian. We assume post-quantum decisional $\mathsf{LWE}_{n,m,q,D_{\mathbb Z,\sigma}}$ for the parameters above, with uniform secret in $\mathbb Z_q^n$ and errors reduced modulo $q$. The distance between $D_{\mathbb Z,\sigma}^m$ and $\chi_B^m$ is at most
\[
 \epsilon_{\rm tail}:=4m e^{-\pi t^2}\leq2^{-16\kappa}.
\]
The same bound holds for $D_{\mathbb Z,\sigma'}^m$ and $\chi_{B'}^m$. For completeness, for $a\geq1$ and $t\geq2$, comparison with a Gaussian tail integral gives
$\Pr_{X\gets D_{\mathbb Z,a}}[|X|>\lceil ta\rceil]\leq4e^{-\pi t^2}$; the denominator is at least $a e^{-\pi}$, by summing the terms with $|j|\leq a$. Apply a union bound over coordinates. Thus replacing setup errors by $\chi_B^m$ costs at most $\epsilon_{\rm tail}$ in the LWE hybrid. Replacing the trapdoor-generated matrix by a uniform matrix costs its statistical distance in each mode. These hybrids establish dual-mode indistinguishability. Here $q/\sigma=2^{\Theta(\kappa)}$; the assumption is for this superpolynomial modulus-to-noise ratio.

The classical samplers may be implemented within statistical error $\epsilon_{\rm gen}\leq2^{-16\kappa}$. The positive-amplitude randomness state may be prepared within trace distance $\epsilon_{\rm prep}\leq2^{-16\kappa}$. To see efficiency, prepare a uniform superposition on each cutoff interval and use a controlled rotation with success amplitude $e^{-\pi a^2/(2\sigma^2)}$ (or with width $\sigma'$). Uncompute the arithmetic workspace before measuring the success qubit. Success produces the required positive amplitudes and has probability $\Omega(1/t)$. Polynomially many trials, with polynomial-bit arithmetic and sufficiently accurate rotations, achieve the stated accuracy; any exhausted trial budget is included in $\epsilon_{\rm prep}$. Classical rejection sampling can be used to obtain the analogous setup sampler. These preparations take place before interacting with the input state. Consequently we may take
\[
 \nu_{\lossy}=\nu_{\MP}(n)+\epsilon_{\rm gen},
 \qquad
 \nu_{\injective}=\nu_{\MP}(n)+\omega_\kappa+\epsilon_{\rm gen}.
\]
With $n=\kappa^2$ and a sufficiently large $\kappa_0$, the MP failure, branch-overlap probability, and matrix distance above are each at most $2^{-16\kappa}$. In particular,
\[
 \nu_{\lossy},\nu_{\injective}\leq2^{-8\kappa}.
\]
These bounds concern the unconditioned setup. On every good key all decoding and extraction statements hold simultaneously, including for inputs chosen after seeing the key. Samplers use bounded trial counts and return fixed, syntactically valid defaults on exhaustion; their total error is included in $\epsilon_{\rm gen}$ or $\epsilon_{\rm prep}$. Evaluation randomness always lies in its prescribed support. All algorithms remain bounded-time on bad keys and invalid inputs.

For clarity, the truncated Gaussian is the ideal evaluation distribution used below. Its coherent distance from the untruncated Gaussian state is $\sqrt{\epsilon}$ when the removed mass is $\epsilon$.
The Hellinger estimate below already accounts for this truncation; here $\epsilon_{\rm prep}$ is the additional error in preparing the truncated state.

\begin{claim}
The above construction has $(1-2^{-\secpar})$-state recovery. More precisely, for a good recovery key and ideal preparation, let $p$ be the normalized mass function of $\chi_{B'}^m$, extended by zero to $\mathbb Z^m$, and let
\[
 c_e=\sum_{\vecz\in\mathbb Z^m}\sqrt{p(\vecz)p(\vecz-\vece)},
 \qquad \eta_e=1-c_e.
\]
Then $\eta_e\leq2^{-4\kappa}$. After discarding $\vecy$, the recovered channel is
\[
 \Lambda_e(\rho)=c_e\rho+(1-c_e)\mathcal D_Z(\rho),
 \qquad \frac12\|\Lambda_e-\mathrm{Id}\|_\diamond=\frac{\eta_e}{2},
\]
where $\mathcal D_Z$ is computational-basis dephasing. For every pure input on $A,B$ as in \Cref{def:dual-trapdoor-recovery}, if $D_{\vecy}$ is its trace distance from the recovered state conditioned on $\vecy$, then
\[
 \mathbb E_{\vecy}[D_{\vecy}^2]\leq\frac{\eta_e}{2},
 \qquad
 \Pr[D_{\vecy}>\tau]\leq\frac{\eta_e}{2\tau^2}.
\]
\end{claim}
\begin{proof}
We first bound the Gaussian shift, accounting for normalization and truncation. Put $g=D_{\mathbb Z,\sigma'}^m$ and let $g_e(\vecz)=g(\vecz-\vece)$. Symmetry and Jensen's inequality give
\[
 \sum_{\vecz}\sqrt{g(\vecz)g_e(\vecz)}
 =\mathbb E_{\vecz\gets g}
 e^{\pi\langle\vecz,\vece\rangle/(\sigma')^2-\pi\|\vece\|_2^2/(2(\sigma')^2)}
 \geq e^{-\pi\|\vece\|_2^2/(2(\sigma')^2)}.
\]
We use $H^2(p_0,p_1)=1-\sum_z\sqrt{p_0(z)p_1(z)}$, as in \cite{BCMVV18}. If $\epsilon=g(\mathbb Z^m\setminus C)$, then $p=g|_C/(1-\epsilon)$ and $H^2(p,g)=1-\sqrt{1-\epsilon}\leq\epsilon$. The triangle inequality for $H$, also applied to the translated distributions, yields
\[
 \eta_e\leq\left(\sqrt{\frac{\pi\|\vece\|_2^2}{2(\sigma')^2}}+2\sqrt{\epsilon_{\rm tail}}\right)^2
 \leq\frac{\pi mB^2}{(\sigma')^2}+8\epsilon_{\rm tail}
 \leq2^{-4\kappa}.
\]

Let $P_b(\vecy)$ be the output distribution on input $b$. The good-key property gives at most one preimage within each branch. It also guarantees uniqueness on the enlarged error domain $C\cup(C+\vece)$: hence whenever both branches occur at $\vecy$, their preimages satisfy
\[
 (\vecx_1,\vecz_1)=(\vecx_0-\vecs,\vecz_0-\vece).
\]
Indeed, both $(\vecx_0,\vecz_0)$ and $(\vecx_1+\vecs,\vecz_1+\vece)$ decode from $\vecy$ within radius $R$. Thus, on their common support,
\[
 P_0(\vecy)=q^{-n}p(\vecz_0),\qquad
 P_1(\vecy)=q^{-n}p(\vecz_0-\vece),
\]
and, counting absent preimages as zero,
\[
 \sum_{\vecy}\sqrt{P_0(\vecy)P_1(\vecy)}=c_e.
\]
The radius and $2(B+B')<q$ ensure that the integer shifts agree with the centered modular representatives on these supports.

Since recovery erases the randomness without changing phases, the recovered instrument has Kraus operators
\[
 K_{\vecy}=\sqrt{P_0(\vecy)}\ketbra{0}+\sqrt{P_1(\vecy)}\ketbra{1}.
\]
Summing $K_{\vecy}\rho K_{\vecy}^\dagger$ leaves the diagonal blocks unchanged and multiplies the off-diagonal blocks by $c_e$, proving the channel identity. It remains valid with an identity on any reference register. Moreover, $\mathcal D_Z=(\mathrm{Id}+\mathcal Z)/2$, where $\mathcal Z(\rho)=Z\rho Z$, and $\|\mathcal Z-\mathrm{Id}\|_\diamond=2$ (attained on $\ket{+}$). This proves the stated half-diamond norm equality.

For the conditional statement, write a normalized pure input as
$\ket\psi=\sqrt a\ket{\phi_0,0}+\sqrt{1-a}\ket{\phi_1,1}$, with normalized reference states. Outcome $\vecy$ has probability $r_{\vecy}=aP_0(\vecy)+(1-a)P_1(\vecy)$. For $r_{\vecy}>0$, the pure-state trace-distance formula gives
\[
 r_{\vecy}D_{\vecy}^2
 =a(1-a)\big(\sqrt{P_0(\vecy)}-\sqrt{P_1(\vecy)}\big)^2.
\]
Summing gives $\mathbb E[D_{\vecy}^2]=2a(1-a)\eta_e\leq\eta_e/2$, and Markov's inequality proves the tail bound. Orthogonality comes from the input bit, so no assumption on the overlap of the reference states is needed. Purification and partial trace extend the bounds to mixed inputs.

If $\vecy$ is retained, the comparison channel must append an independent sample from $P_0$. The two corresponding purified outputs have overlap $a+(1-a)c_e\geq c_e$. Measuring and discarding the purifying registers therefore gives a half-diamond bound $\sqrt{1-c_e^2}\leq\sqrt{2\eta_e}$. The sharper $\eta_e/2$ bound applies when $\vecy$ is discarded.

Approximate preparation changes the entire recovered instrument, with its classical output retained and the identity on the external register $A$, by at most $\epsilon_{\rm prep}$ in half-diamond norm. For a pure target, trace distance squared is at most infidelity, and the average infidelity is the probability that the projector test $\ketbra\psi$ rejects the recovered state. Therefore the implemented experiment satisfies
\[
 \frac12\|\Lambda_e^{\rm impl}-\mathrm{Id}\|_\diamond
 \leq\frac{\eta_e}{2}+\epsilon_{\rm prep},\qquad
 \Pr[D_{\vecy}^{\rm impl}>\tau]
 \leq\frac{\eta_e/2+\epsilon_{\rm prep}}{\tau^2}
\]
on a good key. Averaging over setup adds at most $\nu_{\lossy}$ to each bound.

For $N$ bit evaluations, conditioned on the single good setup key, the ideal channel is $\Lambda_e^{\otimes N}$. Its half-diamond distance from the identity and its average infidelity on any pure input on $A,B$ are at most $N\eta_e/2$: write each factor as a mixture of the identity and $\mathcal Z$, with $\mathcal Z$ chosen with probability $\eta_e/2$, and bound the probability of any such choice. Preparation adds at most $N\epsilon_{\rm prep}$. Consequently the probability that the recovered joint state, conditioned on the protocol randomness and measurement outcomes, has trace distance greater than $\tau$ from the input is at most
\[
 \nu_{\lossy}+\frac{N(\eta_e/2+\epsilon_{\rm prep})}{\tau^2}.
\]
Choose $\kappa=\secpar+\lceil\log_2(N+1)\rceil+\kappa_0$ and $\tau=2^{-\secpar}$. This proves \Cref{thm:dual-mode-state-recovery}, with negligible conditional failure and the separate setup error $\nu_{\lossy}$. All dependent parameters, including $q,n,m,\sigma,\sigma',B,B'$, are chosen using this $\kappa$. If $J$ independently generated keys are used, their setup failure is at most $J\max_{\mode}\nu_{\mode}$; reusing one key incurs its setup error only once.
\end{proof}

\section{Witness-Preserving Arguments for $\NP$}\label{sec:state-pres-np}

In this section, we construct an interactive argument for $\NP$ with the property that if the (honest) prover starts with a state $\ket{\psi}=\sum_w \alpha_w\ket{w}$ which is a superposition over valid witnesses, then after running the protocol it is left with a state $\rho$ that is negligibly far from its original state $\ket{\psi}$. We call such an argument system a {\em witness-preserving argument}.

\begin{definition}\label{def:witness-preserving-NP}
    A \emph{witness-preserving argument} for an $\NP$ language $\lang$ is an interactive protocol between a $\QPT$ prover $P$ and a ${\sf PPT}$ verifier $V$ that satisfies the following guarantees.

    \begin{itemize}
        \item \textbf{Completeness and State Preservation.}
        There exist negligible functions $\mu_1$ and $\mu_2$ such that for every $\secp\in\mathbb{N}$, every $x\in \lang$, and every superposition $\ket{\psi} = \sum_{w} \alpha_{w} \ket{w}_{\calW}\otimes \ket{\phi_w}_{\calB}$ over witnesses for $x$ (encoded in the computational basis and potentially entangled with an external state $\ket{\phi_w}$),
        \[
            \Pr\left[\begin{gathered}
               \left(\frac{1}{2}\big\|\ketbra{\psi} - \rho \big\|_1 \leq \mu_1(\secpar)\right)
                ~~\land~~ \left(b = \Accept\right):\\
                (\rho, b) \gets \langle P(\ket{\psi}), V\rangle(1^\secpar, x)
            \end{gathered}\right] \geq 1- \mu_2(\secpar),
        \]
        where $P$ acts only on register $\calW$ and where the notation $(\rho, b) \gets \langle P(\ket{\psi}), V\rangle(1^\secpar, x)$ means that $V$ outputs $b\in\{\Accept, \Reject\}$ and $\rho$ is the residual joint state on $\calW\calB$, conditioned on the protocol randomness and measurement outcomes; auxiliary registers are discarded.

        When we wish to be more precise about the repair guarantees, we say that the argument is $(1-\mu_1)$-witness-preserving.

        \item \textbf{Computational Soundness.}
        For every polynomial $\ell=\ell(\secp)$ and every $\QPT$ cheating prover $P^*$, there exists a negligible function $\mu$ such that for every $\secpar \in \bbN$, every $x\notin \lang$ of size $\leq \ell(\secpar)$, and every auxiliary input $\ket{\psi}$ consisting of at most ${\sf poly}(\ell(\secpar))$ qubits,
        \[
            \Pr[\langle P^*(\ket{\psi}), V\rangle(1^\secpar, x) = \Accept] \leq \mu(\secpar).
        \]

    \end{itemize}

    We say a witness-preserving argument for NP is an $\delta$-\emph{argument of knowledge} if there exists a QPT extractor $\Ext$ such that for every QPT adversarial prover $P^*$ and every statement $x$,
    \[
        \Pr[\Accept \gets \langle P^*, V\rangle(1^\secpar, x)]
        \leq
        \Pr[(x, w) \in \calR_\lang: w\gets \Ext(P^*)] + \delta
    \]
    where $\calR$ is the NP relation for $\lang$. If $\delta = \negl$, we simply say that it is an argument of knowledge.
\end{definition}

State preservation composes in parallel since it also preserves entanglement with an external register. As a result, witness-preserving arguments can be parallely amplified whenever parallel repetition amplifies soundness: coherently classical-copy the witness to produce $\sum_{w} \alpha_w \ket{w}^{\otimes n}$ and independently run a witness-preserving argument on each witness copy.

\begin{lemma}\label{lem:parallel-repetition}
    Suppose each repetition satisfies the completeness and state-preservation guarantee of \Cref{def:witness-preserving-NP} with $\mu_1=\delta$ and $\mu_2=\beta$. For every input as in that definition, $k$ parallel or sequential repetitions produce a final state whose average over protocol randomness and measurement outcomes is within trace distance $k(\delta+\beta)$ of the input. Both states include the witness register $\calW$ and the untouched external register $\calB$. For any $\tau>0$, the final state conditioned on these random choices and outcomes has trace distance greater than $\tau$ from the input with probability at most $k(\delta+\beta)/\tau^2$.
\end{lemma}
\begin{proof}
    Replace repetitions by identity channels in execution order, including all other registers in the reference. Each replacement acts on a valid witness state and costs at most $\delta+\beta$ by convexity; subsequent channels contract this error. In parallel, coherently copy the witness before the repetitions and uncopy it afterward. Then the averaged bound follows from the triangle inequality. Squared trace distance from a pure target is at most infidelity, whose average is bounded by the averaged trace distance. The conditional bound can then be obtained via Markov's inequality.
\end{proof}

\subsection{Construction}
We construct witness preserving arguments for $\NP$ by building on top of the well-known witness indistinguishable argument for the $\NP$-complete 3-coloring language.
Instances for 3-coloring are graphs $x = (\calV,\calE)$. Witnesses are a 3-coloring of the graph, i.e., a function $w: \calV \rightarrow [3]$, such that for every edge $(i, j) \in \calE$, the vertices $i$ and $j$ have different colors: $w(i) \neq w(j)$. We consider $w$ to be written as a truth-table and denote the color it assigns to vertex $i$ by $w_i$.

Our construction makes use of a dual-mode trapdoor function family $\calF$, with a corresponding setup algorithm, denoted by $\Setup_\calF$ (see \Cref{def:dual-trapdoor-recovery}). We describe a version with an inverse-polynomial soundness gap. Soundness can be amplified by running the protocol sequentially, without losing state preservation.

\begin{construction}[Witness-Preserving Arguments for $\NP$]\label{constr:state-pres-np}
    Let $\lang$ be the 3-coloring language and consider any statement $x\in \lang$. The prover starts with a superposition $\sum_{w} \alpha_w \ket{w}$ over witnesses $w$ for $x\in \lang$.

    During the protocol, the prover will maintain three registers $\calW$, $\calP$, and $\regrand$.
    \begin{itemize}
        \item $\calW = (\calW_1, \dots, \calW_{|\calV|})$ will contain the witness and is initialized to $\sum_{w} \alpha_w \ket{w}$. \item $\calP$ will contain a permutation $\pi:[3]\rightarrow [3]$ and is initialized to $\propto \sum_{\pi\in \mathsf{Perm}([3])} \ket{\pi}$. \item $\regrand = (\regrand_1, \dots, \regrand_{|\calV|})$ will contain randomness $r_i$ for evaluating a trapdoor function and is initialized to $\bigotimes_{i\in \calV} \ket{\rand_\pubparams}$.
\end{itemize}

    \begin{enumerate}
        \item \textbf{Verifier:} Sample $(\pubparams, \td) \gets \Setup_\calF(1^\secpar, \lossy)$.

        Send $\pubparams$ to the prover.

        \item \textbf{Prover:}
        \begin{enumerate}
            \item For every $i\in \calV$ apply an isometry
        \[
            \ket{w_i, \pi, r_i}_{\calW_i, \calP,\regrand_i}
            \mapsto
            \ket{w_i, \pi, r_i}_{\calW_i, \calP,\regrand_i} \otimes \ket{f_{\pubparams}(\pi(w_i); r_i)}
        \]
        and measure the last register to obtain $y_i$.

        \item Send $y_{\calV} \coloneqq (y_1, \dots, y_{|\calV|})$ to the verifier.
\end{enumerate}
        \item \textbf{Verifier:} Sample a random edge $(i,j) \gets \calE$ and send $(i, j)$ to the prover.

        \item \textbf{Prover:} Measure $(\pi(w_i),\pi(w_j),r_i,r_j)$ using registers $\calP$, $\calW_i$, $\calW_j$, $\calR_i$, and $\calR_j$, and send it to the verifier.

        \item \textbf{Verifier:} Upon receiving $c = (c_i, c_j,r_i, r_j)$ from the prover, check that $c_i,c_j\in[3]$ and $r_i,r_j\in\calR$, that $c_i \neq c_j$, that $y_i = f_\pubparams(c_i;r_i)$, and that $y_j = f_{\pubparams}(c_j; r_j)$. If these are all true, output $\Accept$ and send $\td$ to the prover.  Otherwise, output $\Reject$.

        \item \textbf{Prover:} For each $k \in \calV\backslash\{i,j\}$, coherently compute $\pi(w_k)$ in a fresh color register, run $\Recover_\calF(\td,\cdot,y_k)$ on that register and $\regrand_k$, and uncompute and discard the color register.

        Next, observe that for every witness $w$, there is a unique permutation $\pi_{c,w,(i,j)}$ such that $c_i = \pi_{c,w}(w_i)$ and $c_j = \pi_{c,w}(w_j)$; this holds because $w_i \neq w_j$ and fixing two outputs of a permutation on $[3]$ fully determines the permutation. Compute the unitary swapping
        \[
            \ket{w, \pi_{c,w,(i,j)}}_{\calW, \calP}
            \leftrightarrow
            \ket{w, \mathsf{id}}_{\calW, \calP}
        \]
        where $\mathsf{id}$ is the identity permutation, and acting as the identity elsewhere.

        Output register $\calW$.
    \end{enumerate}
\end{construction}

\begin{theorem}\label{thm:witness-preserving-NP}
    \Cref{constr:state-pres-np} is a witness-preserving argument for $\NP$ with soundness $1-\frac{1}{|\calE|}+\negl$, assuming a dual-mode trapdoor function family with state recovery. The construction of \Cref{thm:dual-mode-state-recovery}, with parameters chosen below, makes this argument $(1-2^{-\secpar})$-witness-preserving in the sense of \Cref{def:witness-preserving-NP}, under the LWE assumption stated there.

    Furthermore, sufficiently many parallel repetitions give an argument of knowledge.
\end{theorem}

\paragraph{Remark}
To improve the soundness error to be negligible, one can repeat the protocol in \Cref{constr:state-pres-np} sequentially $\poly$ times. State preservation follows from \Cref{lem:parallel-repetition}, with the internal security parameter adjusted for the number of repetitions.

\begin{proof}
    We need to prove that the completeness, soundness and state preservation properties hold.
    \paragraph{Completeness.}  Completeness follows from the observation that $\pi(w_i) \neq \pi(w_j)$ for all valid witnesses $w$.

    \paragraph{Soundness and Argument of Knowledge.}
    It suffices to prove soundness in the hybrid experiment where in the first round the verifier sends $\pubparams$ which are generated in injective mode: $(\pubparams, \td) \gets \Setup_\calF(1^\secpar, \injective)$, as opposed to lossy mode. The dual-mode indistinguishability property implies that  this experiment is indistinguishable from the real execution.  Therefore, for any $\QPT$ prover $P^*$, its success probability changes by at most a negligible additive term.

    Condition on $\mathsf{Good}_{\injective}$. In injective mode, each $y_i$ with a valid opening corresponds to a unique preimage color $c_i$; assign an invalid color otherwise, and color isolated vertices arbitrarily. If $x\notin \lang$, then there must exist an edge $(i,j)\in \calE$ such that $c_i = c_j$ or with an invalid endpoint color. Therefore, $P^*$'s success probability is at most $1-1/|\calE|$ on good keys, and thus at most $1-1/|\calE|+\negl$ in a real execution, including setup failure and mode indistinguishability.

    The argument of knowledge extractor $\Ext$ interacts with $P^*$ using fresh injective keys in $k$ parallel repetitions until it receives $y_\calV$ for each instance, before sending any challenges. Then, it inverts each $y_\calV$ using the trapdoor to obtain a coloring $c_\calV$ for each instance, assigning arbitrary valid colors to isolated vertices. It checks which coloring is valid and outputs that coloring. On good keys, if extraction fails, each instance has an edge with equal or invalid endpoint colors. Independent challenges therefore give knowledge error at most $(1-1/|\calE|)^k+\negl$, including bad keys, mode switching, and sampling error. Taking $k\geq |\calE|\secpar$ makes this negligible.

    \paragraph{State preservation.}
    At the start of the protocol, the prover initializes their state to
    \[
        \propto \sum_{w} \alpha_w \ket{w} \otimes \sum_{\pi\in \mathsf{Perm}([3])} \ket{\pi} \otimes \bigotimes_{i\in V} \sum_{r_i \in \calR} \beta_{r_i} \ket{r_i}
    \]
    where $\beta_{r_i}$ is an amplitude determined by $\calF$ and $\pubparams$.
    Suppose that the prover (with the verifier's aid) ran all computations on $k \in \calV\backslash\{i,j\}$ before doing any computation on registers corresponding to $i$ and $j$.
    Explicitly, they apply the isometry from step 2, measure $y_k$, and run $\Recover_\calF(\td, \cdot, y_k)$ before doing any operations to the registers corresponding to $i$ and $j$.
    These operations share only computational-basis controls on $\calW,\calP$ and otherwise act on disjoint registers, so this reordering produces the same state.

    Condition on $\mathsf{Good}_{\lossy}$ and let $d$ be the averaged half-diamond recovery error per color from \Cref{def:dual-trapdoor-recovery}. Including all other registers as reference, a channel hybrid bounds the averaged trace distance after these recoveries by $\alpha=(|\calV|-2)d$ from the ideal state
    \[
        \propto \sum_{w} \alpha_w \ket{w} \otimes \sum_{\pi\in \mathsf{Perm}([3])} \ket{\pi} \otimes \sum_{r_i, r_j \in \calR} \beta_{r_i}\beta_{r_j}\ket{r_i, r_j}
    \]

    We will now rewrite this state in terms of the prover's possible last messages -- specifically, the revealed colors $c = (c_i, c_j)$. Fix any $c$ and any witness $w$. The possible settings of $\pi$ are constrained by $\pi(w_i) = c_i$ and $\pi(w_j) = c_j$. Since $w_i \neq w_j$ for valid witnesses, this leaves only one possibility for $\pi$ to map the last color. In other words, there is a unique $\pi_{c,w,(i,j)}$ which is consistent with $c$, $w$, and $(i,j)$. Thus, we can rewrite the state as
    \[
        \propto \sum_{w} \alpha_w \ket{w} \otimes \sum_{c} \ket{\pi_{c,w,(i,j)}} \otimes \sum_{r_i, r_j \in \calR} \beta_{r_i}\beta_{r_j}\ket{r_i, r_j}
    \]
    Measuring $c$, $r_i$, $r_j$, and the evaluations of $f_\pubparams$ on indices $i$ and $j$ collapses the state to
    \[
        \sum_{w} \alpha_w \ket{w} \otimes \ket{\pi_{c,w,(i,j)}}
    \]
    Finally, applying a unitary swapping
    \[
        \ket{w, \pi_{c,w,(i,j)}}
        \leftrightarrow
        \ket{w, \mathsf{id}}
    \]and discarding the register containing $\ket{\mathsf{id}}$ recovers the original joint state of $\calW,\calB$ from \Cref{def:witness-preserving-NP} exactly in the ideal experiment. Contractivity bounds the real averaged error on a good key by $\alpha$. The pure-target argument in \Cref{lem:parallel-repetition} bounds the probability that the joint state conditioned on the protocol randomness and measurement outcomes has trace distance greater than $\tau$ from the input by $\nu_{\lossy}+\alpha/\tau^2$, which is negligible for an appropriate negligible $\tau$.
    \end{proof}

\paragraph{Quantitative preservation under repetition.}
For the Gaussian construction, encode each color by two bits and put $\varepsilon_\kappa=2^{-4\kappa}/2+\epsilon_{\rm prep}\leq2^{-4\kappa}$. On a graph with $v$ vertices, $k$ fresh-key repetitions have trace distance from the input, after averaging over protocol randomness and measurement outcomes, bounded by
\[
 \Gamma_{\NP}\leq\sum_{j=1}^k\nu_{\lossy,j}
       +2k(v-2)\varepsilon_\kappa+\epsilon_{\rm aux}.
\]
This bound holds on the joint state of $\calW,\calB$ from \Cref{def:witness-preserving-NP}, for every valid witness input, after discarding the transcript. Here $\epsilon_{\rm aux}$ accounts for support-preserving permutation preparation and bounded-time classical sampling. Taking $\kappa\geq\secpar+\lceil\log_2(2kv+1)\rceil+\kappa_0$, rescaling all lattice parameters accordingly, and $\epsilon_{\rm aux}\leq2^{-4\kappa}$ makes the repeated argument $(1-2^{-\secpar})$-witness-preserving in the sense of \Cref{def:witness-preserving-NP}, by the preceding proof. Honest acceptance remains perfect: evaluation randomness stays in its support, and recovery occurs after the verdict.

\section{A Gentle Compiler for CVQC}\label{sec:gentle}

In this section, we show how to make any non-adaptive $\CVQC$ ``gentle''. When running a quantum computation locally, one may minimize the damage done by measuring only the decision bit. In our compiled $\CVQC$, we achieve the same behavior -- the end result of running the compiled protocol is as if the prover had simply run the original protocol coherently, acting as both parties, and measured \emph{only} the verifier's decision bit.

\begin{definition}[Gentle Protocol]\label{def:gentle}
    An interactive protocol $(P', V')$ is \emph{gentle} with respect to an interactive protocol $(P, V)$ if the following holds.

    Let $\tilde{V}_r$ be the quantum circuit which computes $P$'s messages coherently, then outputs $V$'s decision bit in an execution of $(P, V)$ where $V$ uses random coins $r$.\footnote{Although the case defined here, where the coins are measured, is already quite useful, one might also hope to purify the verifier's random coins (and therefore its messages). Such an approach seems to require verifiably delegating the verifier's behavior to the prover, which is more akin to verifying a \emph{sampling} problem. Our approach sidesteps the issue by having the verifier choose its own messages.} Let $\Pi^{\mathsf{acc}}_r = \tilde{V}_r^\dagger \Pi_{\mathsf{Accept}} \tilde{V}_r$ be the projector onto $V$ outputting $\Accept$ in this execution and let $\Pi^{\mathsf{rej}}_r = I - \Pi^{\mathsf{acc}}_r$.
    Let $\ket w$ be any joint pure state of the input witness and an external register $\calR$ on which the protocol acts as the identity; $0^*$ denotes the initialized base-prover workspace. Let
    \[
        \sum_{d\in \{\mathsf{acc}, \mathsf{rej}\}} \rho_{P',r,d} \otimes \ketbra{d}
        \gets (P'(\ket{w, 0^*}), V')
    \]
    denote the joint output conditioned on the base verifier's coins being $r$, averaged over the compiler's additional randomness and measurement outcomes. Here $V'$ outputs $d$, and $\rho_{P',r,d}$ is the subnormalized state of the recovered witness, base-prover workspace, and untouched register $\calR$; compiler transcript records and additional auxiliary registers are discarded.

    Then for all random coins $r$ and all states $\ket{w}$,
    \[
        \frac{1}{2}\left\|
            \begin{aligned}&\sum_{d\in \{\mathsf{acc}, \mathsf{rej}\}} (\Pi^{d}_{r}\otimes I_{\calR}) \ketbra{w,0^*} (\Pi^{d}_{r}\otimes I_{\calR}) \otimes \ketbra{d}
            \\
            &-  \sum_{d\in \{\mathsf{acc}, \mathsf{rej}\}} \rho_{P',r,d} \otimes \ketbra{d}
            \end{aligned}\right\|_1
        = \negl
    \]
    When we wish to be more-fine-grained, we say that a protocol is $1-\delta$-gentle if the trace distance is at most $\delta$.
\end{definition}

A gentle protocol allows us to treat the input state as if it were measured by the binary-outcome projector $\{\Pi_r^{\mathsf{acc}}, \Pi_r^{\mathsf{rej}}\}$ corresponding to the verifier. The gentle measurement lemma \cite{DBLP:journals/tit/Winter99} then implies that the distance of the residual state to the original state is related to the probability of obtaining the measurement result. As a special case, if the verifier would accept with overwhelming probability, then the protocol disturbs the witness negligibly, i.e. it is \emph{witness-preserving}.

Furthermore, it is possible to leverage the binary-outcome projectors defining the disturbance even when the starting $\CVQC$ has low completeness.
Section\Cref{sec:wit-pres-cvqc} applies the compiler to a high-completeness base protocol. Appendix D  studies single-copy amplification from lower-completeness protocols using state repair.

\begin{theorem}\label{thm:gentle}
    Let $\calL=(\calL_\yes,\calL_\no)\in \QMA_{a,b}$ with a $\QMA$ verifier $V_\QMA$. Assume the existence of the following primitives:
    \begin{enumerate}
        \item A dual-mode trapdoor function family with $(1-\delta_1)$-state recovery.
        \item A $(1-\delta_2)$-witness-preserving argument for $\NP$.
    \end{enumerate}

    \noindent Then there exists a generic compiler which converts any non-adaptive $\CVQC$ protocol $(P,V)$ w.r.t.\ $V_\QMA$ where the prover runs in strict polynomial time into a new $\CVQC$ protocol $(P',V')$ for $\calL$ (described in  \Cref{constr:cvqc-gentle}) which is $(1-\Delta)$-\textbf{gentle} with respect to $(P, V)$ (\Cref{def:gentle}), where
    \[
        \Delta=\nu_{\lossy}+\sum_{i=1}^{\ell}d_i+\Gamma_{\NP}.
    \]
    Here $\nu_{\lossy}$ is the probability that the shared outer key is not in $\mathsf{Good}_{\lossy}$, and $d_i$ bounds the half-diamond error from \Cref{def:dual-trapdoor-recovery} for the entire $i^{th}$ message on a good key. The quantity $\Gamma_{\NP}$ bounds the trace distance between the reduced joint output state of the witness, untouched external register, and verdict, and the original joint witness--external state tensored with $\ketbra{\mathsf{Accept}}$. The output is averaged over protocol randomness and measurement outcomes, with the transcript and additional auxiliary registers discarded. These bounds are uniform over inputs to recovery and valid witnesses for each public NP instance. One may take $d_i\leq\delta_1+\beta$ and $\Gamma_{\NP}\leq\delta_2+\mu_2$, using $\beta$ from \Cref{def:dual-trapdoor-recovery} and $\mu_2$ from \Cref{def:witness-preserving-NP} for the entire NP argument.

    Moreover, it maintains the following guarantees from the original protocol.
    \begin{enumerate}
        \item {\bf Computational soundness:} If $(P,V)$ has computational soundness $s$, then $(P',V')$ has computational soundness $s+\negl$. If $(P,V)$ has $(\eta,s)$-testable soundness, then $(P',V')$ has $(\eta/2,s+\negl)$-testable soundness for sufficiently large security parameters.

        \item {\bf Completeness: } If $(P, V)$ has completeness~$c$, then for every $x\in\calL_\yes$, the protocol $(P',V')$ has completeness $c-\negl$, using any single witness for $V_{\QMA,x}$.

        \item{\bf Complexity:}  The communication and computational overhead are polynomial in $|x|$, $\secp$, the requested accuracy parameter $h$, and the size of the base protocol.

        \item {\bf Argument of Knowledge:} If $(P,V)$ is a $(d, s,\delta)$ (testable) argument of knowledge, then $(P',V')$  is a $(d+1, s+\negl,\delta-\negl)$ (testable) argument of knowledge.
    \end{enumerate}
\end{theorem}

\paragraph{Accuracy and instance length.}\label{par:compiler-accuracy}
For the instantiation from Sections~\ref{sec:dual-trapdoor-recovery} and~\ref{sec:state-pres-np}, fix an accuracy parameter $h\geq\secp$. Let $L$ bound the total number of outer message bits, and let $v,e\geq1$ bound the vertices and edges in a reversible reduction of the NP relation to 3-coloring, with its workspace retained and later uncomputed. Use $k=\lceil e(h+4)\ln2\rceil$ parallel NP repetitions with fresh keys. With $K=k+1$ distinct keys and at most $M=L+2kv$ recovered bits, the bounds in Sections~\ref{sec:dual-trapdoor-recovery} and~\ref{sec:state-pres-np} give a uniform gentleness bound
\[
    \Delta\leq\Gamma:=K2^{-8\kappa}+M2^{-4\kappa}+\epsilon_{\rm aux}.
\]
Here $\epsilon_{\rm aux}$ bounds the total error of the remaining auxiliary preparations and sampling, and $L,v,e$ bound all instances of the prescribed lengths. Choose all lattice parameters as in \Cref{constr:fns-with-state-recovery}, using
\[
    \kappa\geq h+|x|+\lceil\log_2(K+M+2)\rceil+\kappa_0,
    \qquad \epsilon_{\rm aux}\leq2^{-h-2}.
\]
Then $\Gamma\leq2^{-h}\leq2^{-\secp}$. The counts are polynomial in the input and protocol sizes and in $h,\kappa$, so a polynomially bounded $\kappa$ suffices. Every compiler setup uses this $\kappa$; the base protocol's logical parameter remains fixed. Computational mode-switching losses are separate from $\Gamma$ and negligible in $\kappa$ under the stated LWE assumption; $\kappa\geq|x|+\secp$ makes the reductions polynomial in $\kappa$. The NP repetition term is at most $2^{-h-4}$; including setup and sampling, the total statistical soundness loss is at most $2^{-h}$.

The construction is given in \Cref{sec:gentle-construction} and its properties are analyzed in \Cref{sec:gentle-analysis}.
As a special case, if the starting $\CVQC$ has overwhelming completeness, then the compiled $\CVQC$ disturbs the witness negligibly.

\begin{definition}[Witness Preserving $\CVQC$]\label{def:witness-pres}
    A $\CVQC$ is \textbf{witness-preserving} if for every negligible function $\mu_1$ there exist negligible functions $\mu_2,\mu_3$ such that for every $\secpar\in \bbN$, every statement $x\in \calL_\yes$, and every normalized joint input $\ket{\tilde{w}}_{\calA,\calB}$ satisfying
    \[
        \bra{\tilde{w}}(M_x\otimes I_{\calB})\ket{\tilde{w}}\geq1-\mu_1(\secpar)
    \]
    where $M_x$ is the original verifier's witness acceptance operator and $\calB$ is an untouched reference register, at the end of an execution of $(P'(\ket{\tilde{w}}_{\calA,\calB}), V')$, where $P'$ operates only on register $\calA$, the prover $P'$ outputs a state which is negligibly far from $\ket{\tilde{w}}_{\calA,\calB}$:
    \[
        \Pr_{(\rho,b) \gets (P(\ket{\tilde{w}}), V)(1^\secpar, x)}\left[
        \frac{1}{2}\big\|\rho - \ketbra{\tilde{w}}\big\|_1 \leq \mu_3(\secpar)
        \right]
        \geq 1-\mu_2(\secpar),
    \]
    where $\rho$ is the residual joint state on $\calA\calB$, conditioned on the protocol randomness and measurement outcomes, with auxiliary registers discarded.
\end{definition}

\begin{corollary}[Witness-Preserving $\CVQC$]\label{coro:wit-pres-cvqc}
    Let $\calL=(\calL_\yes,\calL_\no)\in \QMA_{a,b}$ with a $\QMA$ verifier $V_\QMA$. Assume the existence of the following primitives:
    \begin{enumerate}
        \item A dual-mode trapdoor function family with state recovery.
        \item A witness-preserving argument for $\NP$.
        \item A non-adaptive $\CVQC$ $(P, V)$ for $\QMA_{a,b}$ with completeness $1-\negl$ and a strict polynomial-time prover.
    \end{enumerate}
    Then the gentle compilation preserves any normalized joint input that the base honest protocol accepts with overwhelming probability, with the conditional joint residual-state guarantee in \Cref{def:witness-pres}. It satisfies \Cref{def:witness-pres} when the base protocol has this acceptance guarantee for all inputs covered by that definition.
\end{corollary}

\subsection{Construction}\label{sec:gentle-construction}

\paragraph{Notation for the underlying non-adaptive  $\CVQC$ protocol.}
Before giving our construction, we define some additional notation for the underlying $\CVQC$, similar to the notation used in \Cref{sec:cvqc}.
\begin{itemize}
    \item $\ell=\ell(|x|,\secp)$ denotes the number of rounds in the protocol.
    \item The verifier is denoted by $V=(V_1,V_2)$, and its messages by $(q_1,\ldots,q_\ell)\gets V_1(x,1^\secp)$.
    \item The prover's messages are denoted by $(a_1,\ldots,a_\ell)$. \item For every $i\in[\ell]$ we denote by
\[
q_{[i]}=(q_1,\ldots,q_i)~~\mbox{ and }~~a_{[i]}=(a_1,\ldots,a_i)
\]\item We denote the partial transcript in the first $i$ rounds by
\[
\tau_{[i]}=(q_1, a_1, \ldots, q_i, a_i).
\]
\item We denote the state of the prover after the first $i$ rounds as $\ket{\psi_{\tau_{[i]}}}$, and denote its initial state by $\ket{\psi_\emptyset}$.
\item For every round $i\in[\ell]$, and for every message $q_i$ that the verifier sends in the $i$'th round, denote by $P_{i} = P_{i}(q_{i})$ the unitary that the prover applies to its state upon receiving $q_{i}$ from the verifier.  Namely,
\[
    P_{i}: \ket{\psi_{\tau_{[i-1]}}}
    \mapsto
    \sum_{a_{i}} \left(\beta_{a_{[i]}} \ket{a_{[i]}} \otimes \ket{\psi_{\tau_{[i]}}}\right)
\]
Such a unitary exists for every strict-polynomial time prover.

\item For every $i\in[\ell]$ and every query sequence $q_{[i]}=(q_1,\ldots,q_i)$, we denote the coherent state of the prover after the $i$'th round, which is the state obtained by applying $\prod_{j=1}^{i} P_j$ to its initial state denoted by $\ket{\psi_\emptyset}$ (without doing any measurements),  by
\[
    \sum_{a_{[i]}} \beta_{a_{[i]}}\ket{a_{[i]}} \otimes \ket{\psi_{\tau_{[i]}}}
\]
After measuring the first $i$ answers its state becomes $\ket{\psi_{\tau_{[i]}}}$.

\end{itemize}

\begin{construction}[Gentle $\CVQC$]\label{constr:cvqc-gentle}
    Fix a $\QMA$ language $\lang = (\lang_\yes, \lang_\no)\in\QMA_{a,b}$ with verifier $V_\QMA$. To prove to the verifier that $x\in\calL_\yes$,  the prover  begins with a witness state $\ket{w}$ and the protocol proceeds as follows.
    \begin{enumerate}
        \item \textbf{Prover:}
        \begin{enumerate}
       \item
        Initialize three registers $(\regtranscript, \regprover, \regrand)$.
        \begin{itemize}
            \item
        $\regtranscript = (\regtranscript_1, \dots, \regtranscript_\ell)$ is initialized to $\ket{\vec{0}}$ and will contain a superposition over messages $a_{[i]}$ that the prover could have sent in the base $\CVQC$ protocol.
        \item
        $\regprover$ is initialized to $\ket{\psi_{\emptyset}} = \ket{w}\otimes \ket{\vec{0}}$ and will contain the prover's internal state $\ket{\psi_{\tau_{[i]}}}$ for the base $\CVQC$ protocol.
\item
        $\regrand = (\regrand_1, \dots, \regrand_\ell)$ is initialized to $\ket{\vec{0}}$ and will contain randomness used for evaluating the trapdoor family $\calF$.
         \end{itemize}
 \end{enumerate}
        \item \textbf{Verifier:} \begin{enumerate}
            \item Compute the verifier's internal state and queries $(\st, q_1, \dots, q_\ell) \gets V_1(x,1^\secpar)$ for the base $\CVQC$.
            \item Sample $(\pubparams, \td) \gets \Setup_\calF(1^\secpar, \lossy)$.
\item
        Send $\pubparams$ to the prover.
        \end{enumerate}

        \item For each round $i \in [\ell]$ of the base $\CVQC$ protocol:
        \begin{enumerate}
            \item \textbf{Verifier:} Send $q_i$ to the prover.

            \item \textbf{Prover:}
            \begin{enumerate}
                \item Apply $P_{i}(q_{i})$ to registers $(\regtranscript_i, \regprover)$, after which the prover's state is
                \[                \sum_{a_{[i]}}\alpha_{a_{[i]}}\ket{a_{[i]}}_\calA\otimes\ket{\psi_{\tau_{[i]}}}_\calP
                \]
                \item  Coherently evaluate $f_\pubparams$ on $\regtranscript_i$ and measure the result. Specifically, let $p$ be the probability mass function for $f_\pubparams$'s randomness and apply the isometry
                \[
                    \ket{a_{i}}_{\regtranscript_{i}} \otimes \ket{\vec{0}}_{\regrand_{i}}
                    \mapsto
                    \sum_{r_{i}\in \calR} \sqrt{p(r_i)}
                    \ket{a_{i}}_{\regtranscript_{i}} \otimes \ket{r_{i}}_{\regrand_i} \otimes \ket{f_{\pubparams}(a_{i}; r_i)}
                \]
                to registers $(\regtranscript_i, \regrand_i)$ and measure the last register to obtain $y_i$.

                At this step, the state of the prover is

            \begin{equation}\label{eqn:st}
            \begin{aligned}
                &\propto (I_{\regtranscript,\regprover,\regrand}\otimes\Pi_{y_{[i]}})\\
                &\quad\Biggl(\sum_{a_{[i]}}\alpha_{a_{[i]}}\ket{a_{[i]}}_{\regtranscript}
                  \otimes\ket{\psi_{\tau_{[i]}}}_{\regprover}\\
                &\qquad\otimes\sum_{r_{[i]}\in\calR^{i}}
                  \Bigl(\prod_{j=1}^{i}\sqrt{p(r_j)}\Bigr)
                  \ket{r_{[i]}}_{\regrand}\otimes\ket{f_\pubparams(a_{[i]};r_{[i]})}\Biggr).
            \end{aligned}
            \end{equation}
            where
            $\Pi_{y_{[i]}}$ is a projector onto $\ket{y_{[i]}} = \ket{(y_1,\dots,y_i)}$ and where
            \[f_{\pubparams}(a_{[i]};r_{[i]}) = (f_{\pubparams}(a_1;r_1),\dots, f_{\pubparams}(a_i;r_i)).
            \]

            \item Send $y_i$ to the verifier.
             \end{enumerate}

        \end{enumerate}

        \item \textbf{Verifier.} Send $\st$ to the prover.
        \item \label{item:prover-b}\textbf{Prover.}
        Project the state to answers that are accepted or rejected by $V_2$. Namely,
        \begin{enumerate}
            \item Denote by $U_{V_2}$ the unitary that computes the verdict $V_2(\st,a_{[\ell]})$.
            \item Denote by $M$ the measurement of this verdict.
            \item Apply $U_{V_2}^\dagger M U_{V_2}$ to the state, to obtain a verdict bit $b\in\{0,1\}$. Send $b$ to the verifier.
        \end{enumerate}

        The residual state becomes:
        \[
            \begin{aligned}
                &\propto (I_{\regtranscript,\regprover,\regrand}\otimes\Pi_{y_{[\ell]}})\\
                &\quad\Biggl(\sum_{a_{[\ell]}:~V_2(\st,a_{[\ell]})=b}\alpha_{a_{[\ell]}}\ket{a_{[\ell]}}_{\regtranscript}
                  \otimes\ket{\psi_{\tau_{[\ell]}}}_{\regprover}\\
                &\qquad\otimes\sum_{r_{[\ell]}\in\calR^{\ell}}
                  \Bigl(\prod_{j=1}^{\ell}\sqrt{p(r_j)}\Bigr)
                  \ket{r_{[\ell]}}_{\regrand}\otimes\ket{f_\pubparams(a_{[\ell]};r_{[\ell]})}\Biggr).
            \end{aligned}
        \]

        \item If $b=0$, skip the $\NP$ argument and record rejection, then continue with trapdoor release and recovery in steps \ref{step:recover1} and \ref{step:recover2}. Otherwise, continue.
        \item \textbf{Both:} The prover proves to the verifier that the statement $(\st, \pubparams, y_{[\ell]})$ is in the $\NP$ language
         \begin{equation}\label{eq:ndcvqc-np-language}
            \lang_{\NDCVQC} \coloneqq
            \left\{
                \big(\st, \pubparams, y_{[\ell]}\big) :\
                \begin{array}{c}
                    \exists\ \big(a_{[\ell]}, r_{[\ell]}\big) \text{ s.t. } \\
                     V_2\big(\st, a_{[\ell]}\big) = \mathsf{Accept}  \\
                     \land\ y_i = f_{\pubparams}(a_{i}; r_i) \quad\forall i\in[\ell]
                \end{array}
            \right\}
        \end{equation}
        using $\SPNP$, the witness-preserving argument for $\NP$. The prover uses the contents of registers $(\regtranscript, \regrand)$ as the witness. As in the convention in the accuracy paragraph above, it applies the reduction to 3-coloring coherently, retains its workspace, and uncomputes the reduction after the argument.

        \item \label{step:recover1}\textbf{Verifier:} If $b=0$, output $\Reject$; otherwise, output the result of $\SPNP$. Send $\td$ to the prover in either case.

        \item\label{step:recover2} \textbf{Prover:} Recover the post-state as follows.
        \begin{enumerate}
            \item For every $i\in[\ell]$ run the trapdoor function recovery algorithm $\Recover_{\calF}(\td, \cdot, y_i)$ on registers $(\regtranscript_i, \regrand_i)$.
            \item For $i=\ell,\ell-1,\ldots,1$ apply $P_{i}^\dagger$ to registers $(\regtranscript_i,\regprover)$.
        \end{enumerate}
    \end{enumerate}
\end{construction}

\subsection{Analysis}\label{sec:gentle-analysis}
\subsubsection{Gentleness}

\begin{claim}
    \Cref{constr:cvqc-gentle} is gentle (\Cref{def:gentle}) with respect to $(P, V)$.
\end{claim}
\begin{proof}
    Fix the base tape $r$ and retain the output registers from \Cref{def:gentle} throughout. Consider the following hybrid experiments:
    \begin{itemize}
        \item $\Hyb_0$ runs $(P'(\ket{w}), V')$ and outputs the prover's residual state $\rho$ along with the verifier's decision bit $b$.
        \item $\Hyb_1$ modifies $\Hyb_0$ by computing the transcript
        \[
            \sum_{a_{[\ell]}:~V_2(\st,a_{[\ell]})=b} \alpha_{a_{[\ell]}} \ket{a_{[\ell]}}_{\regtranscript} \otimes \ket{\psi_{\tau_{[\ell]}}}_{\regprover}
        \]
        and measuring if it is accepting \emph{before} evaluating $f_\pubparams$ on it. The overall description of this hybrid can be broken into the following steps:
        \begin{enumerate}
            \item The prover coherently computes the \emph{entire} transcript $a_{[\ell]}$, where the verifier's queries $q_{[\ell]}$ are pre-sampled, by iteratively applying each $P_i$ to its state $\ket{w}\otimes \ket{0^*}$.

            Then, it measures whether $V$ would accept $(a_{[\ell]}, q_{[\ell]})$ as a valid transcript for an execution $(P, V)$.

            \item \textbf{Evaluate $f_{\pubparams}$:} The prover evaluates $f_{\pubparams}$ on $a_{[\ell]}$ and measures the result. \label{claim-low-c-function-eval}
            \label{claim-low-c-measure-accept}
            \item \textbf{Prove in $\SPNP$:} If the base $\CVQC$ verifier would accept $a_{[\ell]}$, prove this fact using $\SPNP$. \label{claim-low-c-spnp-step}
            \item \textbf{Recover from $f_{\pubparams}$:} The prover runs $\Recover_\calF$ on the registers containing the transcript $a_{[\ell]}$ and $f_\pubparams$ evaluation randomness, using the $\td_\pubparams$ from the verifier.\label{claim-low-c-function-recover}
            \item The prover uncomputes its transcript by applying $P_i^\dagger$ for $i=\ell,\ell-1,\ldots,1$.
        \end{enumerate}

        \item $\Hyb_2$ removes step \ref{claim-low-c-spnp-step}.

        \item $\Hyb_3$ removes steps \ref{claim-low-c-function-eval} and \ref{claim-low-c-function-recover}.

        Observe that $\Hyb_3$ is precisely the result of measuring $\ket{w}\otimes \ket{0^*}$ with respect to $\{\Pi_r, I - \Pi_r\}$ where $r$ is the internal random coins of $V$.
    \end{itemize}

    $\Hyb_0 = \Hyb_1$ because each function evaluation acts on a retained answer register and fresh randomness, so it commutes with subsequent base-prover operations and with the computational-basis predicate defining the verdict. The NP argument remains after the function evaluations in both hybrids.

    On the $b=0$ branch, the NP argument is absent and the verdict is rejection in both $\Hyb_1$ and $\Hyb_2$. For each fixed public NP instance $(\st,\pubparams,y_{[\ell]})$, the accepting branch is supported on valid witnesses. Applying the definition of $\Gamma_{\NP}$ to this branch and summing its subnormalized states bounds the trace distance between $\Hyb_1$ and $\Hyb_2$ by $\Gamma_{\NP}$, without dividing by the acceptance probability. On a good outer key, the recovery bound of \Cref{lem:eval-measure-recover}, applied to each complete message with all other registers untouched, bounds the distance between $\Hyb_2$ and $\Hyb_3$ by $\sum_i d_i$. Discarding the measured function outputs is part of this bound. Averaging over the shared key adds at most $\nu_{\lossy}$ once. The triangle inequality gives the claimed bound $\Delta$ on the joint output in \Cref{def:gentle}.
\end{proof}

\subsubsection{Completeness and Soundness}

\begin{claim}\label{claim:gentle-completeness}
    Let $(P, V)$ be a $\CVQC$ and let $(P', V')$ be result of compiling $(P, V)$ using \Cref{constr:cvqc-gentle}. If $(P,V)$ has completeness $c$, then $(P', V')$ has completeness $c-\negl$.
\end{claim}
\begin{proof}
    By the completeness of $(P, V)$, the prover's measurement in step 5 outputs a verdict bit $b=1$ with probability $\geq c$. Whenever this occurs, the prover's state is supported on witnesses for $\lang_{\NDCVQC}$, so completeness of $\SPNP$ implies that the verifier accepts with overall probability $c - \negl$. For the concrete NP argument of \Cref{constr:state-pres-np}, honest acceptance is perfect and recovery occurs after the verdict, so the compiled completeness is at least $c$.
\end{proof}

\begin{claim}\label{claim:gentle-soundness}
    Let $(P, V)$ be a $\CVQC$ and let $(P', V')$ be result of compiling $(P, V)$ using \Cref{constr:cvqc-gentle}.

    If $(P,V)$ has $(\eta,s)$-testable soundness, then $(P',V')$ has $(\eta/2,s+\negl)$-testable soundness for sufficiently large security parameters. If $(P,V)$ is a $(d,s,\delta)$-testable proof of quantum knowledge for the relation decided by the $\QMA$ verifier $V_\QMA(1^\secpar,\cdot,\cdot)$, then $(P',V')$ is a $(d+1,s+\negl,\delta-\negl)$-testable proof of quantum knowledge for sufficiently large security parameters.

\end{claim}
\begin{proof}
    Let $x \in \lang_\yes \cup \lang_\no$.
    First, note that when $V$ is composed of two possible verifier algorithms $V_{\test}$ and $V_{\honest}$, then $V'$ can also be decomposed into analogous verification algorithms $V'_{\test}$ and $V'_{\honest}$. We consider $V_{\test}$ to be the ``test'' verifier and $V_{\honest}$ to be the ``honest'' verifier from \Cref{def:2part} (or \Cref{def:2partpok} for argument of knowledge).
    We show that for every $\QPT$ prover $P^*$ which passes  $V'_{\test} (x)$ (resp., $V'_{\honest} (x)$)
    with probability $p$ in an execution of $(P', V')$, there exists a $\QPT$ prover $\widetilde{P}^*$
    that convinces $V_{\test}(x)$ (resp., $V_{\honest}(x)$) in an execution of the underlying $(P,V)$ protocol with probability  $\geq p - \negl$.
    Thus, if $(P, V)$ is sound, then so is $(P',V')$, and if it is a proof of quantum knowledge, then so is $(P', V')$.

    $\widetilde{P}^*$ begins by sampling $(\pubparams, \td)\gets \Setup_{\calF}(1^\secpar, \injective)$. Then, for each round $i$ of the underlying $(P,V)$ protocol, it receives the external verifier's query $q_i$ and forwards it internally to $P^*$, who sends back a message $y_i$. $\widetilde{P}^*$ runs the trapdoor function extractor $a'_i \gets \Ext_{\calF}(\td, y_i)$, then sends $a'_i$ to the external verifier, using a fixed message of the prescribed length if extraction returns $\bot$.

    We now analyze the success probability of $\widetilde{P}^*$. Let $\st$ be the verifier's final internal state.
    Suppose that $(\st,  \pubparams, y_{[\ell]}) \in \lang_\NDCVQC$. Then there exist $(a_{[\ell]}, r_{[\ell]})$ such that the following two conditions hold:
    \begin{gather*}
        V_2(\st, a_{[\ell]}) = \Accept
        \\
        y_i = f_{\pubparams}(a_i; r_i) \quad \forall i\in[\ell]
    \end{gather*}where $a_{[\ell]} = (a_1, \dots, a_\ell)$, where $r_{[\ell]} = (r_1, \dots, r_\ell)$. In other words, if $\tilde{P}^*$ had sent responses $a_{[\ell]}$ in the execution of $(P, V)$, then the external verifier $V$ would have accepted.

    On $\mathsf{Good}_{\injective}$, the trapdoor property (\Cref{def:dual-trapdoor-recovery}) gives $a'_i=a_i$ simultaneously for all rounds with valid openings, including messages chosen adaptively after seeing the key. The complement of this event contributes at most $\nu_{\injective}$ once for the shared key to the negligible reduction loss.
    Furthermore, the first condition says that sending each $a_i$ to the verifier constitutes an accepting transcript.
    Thus, it suffices to show that
    \[
        \Pr_{\widetilde{P}^*}\big[(\st, \pubparams, y_{[\ell]}) \in \lang_\NDCVQC] \geq p - \negl
    \]

    Consider the following hybrid experiments.
    \begin{itemize}
        \item $\Hyb_0$ is an execution of $(P^*, V'(x))$ up to the stage where $V$ decides its output. The experiment ends \emph{before} $V'$ sends $\td$ to $P^*$ for witness recovery. Let $b$ be the decision bit of $V'$: it is rejection if the prover reported $0$, and otherwise it is the decision in the execution of $\SPNP\langle P^*, V\rangle(1^\secpar, (\st, \pubparams, y_{[\ell]}))$.

        \item $\Hyb_1$ replaces the trapdoor function specified by $\pubparams$ with an \emph{injective-mode} function $(\pubparams, \td)\gets\Setup_\calF(1^\secpar, \injective)$.

    \end{itemize}
    $\Hyb_1$ can be run using an externally generated $\pubparams$ and generating everything else internally. Therefore dual-mode indistinguishability of $\calF$ implies that
    \[
        \left|\Pr_{\Hyb_1}[b = \Accept] - \Pr_{\Hyb_{0}}[b = \Accept]\right| = \negl
    \]
    The soundness of $\SPNP$ implies that
    \[
        \Pr_{\Hyb_1}[b = \Accept] \leq \Pr_{\Hyb_{1}}[(\st, \pubparams, y_{[\ell]}) \in \lang_\NDCVQC] + \negl
    \]
    Observe that $\Hyb_{1}$ is distributed precisely as in the internal evaluation of $\widetilde{P}^*$.
    Combining this with the observations from before, $\widetilde{P}^*$ convinces the verifier in an execution of $(P,V)$ with probability at least $p-\negl$. This loss includes the mode-switching advantage, the NP soundness error, and $\nu_{\injective}$.

    For ordinary computational soundness, applying the same reduction to the full verifier shows that the compiled acceptance probability is at most $s+\negl$ whenever the base protocol has soundness $s$.

    For testable soundness, fix $x\in\lang_\no$. Let $p'_\test,p'_\honest$ be the acceptance probabilities of $P^*$ against the compiled test and check verifiers, and let $\widetilde{p}_\test,\widetilde{p}_\honest$ be those of the same reduced prover $\widetilde{P}^*$ against the base verifiers. Let $\nu=\nu(\secp,|x|)$ be a negligible upper bound on the reduction loss for both branches, so that
    \[
        \widetilde{p}_\test\geq p'_\test-\nu,
        \qquad
        \widetilde{p}_\honest\geq p'_\honest-\nu.
    \]
    Suppose $p'_\test\geq1-\eta/2$. For the concrete instantiation, choose $h$ in \Cref{par:compiler-accuracy} so that $2^{-h}\leq\eta/4$. The statistical part of $\nu$ is then at most $\eta/4$; since $\kappa\geq|x|+\secp$, the computational part is also at most $\eta/4$ for sufficiently large parameters. For general primitives the same conclusion follows from negligible reduction loss on polynomially bounded instances. Thus $\nu\leq\eta/2$. Hence
    \[
        \widetilde{p}_\test\geq p'_\test-\nu
        \geq1-\eta/2-\nu\geq1-\eta.
    \]
    The $(\eta,s)$-testable soundness of $(P,V)$ therefore gives $\widetilde{p}_\honest\leq s$, and consequently
    \[
        p'_\honest\leq\widetilde{p}_\honest+\nu\leq s+\nu,
    \]
    proving the claimed $(\eta/2,s+\negl)$-testable soundness.
    Furthermore, suppose $(P,V)$ is a $(d,s,\delta)$-testable proof of quantum knowledge. For this guarantee, take any $x\in\lang_\yes\cup\lang_\no$, let $\brho$ denote the initial prover state, and again let $\nu$ be a negligible upper bound on both branch losses. Write $\theta_j=1/(|x|^j\secp^j)$ for $j\in\{d,d+1\}$. If
    \[
        p'_\test\geq1-\theta_{d+1},\qquad p'_\honest\geq s+\nu,
    \]
    then, choosing $h$ also so that $2^{-h}\leq(\theta_d-\theta_{d+1})/4$ in the concrete instantiation, the same reasoning gives $\nu\leq\theta_d-\theta_{d+1}$ for sufficiently large security parameters, and therefore
    \[
        \widetilde p_\test\geq1-\theta_{d+1}-\nu\geq1-\theta_d,
        \qquad
        \widetilde p_\honest\geq p'_\honest-\nu\geq s.
    \]
    Thus the reduced prover meets both hypotheses of the underlying knowledge guarantee. Running its extractor $\cE$ on $\widetilde P^*$ gives
    \[
        \Pr[\mathsf{Accept}\gets\calR(1^\secp,x,\sigma):\sigma\gets\cE^{\widetilde P^*}(\brho,1^\secp,x)]\geq\delta(\secp,|x|),
    \]
    which implies the claimed $(d+1,s+\negl,\delta-\negl)$ knowledge guarantee.
\end{proof}

\section{Witness-Preserving CVQC for QMA}\label{sec:wit-pres-cvqc}

In this section, we leverage the gentle compiler introduced in \Cref{sec:gentle} to build witness-preserving $\CVQC$ for $\QMA$ with almost-perfect completeness and soundness, and using one copy of the same witness as the original $\QMA$ verifier.

Clearly, we cannot hope to obtain almost-perfect completeness with {\em any} witness, since the witness can be a superposition of a ``good'' witness and ``junk.''  Instead, we require that the $\CVQC$ match the original $\QMA$ verifier on high quality witnesses: if the $\QMA$ verifier would accept $\ket{w}$ with overwhelming probability, then the $\CVQC$ verifier accepts the same witness with overwhelming probability.

\begin{theorem}\label{thm:wit-pres-cvqc}
Fix any $\calL\in \QMA_{a,b}$ with a $\QMA$ verifier $V_\QMA=\{V_{\QMA,x}\}_{x\in\{0,1\}^*}$, where $a-b\geq1/{\sf poly}(|x|)$. Fix a polynomial $q$ bounding the instance length, including any lower-bound padding.
Then there exists a $\CVQC$ protocol $(P,V)$ for $\calL$ w.r.t.\ $V_\QMA$  with the following properties:
\begin{itemize}
    \item {\bf Computational soundness: } It has $\negl$ computational soundness error (as defined in \Cref{def:CVQC}), assuming the post-quantum hardness of ${\sf LWE}$ with superpolynomial modulus-to-noise ratio.

    \item {\bf Completeness:} For every negligible function $\mu_1$ there exists a negligible function $\mu_2$ s.t.\ for every $x\in\calL_\yes$ and every ({\em single}) witness $\ket{w}$ s.t.\ \begin{equation}\label{eqn:goodW}
    \Pr[V_{\QMA,x}(1^\secp,\ket{w})=1]=1-\mu_1(\secp)
    \end{equation}
    it holds that
    \[
    \Pr[(P(\ket{w},V)(1^\secp,x)=1]=1-\mu_2(\secp).
    \]\item {\bf Witness preserving:} For every negligible function $\mu_1$ there exist negligible functions $\mu_2,\mu_3$ such that for every $x\in\calL$ and
    every normalized joint input $\ket{\tilde{w}}_{\calA,\calB}$ satisfying
    \[
        \bra{\tilde{w}}(M_x\otimes I_{\calB})\ket{\tilde{w}}\geq1-\mu_1(\secp),
    \]
    where $M_x$ is the original verifier's witness acceptance operator and $\calB$ is an untouched reference register,
    \[
        \Pr_{(\rho,b) \gets (P(\ket{\tilde{w}}), V)(1^\secpar, x)}
        \left[\frac{1}{2}\big\|\rho-\ketbra{\tilde{w}}\big\|_1 \leq \mu_3(\secpar)\right]
        \geq 1-\mu_2(\secpar),
    \]
    where $P$ operates only on register $\calA$, and $\rho$ is the residual joint state on $\calA\calB$, conditioned on the protocol randomness and measurement outcomes, with auxiliary registers discarded.
 \item {\bf Complexity:}  The number of rounds, the communication complexity and the computational complexity are ${\sf poly}(\secp,|x|)$.

 \item {\bf $\eta$-Argument of Knowledge up to Complex Conjugation:} It is an $\eta$-classical argument of quantum knowledge (\Cref{def:pok}) for the relation $\calR_{\pm}$ that, on an extracted bit $z$ and state $\sigma$, runs $V_\QMA$ on $\sigma$ if $z=0$ and $\overline{V_\QMA}$ if $z=1$, where the bar denotes entrywise complex conjugation in the computational basis. Here $\eta=1/p(\secp)$ for some polynomial $p$.

    \end{itemize}
All of the above holds for $\secp<n=|x|\leq q(\secp)$; the fixed polynomial $q$ includes the padded length if the lower bound is ensured via padding. The negligible bounds may depend on $q$ and $V_\QMA$ and, for computational guarantees, on the efficient prover, uniformly over allowed instances and advice states.
\end{theorem}

Our approach to obtaining this theorem works in two steps:
\begin{enumerate}
    \item Construct a $\CVQC$ which preserves the witness, so that it can be re-used in subsequent $\CVQC$ repetitions.
    \item Amplify soundness via sequential repetition.
\end{enumerate}
In \Cref{sec:ma-natarajan}, we gently compile a specific $\CVQC$ which is obtained by combining Mahadev's measurement protocol with Ma and Natarajan's recent high-completeness XZ Hamiltonians~\cite{MN25}. The $\CVQC$ pre-compilation has completeness $1-\negl$, so gently compiling it yields a witness preserving $\CVQC$. However, it also has very high (testable) soundness error, accepting invalid statements $x\in \lang_\no$ with probability $1-1/\poly[|x|]$. In section \Cref{sec:sequential-rep} we then show that sequential repetition amplifies soundness to negligible error.

\subsection{Witness-Preservation via High Completeness XZ Hamiltonians}
\label{sec:ma-natarajan}

In this section, we observe that recent advances in frustration-free ZX Hamiltonians can be combined with Mahadev's measurement protocol to create a CVQC with high completeness.
Compiling the resulting protocol via the gentle compiler then yields a witness-preserving $\CVQC$, which can be sequentially repeated to amplify soundness.

\begin{theorem}[Implicit in \cite{Mahadev18}]\label{thm:mahadev}
    Assuming the post-quantum hardness of LWE with superpolynomial modulus-to-noise ratio, for every ZX Hamiltonian $H$ with $1-\negl$ completeness and soundness error $1-g$, where $g\geq1/{\sf poly}(\secp,|x|)$, there exists a $\CVQC$ protocol with $1-\negl$ completeness and $(\eta,s)$-testable soundness (\Cref{def:2part}) using the same witnesses as $H$, where
    \[
        \eta=g^2/16,\qquad s=1-g/4,
    \]
    for sufficiently large security parameters.
\end{theorem}
\begin{proof}[Proof Sketch]
    We use the completeness and soundness of Mahadev's measurement protocol applied to $H$~\cite[Lemmas 5.3 and 6.12 and the proof of Claim 6.13 in Section 6.4, full version arXiv:1804.01082v3]{Mahadev18}. Lemma 5.3 applies to any supplied state; sampling an $X/Z$-diagonal projector and checking its measurement outcome is efficient classical post-processing of the measurements covered by Lemma 6.12.

    The completeness guarantee states that an honest prover passes test rounds perfectly and passes the measurement protocol's Hadamard-round checks with $1-\negl$ probability. Conditioned on passing these checks, the decoded outcomes are negligibly far in statistical distance from directly measuring the committed state. Since the Hamiltonian test accepts that state with $1-\negl$ probability, the honest CVQC prover passes either branch with $1-\negl$ probability.

    For soundness, fix $x\in\lang_\no$ and an adversarial prover. Let $p_T$ and $p_H$ be its rejection probabilities in the measurement protocol's test and Hadamard rounds, respectively, averaged over the verifier's basis choice. Let $q$ be its acceptance probability in the CVQC check branch, including the Hamiltonian test. The measurement soundness guarantee provides a state $\rho$ whose basis-measurement distributions approximate the decoded outcomes conditioned on passing the Hadamard-round checks, up to computational indistinguishability and statistical error $p_{h,H}+\sqrt{p_{h,T}}+\negl$ for each basis $h$. Here $p_{h,H}$ and $p_{h,T}$ are the corresponding rejection probabilities for that basis.
    Averaging over bases and applying Jensen's inequality, and using that the Hamiltonian test accepts $\rho$ with probability at most $1-g$, gives
    \[
        q\leq1-g+p_H+\sqrt{p_T}+\negl.
    \]
    Also $q\leq1-p_H$, since the check branch rejects whenever the Hadamard-round checks fail. Combining the two bounds, if $p_T\leq\eta=g^2/16$, then
    \[
        q\leq1-\frac g2+\frac{\sqrt{p_T}}2+\negl
        \leq1-\frac{3g}{8}+\negl
        \leq1-\frac g4=s
    \]
    for sufficiently large security parameters. This proves $(\eta,s)$-testable soundness.
\end{proof}

Originally, Mahadev obtains such high-completeness ZX Hamiltonians $H$ by parallel repetition of a single ZX Hamiltonian $H'$ with \emph{low} completeness. Thus, her main result requires many copies of the witness.

Recently, Ma and Natarajan constructed high-completeness ZX Hamiltonians which use only one copy of the original QMA witness. More specifically, they show a variant of Kitaev's clock construction where each term is natively diagonal in either the X or Z bases. Although their focus is on perfect completeness $\QMA_1$, their result can also be applied to any QMA language which has been amplified to $1-\negl$ completeness.\footnote{Checking that the witness is a history state of $V_\QMA$ is perfect. The error comes from checking the final output bit, which can reject with negligible probability in an honest computation of $V_\QMA$.}
Similarly to Kitaev's clock construction, an honest prover coherently prepares the normalized history state by appending a clock in $(T+1)^{-1/2}\sum_{t=0}^{T}\ket{t}$ and applying the first $t$ implemented gates controlled on the clock. This is an efficient reversible encoding of one copy of the input witness.

\begin{theorem}[Adapted from \cite{MN25}]\label{thm:ma-natarajan}
    Let $\calL\in\QMA_{1-\negl,\negl}$ with verifier $V_\QMA$. There is an efficient reversible honest encoding, using one copy of its witness, into a Hamiltonian $H_x=\sum_{j=1}^{M}\Pi_j$ with polynomially many six-local projector terms, each diagonal in the $X$ or $Z$ basis. For final circuit length $T$, an input accepted with probability $p_{\rm acc}$ has ideal history energy at most
    \[
        \langle H_x\rangle\leq\frac{1-p_{\rm acc}+2\xi}{T+1}.
    \]
    The uniform-term test rejects a state $\rho$ with probability $\Tr(H_x\rho)/M$. An implementation of the honest history preparation within trace distance $\zeta$ adds at most $\zeta$ to this rejection probability. With $\xi,\zeta=2^{-(\secp+n)}$, the honest test has completeness $1-\negl$, while on NO instances every state has energy at least $1/{\sf poly}(\secp,n)$ and test rejection probability at least $g=1/{\sf poly}(\secp,n)$. Testable soundness of the resulting $\CVQC$ follows from \Cref{thm:mahadev}.
\end{theorem}
\begin{proof}[Construction details]
Use Appendix~\ref{app:post-hoc-complex-hamiltonians}'s real compilation, whose honest encoding appends $\ket{-Y}$. This phase qubit is an unrestricted witness register for soundness; the other clean verifier and synthesis ancillas receive the usual initialization checks. Pad each at-most-two-qubit real gate to three wires: its determinant becomes $+1$. Real universality of $H$ and Toffoli~\cite{Shi02}, with efficient gate approximation, then approximates each such gate to error at most $\xi/L$ for $L$ real gates, giving total operator error at most $\xi$. The required overhead is polynomial in $L$ and $\log(1/\xi)$.

Replace each $H$ by $H\otimes H$ on its wire and an initialized spectator that does not affect acceptance. Apply the exact gate conversion and alternating-basis padding of~\cite[Lemmas 4.1 and 4.3, full version arXiv:2509.24390v1]{MN25}, followed by the history-state construction in Sections 4.3--4.4; Lemmas 4.2 and 4.4 give the $X/Z$-diagonal gates and six-local projector terms. All non-output history constraints vanish on the ideal history state of this final circuit. The output constraint in~\cite[Section 4.4, Eq.~(4)]{MN25} contributes its rejection probability divided by $T+1$; operator error $\xi$ changes acceptance by at most $2\xi$, giving the displayed energy bound. This is our adaptation from perfect to overwhelming completeness. The NO-instance acceptance bound from Appendix~\ref{app:post-hoc-complex-hamiltonians} also changes by at most $2\xi$, so the history construction retains its inverse-polynomial NO-instance gap by~\cite[Section 5, Eqs.~(9)--(10) and (21)]{MN25}. Sampling one of the $M$ terms uniformly gives the stated test normalization. Approximate history preparation adds at most $\zeta$ by trace-distance contractivity. The choices $\xi,\zeta=2^{-(\secp+n)}$ have polynomial overhead.
\end{proof}

For the original verifier in \Cref{thm:wit-pres-cvqc}, first apply Marriott--Watrous amplification (\Cref{thm:MW}) with completeness cutoff $t=(a+b)/2$ and soundness cutoff $u=(a+3b)/4$ and error $2^{-\secp}$, before forming the history state. Let $M$ be the original verifier's acceptance operator. On an input with rejection probability at most $\mu_1(\secp)$, the weight in eigenspaces of $M$ with eigenvalue below $t$ is at most
\[
    \frac{\mu_1(\secp)}{1-t}\leq\frac{2\mu_1(\secp)}{1-b}.
\]
Since $1-b\geq a-b\geq1/{\sf poly}(|x|)$ and $n=|x|\leq q(\secp)$ for the fixed polynomial $q$, this weight is negligible in $\secp$. The amplified verifier therefore rejects with probability at most $2\mu_1(\secp)/(1-b)+2^{-\secp}$. The implemented encoding is a unitary circuit on the witness and initialized ancillas. Applying its actual inverse recovers the unknown input exactly before protocol disturbance, and contractivity transfers any joint witness--reference disturbance bound after the interaction; the approximation error is charged to completeness, not to this inverse identity. The Hamiltonian's output check and Mahadev's measurement protocol add only negligible error, so there is a negligible function $\beta(\secp)$ bounding the base protocol's rejection probability on the original input, uniformly over these instance lengths. For an actual reduced witness input $\rho_{\calA}$ with $\Tr(M\rho_{\calA})\geq1-\mu_1$, \Cref{lem:mixed-to-pure} gives a convex decomposition with negligible bad weight and pure components of overwhelming acceptance. Pure-state completeness therefore extends to $\rho_{\calA}$ by linearity. All disturbance comparisons retain the untouched reference register $\calB$. This honest-input recovery argument is separate from the extraction argument below.

\paragraph{Extraction up to complex conjugation.}\label{par:conjugate-extraction}
We establish the base knowledge guarantee for the relation $\calR_{\pm}$ of \Cref{thm:wit-pres-cvqc}. Let $C$ be the final implemented real circuit, let $L=T+1$, and write $B_C$ for its acceptance operator on the witness and unrestricted phase qubit. Let $W$ be its history-state isometry, with every other input ancilla initialized to zero, and put $P=WW^\dagger$. For
\[
 H_0=H_x-H_{\rm out}
     =H_{\rm prop}+H_{\rm format}+H_{\rm in},
\]
we have $\ker H_0=\operatorname{im}W$ and
\[
 H_0\succeq\Delta(I-P),\qquad \Delta=1/(4L^2).
\]
We derive this initialized-history gap here for decoding. On the good clock subspace, undoing the controlled circuit prefixes gives the path Laplacian plus the input-ancilla penalty at time zero. Its clean-ancilla sector has nonzero eigenvalues at least $1/(4L^2)$. On any other input-ancilla sector its quadratic form is at least $E=\|v_0\|^2+\frac12\sum_{t=1}^{T}\|v_t-v_{t-1}\|^2$, and $\|v_t\|^2\leq4LE$ implies $\sum_t\|v_t\|^2\leq4L^2E$. On the invariant complement of the good clock subspace, $H_{\rm prop}+H_{\rm format}\succeq I/4$ by~\cite[Section 5, Eqs.~(12)--(21), using Corollary 6.6]{MN25}.

The decoder undoes the controlled circuit prefixes and the uniform clock preparation, checks the clock and clean input ancillas, and returns the witness and phase qubit, or failure. Its successful output on an arbitrary normalized state $\rho$ is the subnormalized state $\tau=W^\dagger\rho W$. Put $e=\Tr(H_x\rho)$ and $r=\Tr((I-P)\rho)\leq e/\Delta$. Since $H_{\rm out}$ is a projector, the Hilbert--Schmidt triangle inequality gives
\[
 \sqrt{\Tr(H_{\rm out}P\rho P)}\leq\sqrt e+\sqrt r.
\]
Using $W^\dagger H_{\rm out}W=(I-B_C)/L$, we obtain the unconditional bound
\begin{align*}
 \Tr(B_C\tau)
 &=1-r-L\Tr(H_{\rm out}P\rho P)\\
 &\geq1-e\bigl[\Delta^{-1}+L(1+\Delta^{-1/2})^2\bigr]
 \geq1-13L^3e.
\end{align*}
Approximate implementation of this decoder to negligible channel error subtracts only a negligible term; no copy or reinitialization of the input is used.

Let $A=f(M)$ be the amplified acceptance operator, with cutoffs $(t,u)$ fixed above and error $\epsilon_{\rm MW}=2^{-\secp}$, and let $P_{\mp}=\ketbra{\mp Y}$. Appendix~\ref{app:post-hoc-complex-hamiltonians} gives the ideal real acceptance operator
\[
 \widetilde A=A\otimes P_-+\overline A\otimes P_+,
 \qquad \|B_C-\widetilde A\|\leq2\xi.
\]
Measure the decoded phase qubit in the $Y$ basis and output $z=0$ on $-Y$ and $z=1$ on $+Y$. The remaining state is tested with $V_\QMA$ or $\overline{V_\QMA}$, respectively. Writing $M_{\pm}=M\otimes P_-+\overline M\otimes P_+$, the amplification inequality in both sectors gives $M_{\pm}\succeq u(\widetilde A-\epsilon_{\rm MW}I)$. Hence the probability of decoding and acceptance by $\calR_{\pm}$, counting failure as rejection, is at least
\[
 u\bigl(1-13L^3e-2\xi-\epsilon_{\rm MW}\bigr)-\negl.
\]

Write $m$ for the number of projector terms in $H_x$, and choose a fixed inverse-polynomial Hamiltonian test-gap lower bound $g$ small enough that $g\leq1/(26mL^3)$. The extractor of~\cite[Claim 6.4 and Section 6.3]{VidickZhang21} produces a state whose basis-measurement distributions approximate the decoded outcomes with statistical error $C_1\theta^{c_1}$ when test rejection is at most $\theta$, up to computational indistinguishability. Fix an inverse-polynomial tolerance with $C_1\theta^{c_1}\leq g/8$. For sufficiently large parameters the computational loss in the efficient Hamiltonian test is also at most $g/8$. A prover with check probability at least $1-g/4$ therefore yields an extracted averaged state with $e/m\leq g/2$. The displayed decoder bound is then at least $u/2$. Taking $d$ sufficiently large that $1/(n^d\secp^d)\leq\theta$ proves \Cref{def:2partpok} for $\calR_{\pm}$ with fixed extraction bound $\delta=u/2$.

Conjugation is allowed here because the phase qubit is unrestricted: either sector may carry the extracted witness. Conjugating a prover's initial state and all its quantum operations preserves the probabilities of its classical transcripts. The extractor reports which sector it obtains and never conjugates an unknown state.

Combining \Cref{thm:ma-natarajan} with \Cref{thm:mahadev}, there exists a $\CVQC$ protocol with $1-\negl$ completeness using one copy of the same witness as $V_\QMA$. With the conservative gap bound $g$ chosen above, choose fixed bounds $c_0=1-g/8$, $s_0=1-g/4$ and test tolerance $\eta_0=g^2/16$. Absorb compiler losses with $c_*=c_0-(c_0-s_0)/4$ and $s_*=s_0+(c_0-s_0)/4$, and use $c_*,s_*$ in \Cref{constr:sequential-repetition}. The gentle compiler allows tolerance $\eta_0/2$. For the knowledge conclusion for $\calR_{\pm}$ with base exponent $d$ established above, set $d_*=d+1$, use the fixed bound $\delta_*=\delta/2=u/4$, and further reduce the tolerance to $\min\{\eta_0/2,1/(n^{d_*}\secp^{d_*})\}$ before choosing $N$. These fixed margins absorb negligible losses for sufficiently large parameters; $N$ and the acceptance threshold are independent of the adversary and remain polynomially bounded.

Writing $\Delta=c_*-s_*$, the repetition count in \Cref{constr:sequential-repetition} is
\[
    N=\Theta\!\left(\secp\left(\frac{1}{\eta\Delta}+\frac{1}{\Delta^2}\right)\right).
\]
The soundness-only choices $\eta=\Theta(g^2)$ and $\Delta=\Theta(g)$ give $N=\Theta(\secp/g^3)$. The additional knowledge requirement may reduce $\eta$ further and increase $N$. This counts repetitions; total runtime and communication also include the compiled base protocol and its NP arguments. The construction establishes polynomial efficiency without claiming practical parameters.

For \Cref{thm:wit-pres-cvqc}, we compile this $\CVQC$ with \Cref{coro:wit-pres-cvqc} and sequentially repeat it to amplify soundness via \Cref{constr:sequential-repetition} (\Cref{thm:sequential-soundness}). Let $N$ be the number of repetitions. For repetition $i$, let $b_i$ be the original, uncompiled honest protocol's rejection probability on the actual reduced input $\rho_{\calA}=\Tr_{\calB}\ketbra{\tilde w}$, averaged over its random choices (including the implemented encoding). Let $\Gamma_i$ bound the trace distance between the output of the compiled repetition and the output of the accept/reject measurement specified in \Cref{def:gentle}. We include in each state comparison any additional quantum register entangled with the witness; the protocol leaves that register untouched.

On the original input, the gentle measurement lemma bounds the trace distance between one repetition's output and the original joint state tensored with $\ketbra{\mathsf{Accept}}$ by $\sqrt{b_i}+b_i+\Gamma_i$. To analyze all repetitions, replace them one at a time, from first to last, by operations that leave the witness unchanged and report acceptance. At the $i$th replacement the input witness is still the original one, and the subsequent operations cannot increase trace distance. The triangle inequality therefore bounds the trace distance between the final joint state of the witness, untouched register, and verdict, and the original joint witness--register state tensored with $\ketbra{\mathsf{Accept}}$, by
\[
    E:=\sum_{i=1}^{N}\bigl(\sqrt{b_i}+b_i+\Gamma_i\bigr).
\]
Here the output state is averaged over random choices and measurement outcomes, with the transcript and auxiliary workspace discarded. Section~\ref{par:compiler-accuracy} gives $\Gamma_i\leq2^{-h}$, where $h$ is the compiler's accuracy parameter. Choosing $h\geq\secp+\lceil\log_2N\rceil$ thus gives $\sum_i\Gamma_i\leq2^{-\secp}$. Here $N={\sf poly}(\secp,|x|)={\sf poly}(\secp)$ under the instance-length convention, and $b_i\leq\beta(\secp)$ for every $i$. Hence
\[
    E\leq N\bigl(\sqrt{\beta(\secp)}+\beta(\secp)\bigr)+2^{-\secp}
    =\negl.
\]
Completeness is at least $1-E$. For a pure joint input, the expected squared trace distance of the conditional residual joint state from the input is at most $E$, by the pure-state infidelity bound. Markov's inequality therefore bounds the probability of distance greater than $\tau$ by $E/\tau^2$; taking $\tau=E^{1/3}$ gives negligible distance except with negligible probability.

Let $M$ be the original acceptance operator and amplify with completeness cutoff $t=(a+b)/2$, soundness cutoff $u=(a+3b)/4$, and error $\epsilon_{\rm MW}=2^{-\secp}$. The amplified acceptance operator $A=f(M)$ satisfies
\[
    A\preceq\epsilon_{\rm MW}I+\Pi_{>u},
    \qquad M\succeq u\Pi_{>u}.
\]
Hence every extracted state $\rho$ on the original witness register satisfies
\[
    \Tr(M\rho)\geq u\bigl(\Tr(A\rho)-\epsilon_{\rm MW}\bigr).
\]
An extraction coefficient $\eta_{\rm MW}$ for the amplified relation therefore gives coefficient $u\eta_{\rm MW}$ for the original relation, with negligible additive error (also after averaging over extraction outcomes and failures). Since $u\geq(a-b)/4$, this remains inverse polynomial under $n\leq q(\secp)$. For extraction through the real encoding, the preceding decoding argument uses this inequality and its complex conjugate, giving the guarantee for $\calR_{\pm}$.

\begin{corollary}[Single-copy completeness for arbitrary valid witnesses]
\label{cor:single-copy-arbitrary-witnesses}
Let $V_\QMA$ verify $\calL\in\QMA_{a,b}$, where $0\leq b<a\leq1$ and $a-b\geq1/{\sf poly}(\secp,|x|)$. Assuming the post-quantum hardness of LWE with superpolynomial modulus-to-noise ratio, applying \Cref{thm:wit-pres-cvqc} to a Marriott--Watrous amplified verifier yields a single-copy $\CVQC$ with negligible computational soundness error such that, for every $x\in\calL_\yes$,
\[
    \Pr[\mathrm{CVQC\ accepts}\mid\ket{w}]
    \geq \frac{a-b}{2-a-b}-\negl
    \quad\text{whenever}\quad
    \Pr[V_{\QMA,x}(1^\secp,\ket{w})=1]\geq a.
\]
\end{corollary}
\begin{proof}
Fix $x$ and let $0\leq M\leq I$ be the acceptance operator of $V_{\QMA,x}$. Set $t=(a+b)/2$, let $\Pi$ project onto the eigenspaces of $M$ with eigenvalues at least $t$, and write $q=\|\Pi\ket{w}\|^2$. Since
\[
    a\leq\bra{w}M\ket{w}\leq q+t(1-q),
    \qquad
    q\geq\frac{a-t}{1-t}=\frac{a-b}{2-a-b},
\]
a valid witness has at least this much weight in $\operatorname{im}\Pi$.

Apply Marriott--Watrous amplification~\cite{CC:MW05} with completeness threshold $t$ and soundness threshold $u=(a+3b)/4$. The gap $t-u=(a-b)/4$ is inverse polynomial, so the amplified verifier is efficient and uses the same witness register. Its acceptance operator is a function of $M$ that is at least $1-2^{-\secp}$ on every eigenspace with eigenvalue at least $t$, and at most $2^{-\secp}$ on those with eigenvalue at most $u$. Thus every normalized state in $\operatorname{im}\Pi$ is a high-quality witness for this amplified verifier, whose soundness error on no-instances is negligible. Applying \Cref{thm:wit-pres-cvqc} gives a single-copy $\CVQC$ whose honest-execution acceptance effect $0\leq F\leq I$ satisfies $\bra{g}F\ket{g}\geq1-\delta$ for every normalized $\ket{g}\in\operatorname{im}\Pi$, for a negligible $\delta$.

To account for coherent cross terms, write $\ket{w}=\sqrt q\ket{g}+\sqrt{1-q}\ket{h}$, where $\ket{g}\in\operatorname{im}\Pi$ and $\ket{h}\perp\operatorname{im}\Pi$ are normalized (omitting a term of zero weight). Orthogonality and Cauchy--Schwarz for the positive operator $I-F$ imply
\[
    |\bra{g}F\ket{h}|
    =|\bra{g}(I-F)\ket{h}|
    \leq\sqrt{\bra{g}(I-F)\ket{g}\,\bra{h}(I-F)\ket{h}}
    \leq\sqrt\delta.
\]
Consequently, using $F\geq0$,
\[
    \bra{w}F\ket{w}
    \geq q(1-\delta)-2\sqrt{q(1-q)\delta}
    \geq q-\delta-\sqrt\delta
    \geq\frac{a-b}{2-a-b}-\negl.
\]
The spectral decomposition is only for analysis; the prover runs the compiled protocol directly on its single copy of $\ket{w}$.
\end{proof}

\subsection{Sequential Soundness Amplification}\label{sec:sequential-rep}

In this section, we show that sequentially repeating a testably-sound $\CVQC$ yields a $\CVQC$ with overwhelming soundness.
Specifically,  we consider a compiler consisting of many sequential repetitions of the underlying $\CVQC$ protocol, and the verifier accepts if the number of accepting executions exceeds a threshold.

For the protocol obtained from the gentle compiler, each repetition generates new compiler keys independently of all previous repetitions. Section~\ref{sec:ma-natarajan} bounds the probability of rejection and the change to the witness after all repetitions. We analyze soundness below.

\begin{construction}[Sequential Amplification]\label{constr:sequential-repetition}
    Fix a language $\calL=(\calL_\yes,\calL_\no)\in\QMA_{a,b}$ with a $\QMA$ verifier $V_\QMA=\{V_{\QMA,x}\}_{x\in\calL_\yes\cup \calL_\no}$.
    Let $(P,V)$ be a $\CVQC$ for $\calL$ w.r.t.\ $V_\QMA$, where $V=(V_\test,V_{\sf check})$, with check-round completeness $c$ and $(\eta,s)$-testable soundness. Set $\Delta:=c-s$ and assume that $\eta=\eta(\secp,|x|)$ and $\Delta=\Delta(\secp,|x|)$ are both at least $1/{\sf poly}(\secp,|x|)$.

    All parameters used by the protocol are fixed computable functions of $(\secp,n)$, independently of the adversary. When a compiler introduces negligible losses into fixed bounds $c_0,s_0$, use
    \[
        c_*:=c_0-\frac{c_0-s_0}{4},\qquad
        s_*:=s_0+\frac{c_0-s_0}{4},
    \]
    and write $c=c_*$, $s=s_*$ and $\Delta=c_*-s_*$ below. These margins absorb the losses for sufficiently large parameters, uniformly over allowed instances and advice for each efficient prover; no negligible function is used to compute the threshold or $N$. For the additional knowledge conclusion, let $(d_*,s_*,\delta_*)$ be fixed conservative parameters for the protocol actually repeated. Before setting $N$, replace the test tolerance by $\eta\leftarrow\min\{\eta,1/(n^{d_*}\secp^{d_*})\}$. One may take $\delta_*=\delta/2$ for sufficiently large parameters when the base extraction bound $\delta$ is uniformly inverse polynomial. The main gentle compilation gives $d_*=d+1$; the additional Appendix~D composition gives $d_*=d+2$ where applicable.

    For $N$ large enough (that we set below), we  construct a $\CVQC$ $(P',V')$ for $\calL$ w.r.t.\ $V_\QMA$ with negligible soundness error.  $P'$ and $V'$ take as input $(1^\secp,x)$ where $x\in\calL_\yes$, and  in $P'$ takes as an additional input a witness state.\footnote{This witness state may require special form, depending on the $\CVQC$ being repeated.}  The protocol $(P'(\ket{w_0}),V')(1^\secp,x)$ proceeds as follows:

    \begin{enumerate}
    \item Set
    \[
        N=\left\lceil128\secp\left(\frac{1}{\eta\Delta}+\frac{1}{\Delta^2}\right)\right\rceil.
    \]
        \item For $i=1,\dots, N$, do the following:

        Emulate the protocol  $\bigg(P\big(\ket{w_{i-1}}\big), V\bigg)(1^\secpar, x)$, where $V$ implements  $V_\test$ with probability $1/2$ and implements $V_{\sf check}$ with probability $1/2$.
        \begin{enumerate}
        \item
        If $V$ implemented $V_\test$ and the test failed (i.e., $V_\test$ rejects), then $V'$ aborts and outputs $\Reject$.

        \item  If $V$ implemented  $V_{\sf check}$ then denote by $b_i$ the output bit of $V$ ($1$ denotes accept, $0$ denotes reject).
        \item Denote by $\ket{w_{i}}$ the residual state of $P$.

         \end{enumerate} \item Denote by $B\subseteq[N]$ all the repetitions where $V$ executed $V_{\sf check}$.

        \item $V'$ outputs $\Accept$ if $\sum_{i\in B} b_i \geq  \left(\frac{c+s}{2}\right) |B|$. Otherwise, $V'$ outputs $\Reject$.

        \item $P'$ outputs $\ket{w_{N}}$.
    \end{enumerate}

\end{construction}

\begin{theorem}[Sequential Soundness Amplification]\label{thm:sequential-soundness}
    Let $(P, V)$ be a $\CVQC$ with check-round completeness $c$ and $(\eta,s)$-testable soundness, where $\eta$ and $\Delta:=c-s$ are both at least $1/{\sf poly}(\secp,|x|)$.

    Then the result of compiling $(P, V)$ to $(P', V')$ via \Cref{constr:sequential-repetition} has $\negl$ soundness error.
\end{theorem}
\begin{proof}
    Below, we assume that $\secp<n=|x|\leq q(\secp)$, where the fixed polynomial $q$ includes the length after any lower-bound padding.

    \paragraph{Argument of Knowledge.} Fix any allowed instance $x\in\lang_\yes\cup\lang_\no$ and any efficient prover $P^*$ with an arbitrary initial advice state $\rho^*$. Write $p_{\rm acc}$ for its acceptance probability, without assuming that this probability is noticeable.

    For the additional knowledge conclusion, assume that the protocol actually repeated has the fixed $(d_*,s_*,\delta_*)$ testable knowledge guarantee for $\calR$ specified above. Our extractor picks $j\gets[N]$, simulates the first $j-1$ repetitions, and invokes the underlying extractor on the residual state in repetition $j$ (outputting failure after any prior abort). It has coherent access to the prover's round unitaries and their inverses, with one initialized private advice state. The reduction retains the verifier's coins, trapdoor, and workspace and implements its wrapper reversibly, without cloning or reinitializing the unknown advice. The transcript bounds below require no knowledge assumption.

   Formally, we will first prove that for accepting transcripts, nearly all prefixes of the transcript satisfy the following: the prover, conditioned on this prefix, will pass $V_{\test}$ with high probability in the upcoming round. We call such a transcript a $\mathsf{Good}$ transcript.

\begin{definition}[Good Transcripts]
    Set $K=\lfloor\Delta N/24\rfloor$. We say that a transcript $\tau = \tau_\secp \in \mathsf{Good}_\secp$ if (1) $\tau$ is accepted by the verifier and (2) there exists a large subset of repetitions $S_\tau \subseteq [N], |S_\tau| \geq N-K$ such that for all $i \in S_\tau$,
    \[\Pr\left[P^*_i(\rho_{i-1}), V_{i,\test}(1^\secp, x) = 1
    \big|\tau_{i-1}
    \right] \geq 1-\eta\]
    where $\rho_{i-1}$ denotes the state on the prover's registers at the end of $i-1$ repetitions with $\rho_0 = \rho^*$, $V_{i,\test}$ denotes the output of the verification algorithm of the underlying CVQC in round $i$.
    Here $\tau_{i-1}$ records the full classical history of the experiment through repetition $i-1$, including the verifier's coins and previous verdicts; the probability is over fresh randomness and measurements in repetition $i$.
\end{definition}

The next claim bounds the unconditional probability of accepting outside $\mathsf{Good}$.

\begin{claim}
\label{clm:six}
The unconditioned experiment satisfies
\[
    \Pr[\tau_\secp\text{ is accepted and }\tau_\secp\notin\mathsf{Good}_\secp]
    \leq e^{-\eta\Delta N/48}.
\]
\end{claim}

\begin{proof}
    Any accepted transcript outside $\mathsf{Good}$ has at least $K+1$ repetitions whose conditional test-acceptance probability is below $1-\eta$.
    Since the verifier picks ``test'' with probability $1/2$, each such repetition causes rejection with conditional probability at least $\eta/2$. Exposing the first $K+1$ such repetitions in order, and conditioning on reaching each without a prior rejection, gives the following bound even when these repetitions are selected adaptively:
    \[
    \begin{aligned}
        \Pr[\tau_\secp\text{ is accepted and }\tau_\secp\notin\mathsf{Good}_\secp]
        &\leq(1-\eta/2)^{K+1}\\
        &\leq e^{-\eta\Delta N/48}
        \leq e^{-8\secp/3}=\negl,
    \end{aligned}
    \]
    This proves the unconditional bound.
\end{proof}

We next bound the unconditional event that the repeated protocol accepts with fewer than $\Delta N/12$ prefixes at which the test succeeds with probability at least $1-\eta$ and the check succeeds with probability greater than $s+\Delta/4$.

To do this, we will first prove the following claim, which gives tail bounds on the number of accepts in rounds where the verifier picks $V_\honest$.

\begin{claim}
    Define a sequence of variables $X_1 \ldots X_{N}$ as follows.
    For each $i \in [N]$, denote by $X_i = 1$ the event that
    (1) the verifier picked $V_{\honest}$ in repetition $i$, and (2) the prover passed $V_\honest$ in repetition $i$ (and otherwise, $X_i = 0$). Set $X_i=0$ for every repetition after the verifier aborts.
    Then for every $\gamma \in (0,1)$,
    \begin{equation}
    \label{eq:d2}
    \Pr\left[ \sum_{i \in [N]}X_i - \sum_{i \in [N]}\mathbb{E}[X_i|\tau_{i-1}] \geq \gamma N \right] \leq \exp( - \gamma^2 N/2 )
    \end{equation}
\end{claim}
\begin{proof}
        For any $i$ and   transcript $\tau$, define $Y_i = X_i - \mathbb{E}[X_i|\tau_{i-1}]$, where $\tau_{i-1}$ denotes a truncation of the transcript upto round $i-1$. We have that for every $i \in [N]$ and every $\tau_{i-1}$,
        \[\mathbb{E}[ Y_i |  \tau_{i-1}] = \mathbb{E}[ X_i - \mathbb{E}[X_i|\tau_{i-1}] | \tau_{i-1}] = 0\]
        Moreover, for every $i \leq N$ and  transcript $\tau$, $Y_i \in [-1,1]$.

Then $S_0 = 0, S_k = \sum_{i=1}^{k} Y_i$
        is a martingale with respect to the transcript history because $\mathbb{E}[S_{k+1}|\tau_k] = S_k$, and it
        satisfies $|S_k - S_{k-1}| \leq 1$.

        Then by Azuma-Hoeffding (Theorem \ref{thm:azuma}) applied to Martingales with distance $c_k = 1$, we have
        \[\Pr[ S_N - S_0 \geq t ] \leq \exp( - t^2/(2N) ).\]

Let $t = \gamma N$. Then,
\[\Pr[S_N \geq \gamma N]
    \leq \exp( - (\gamma N)^2 / (2N) )
    = \exp( - \gamma^2 N/2 )\]
which is the same as saying
\[
\Pr\left[ \sum_{i \in [N]}X_i - \sum_{i \in [N]}\mathbb{E}[X_i|\tau_{i-1}] \geq \gamma N \right] \leq \exp( - \gamma^2 N/2 )
\]
This completes the proof.
\end{proof}

Equation (\ref{eq:d2}) above proves that for nearly all transcripts, the total number of accepts in rounds where the verifier picks $V_\honest$ deviates from its expected value by only a small amount.
We will now use the definition of $\mathsf{Good}$ to show that there are many rounds where both conditions hold. For a prefix $\tau_{i-1}$ reached before abort, let $p_{T,i}$ and $p_{C,i}$ denote the conditional acceptance probabilities of $V_{i,\test}$ and $V_{i,\honest}$, respectively, on the residual prover state. Call a round qualifying if
\begin{equation}\label{eq:d1}
    p_{T,i}\geq1-\eta
    \qquad\text{and}\qquad
    p_{C,i}>s+\Delta/4.
\end{equation}
These conditions depend only on the history before round $i$. On every full execution, let $J=J(\tau)$ count qualifying rounds, with no qualifying rounds after an abort. Under the additional knowledge assumption above, the underlying extractor succeeds at each qualifying prefix with probability at least $\delta_*$.

\begin{claim}
\label{clm:eight}
   In the unconditioned experiment,
   \[
       \Pr[\tau\in\mathsf{Good}\text{ and }J<\Delta N/12]
       \leq e^{-\Delta^2N/2048}+e^{-\Delta^2N/512}.
   \]
\end{claim}

\begin{proof}
On a transcript $\tau\in\mathsf{Good}$, at most $K$ prefixes fail the test condition. Every other nonqualifying round has $p_{C,i}\leq s+\Delta/4$, and hence $\mathbb{E}[X_i\mid\tau_{i-1}]\leq(s+\Delta/4)/2$. For the at most $K+J$ remaining rounds, this conditional expectation is at most $1/2$. Therefore
\[
\sum_{i\in[N]}\mathbb{E}[X_i\mid\tau_{i-1}]
\leq \frac{(s+\Delta/4)N}{2}+\frac{K+J}{2}.
\]

Equation (\ref{eq:d2}), with $\gamma=\Delta/32$, proves that except with probability $\exp(-\Delta^2N/2048)$ over the unconditioned experiment,
\[
\sum_{i\in[N]}X_i
\leq\sum_{i\in[N]}\mathbb{E}[X_i\mid\tau_{i-1}]+\frac{\Delta N}{32}.
\]
Using $K\leq\Delta N/24$, every $\mathsf{Good}$ transcript outside this exceptional event satisfies
\[
\begin{aligned}
\sum_{i\in[N]}X_i
&\leq\frac{sN}{2}+\frac{\Delta N}{8}+\frac{\Delta N}{48}+\frac{\Delta N}{32}+\frac J2\\
&=\frac{sN}{2}+\frac{17\Delta N}{96}+\frac J2.
\end{aligned}
\]

Let $n_\honest$ denote the number of check choices among the $N$ independent fair branch choices, extending the choices after an abort by fresh independent fair bits for this analysis. On accepting transcripts, $n_\honest=|B|$. By a Chernoff--Hoeffding bound,
\[
\Pr\left[n_\honest<\frac N2-\frac{\Delta N}{32}\right]
\leq\exp(-\Delta^2N/512).
\]
Recall that $N\geq128\secp/\Delta^2$, so both exceptional probabilities are negligible in $\secp$, uniformly over all allowed instances and initial states.

On an accepting transcript outside this second exceptional event, since $s+\Delta/2=(c+s)/2\leq1$, the acceptance threshold gives
\[
\begin{aligned}
\sum_{i\in[N]}X_i
&\geq\left(\frac{c+s}{2}\right)n_\honest\\
&\geq\left(s+\frac\Delta2\right)\left(\frac N2-\frac{\Delta N}{32}\right)
\geq\frac{sN}{2}+\frac{21\Delta N}{96}.
\end{aligned}
\]
Combining the bounds gives $J\geq\Delta N/12$. Thus every $\mathsf{Good}$ transcript with $J<\Delta N/12$ lies in one of the two exceptional events. A union bound proves the claim.
\end{proof}

Combining Claims~\ref{clm:six} and~\ref{clm:eight}, the unconditional probability of acceptance with $J<\Delta N/12$ is at most
\[
    \nu_{\rm stat}(\secp):=e^{-8\secp/3}+e^{-\secp/16}+e^{-\secp/4}.
\]
Thus $J\geq\Delta N/12$ with probability at least $p_{\rm acc}-\nu_{\rm stat}(\secp)$. Since the extractor's simulation has the ordinary, unconditioned prefix distribution, its uniformly chosen round qualifies with probability
\[
    \frac{\mathbb{E}[J]}{N}
    \geq\frac{\Delta}{12}\bigl(p_{\rm acc}-\nu_{\rm stat}(\secp)\bigr).
\]
This is an average over unconditioned prefixes: the extractor does not condition on future acceptance, identify qualifying prefixes, retry, or obtain another copy of the advice state. The choice $\eta\leq1/(n^{d_*}\secp^{d_*})$ ensures that every qualifying prefix meets the underlying extractor's test threshold, and its check probability exceeds $s_*=s$.

The knowledge conclusion assumes that the underlying protocol satisfies \Cref{def:2partpok} for the same relation $\calR$ named in that conclusion. For our concrete construction, the decoding argument in \Cref{sec:ma-natarajan} establishes this premise for $\calR=\calR_{\pm}$, with $\delta_*=u/4$ after compilation; thus extraction recovers a witness up to complex conjugation and reports its orientation. With fixed parameters $(d_*,s_*,\delta_*)$ for the protocol actually repeated, the preceding random-round extraction argument gives
\[
    p_{\rm ext}\geq\frac{\delta_*\Delta}{12}
    \bigl(p_{\rm acc}-\nu_{\rm stat}(\secp)\bigr).
\]
A negative right-hand side is harmless. Under $n\leq q(\secp)$, take a uniform lower bound $1/p(\secp)\leq\delta_*\Delta/12$ for a fixed polynomial $p$. Consequently, under this premise, the repeated protocol satisfies \Cref{def:pok} for $\calR$, with this inverse-polynomial extraction coefficient and negligible additive error, uniform over allowed instances and advice for each efficient prover. The fixed compiler margins above absorb their negligible losses for sufficiently large parameters; any remaining finite-parameter exceptions are included in the additive negligible bound.

\begin{remark}\label{remark:Nloss}
The extraction coefficient depends on the gap $\Delta$ and the base guarantee $\delta_*$, but not on the repetition count $N$. The counting argument also covers provers that use high-quality witnesses in some executions and junk states in others; the extractor need not identify which executions qualify.
\end{remark}

        \paragraph{Soundness.}
        Fix any $x \in \lang_\no$ and any prover $P^*$ that generates accepting proofs for $x$ with probability $\frac{1}{p_{\mathsf{acc}}(\secp)}$, for some polynomial $p_{\mathsf{acc}}$.
        Towards a contradiction, for such a prover, note that by Claim \ref{clm:six} and Claim \ref{clm:eight} above, there exists $j \in [N]$, transcript prefix $\tau_{j-1}$ and corresponding prover state $\rho_{j-1}$ such that:
    \[\Pr\left[P^*_{j}(\rho_{j-1}), V_{j,\test,\NDCVQC}(1^\secp, x) = (\cdot,1)
    \big|\tau_{j-1}
    \right] \geq 1-\eta\]
    and
   \[ \Pr\left[P^*_{j}(\rho_{j-1}), V_{j,\honest,\NDCVQC}(1^\secp, x) = (\cdot,1)|\tau_{j-1}
        \right] > s + \Delta/4
    \]
        This circuit $P^*_j$ with input state $\rho_{j-1}$ contradicts $(\eta,s)$-testable soundness of the underlying CVQC, as desired.
\end{proof}

\section*{Acknowledgments}
\addcontentsline{toc}{section}{Acknowledgments}
DK was supported in part by AFOSR, NSF CAREER CNS-2238718, NSF 2112890, NSF CNS-2247727 and a Google
Research Scholar award. This material is based upon work supported by the Air Force Office of
Scientific Research under award number FA9550-23-1-0543.
This work was partially done while the authors were visiting the Simons Institute for the Theory of Computing, Berkeley.

\bibliographystyle{alpha}
\phantomsection
\addcontentsline{toc}{section}{References}
\bibliography{references}
\appendix
\section{Completeness-Soundness Gaps for Existing CVQCs with One Witness Copy}\label{app:c-s-gaps}

In this section, we briefly describe the soundness and completeness of common approaches to CVQC with one witness copy.
Most existing approaches are not stated directly in terms of CVQC for QMA. Many of them can be adapted trivially, but not all.

Currently, the approaches can be broadly divided into three camps: those based on Mahadev's measurement protocol~\cite{Mahadev18}, those based on compiling non-local games~\cite{KLVY23,NZ23FOCS}, and those based on remote state preparation~\cite{FOCS:GheVid19,FOCS:Zhang22,ICALP:GheMetPor23,ITCS:Zhang25}.
The former two generalize naturally to QMA, but incur an inverse polynomial completeness-soundness gap. The third approach is based on techniques that do not generalize to QMA verification.

\paragraph{Mahadev-Style Measurement Protocols.}
Mahadev~\cite{Mahadev18} gave a way to weakly ``commit'' to a quantum state with a classical message, then measure it verifiably.
Verification of the measurement is ensured via randomly either running the measurement or testing the commitment. The guarantee is that
The verifier then either tests the commitment, or asks the prover to perform some measurement, with the guarantee that if the prover would pass the test round with high probability, then the measurement corresponds to measuring a real quantum state (i.e. testable soundness,~\Cref{eqn:test-sound}). To argue standard soundness, she first needs to establish that a prover who passes the CVQC non-trivially must pass test rounds with high probability. Without this guarantee, a cheating prover can convince the verifier with \emph{higher} probability than an honest prover using a valid witness.
Mahadev ensures high test round success by repeating her protocol many times, which requires the same state to be committed to many times and each one measured. With one repetition, Mahadev's measurement protocol has worse completeness than soundness.

To achieve CVQC, she uses her measurement protocol to implement \cite{MF16}'s post-hoc verification of the local Hamiltonian problem. This procedure measures an arbitrary quantum state to determine whether it is a ground state of the local Hamiltonian problem. Since the local Hamiltonian problem is QMA complete, this naturally generalizes to QMA.

However, the use of the local Hamiltonian problem also creates a small completeness-soundness gap.
The best-known circuit-to-Hamiltonian conversions incur a completeness-soundness gap which scales inverse polynomially in the size of the circuit~\cite{Quantum:BC18}.\footnote{One approach to overcoming this gap is to positively resolve the quantum PCP conjecture.}
This gap is inherited by Mahadev's protocol.

Recent advances in frustration-free ZX Hamiltonians~\cite{MN25} can be combined with her measurement protocol to achieve completeness close to $1$. However, the soundness error is quite high at $1-1/\poly[|x|]$, which still leaves a $1/\poly[|x|]$ gap.

\newcommand{\CHSH}{\mathsf{CHSH}}
\paragraph{Compiling Non-Local Games.}
Natarajan and Zhang~\cite{NZ23FOCS} showed a non-local game which, when compiled using \cite{KLVY23}'s cryptographic non-local to local compiler, gives a $\CVQC$. The non-local game in question is based verifying that the prover's state is a ground state for some local Hamiltonian. This naturally allows it to generalize to $\QMA$, but saddles it with the same gap shrinkage from a circuit-to-Hamiltonian conversion as Mahadev's approach.

We mention that their protocol has a similar testable soundness to Mahadev's, but ensures the correctness-soundness gap is positive for a single iteration by carefully tuning the probability with which a test round is run.

\paragraph{Remote State Preparation.}
This class of protocols based on replacing quantum communication in delegation protocols (e.g. ~\cite{PhysRevA.96.012303,PhysRevA.98.022323}) with remote state preparation~\cite{FOCS:GheVid19,AC:CCKW19}.
The underlying quantum-communication delegation protocols are typically in the measurement-based quantum computation (MBQC) model, which has the same computational power as the circuit model. It works by preparing an entangled resource state starting from $\ket{+}$ states, then performing a series of adaptive measurements on the entangled resource state. If an chosen input is desired, it must also be entangled together with these $\ket{+}$ states to create the resource state.

The quantum-communication delegation protocols~\cite{PhysRevA.96.012303,PhysRevA.98.022323} permit completeness 1 verification of any MBQC computation,\footnote{This completeness is distinct from the completeness of the circuit being delegated. For example, if the circuit would output 1 with probability $2/3$, the delegated version would too.} but require that the verifier sends masked $\ket{+}$ states to the prover for the prover to entangle. If a chosen input $\ket{w}$ is desired, \cite{PhysRevA.96.012303} explicitly requires that the verifier mask and and send $\ket{w}$, while \cite{PhysRevA.98.022323} does not allow chosen inputs at all to simplify the family of states that the verifier must send.

In the case of QMA, the prover knows the witness, but not the verifier. This prevents the underlying quantum-communication delegation protocols from generalizing to QMA.
Moreover, a classical verifier could not even hold the witness to encode it, let alone send it.
\section{Proofs of No-Intrusion}\label{app:PoNI}

As an application of state-preserving arguments for NP, we show how to construct public-key encryption (PKE) with \emph{proofs of no-intrusion} (PoNIs). PKE with PoNIs was introduced by \cite{eprint:GR25}. Roughly, it allows a classical verifier to test that a quantum ciphertext has not been stolen from a server, in the sense that nobody besides the server can decrypt it even given the key, without destroying the server's copy of the ciphertext.

Our construction achieves a stronger notion of verifiability than \cite{eprint:GR25}. In their construction, only the original encryptor can verify their own ciphertext because the $\PoNI$ verifier requires secret information about it. In our construction, the encryptor may publish part of their verification key to allow \emph{anyone} to act as the $\PoNI$ verifier.

\subsection{Definition}

We adapt \cite{eprint:GR25}'s definition to allow dividing the verification $\verkey$ into public and secret components. The $\PoNI$ only requires the verifier to use the public component, so anyone may audit the prover.

\begin{definition}[Encryption with Proofs of No-Intrusion]
    A public-key \textbf{encryption scheme with a public proof of no-intrusion} consists of the following QPT algorithms.
    \begin{itemize}
        \item $(\pk, \sk) \gets \KeyGen(1^\secpar)$ takes as input the security parameter $1^\secpar$ then outputs a public key $\pk$ and a secret key $\sk$.
        \item $\verkey \gets \VerKeyGen(1^\secpar)$ takes as input the security parameter $1^\secpar$ then outputs a verification key $\verkey = (\verkey_p, \verkey_s)$ consisting of a public component $\verkey_p$ and a secret component $\verkey_s$.
        \item $\ct \gets \Enc(\pk, m, \verkey_s)$ takes as input a public key $\pk$, a message $m$, and a verification key secret component $\verkey_s$, then outputs a ciphertext $\ct$.
        \item $m\gets \Dec(\sk, \ct)$ takes as input a secret key $\sk$ and a ciphertext $\ct$, then outputs a message $m$.
    \end{itemize}

    Additionally, it is equipped with a \textbf{proof of no-intrusion}, which is an interactive protocol \emph{with classical communication} between a QPT prover $P$ holding a state $\ket{\psi}$ and a PPT verifier holding a verification key $\verkey$:
    \[
        (\rho', b)
        \gets
        \PoNI\langle \prover(\ket{\psi}), \verifier(\verkey)\rangle
    \]
    At the end of the protocol, $P$ outputs a residual state $\rho'$ and $V$ outputs a decision bit $b\in \{\Accept, \Reject\}$.

    The scheme must satisfy the following properties:
    \begin{itemize}
        \item \textbf{Correctness.} For every message $m$, every key pair $(\pk, \sk)$ in the support of $\KeyGen(1^\secpar)$, and every verification key $\verkey$ in the support of $\VerKeyGen(1^\secpar)$,
        \[
            \Pr\left[m = \Dec(\sk, \Enc(\pk, m, \verkey_s)\right] = 1-\negl
        \]

        \item \textbf{Semantic Security.} For every pair of messages $(m_0, m_1)$,
        \[
            \left\{
                (\ct_0, \verkey_p) :
                \begin{array}{c}
                    (\pk, \sk) \gets \KeyGen(1^\secpar) \\
                    (\verkey_p, \verkey_s) \gets \VerKeyGen(1^\secpar) \\
                     \ct_0 \gets \Enc(\pk, m_0, \verkey)
                \end{array}
            \right\}
            \approx_c
            \left\{
                (\ct_1, \verkey_p) :
                \begin{array}{c}
                (\pk, \sk) \gets \KeyGen(1^\secpar) \\
                (\verkey_p, \verkey_s)\gets \VerKeyGen(1^\secpar) \\
                 \ct_1 \gets \Enc(\pk, m_1, \verkey)
                \end{array}
            \right\}
        \]

        \item \textbf{PoNI Correctness and State Preservation:} There exist negligible functions $\mu_1,\mu_2$ such that, for every message $m$, sample $(\pk,\sk)\gets\KeyGen(1^\secpar)$ and $(\verkey_p,\verkey_s)\gets\VerKeyGen(1^\secpar)$, and let $\rho_{\ct}\gets\Enc(\pk,m,\verkey_s)$ be an honestly generated ciphertext. For
        \[
            (\rho',b)\gets\PoNI\langle P(\rho_{\ct}),V(\verkey_p)\rangle,
        \]
        it holds that
        \[
            \Pr\left[b=\Accept\ \land\ \tfrac12\|\rho'-\rho_{\ct}\|_1\leq\mu_1(\secpar)\right]\geq1-\mu_2(\secpar).
        \]
        Here $\rho'$ is the residual ciphertext state conditioned on the protocol randomness and outcomes, with auxiliary registers discarded. The probability includes key generation, encryption, and audit randomness. Together with decryption correctness, this guarantees successful decryption after an honest audit except with negligible probability.

        \item \textbf{PoNI Security:} The construction satisfies either search PoNI security (\Cref{def:PoNI-security-search}) or decisional PoNI security (\Cref{def:PoNI-security-decision}).
    \end{itemize}
\end{definition}

For security, we consider a general scenario where the server receives a ciphertext, then splits it into two pieces after interacting in many PoNIs. After the split, the server simultaneously gives a PoNI using one piece while the ``hacker'' receives the secret key and attempts to decrypt the ciphertext using the other piece.

\begin{definition}[Public Proofs of No-Intrusion for Encryption: Search Security]\label{def:PoNI-security-search}
    Consider the following security game $\PoNIEncSearch(\adv)$, played by an adversary $\adv = (\adv_1, \adv_P, \adv_H)$ consisting of three QPT algorithms (with auxiliary quantum input).
    \begin{enumerate}
        \item Sample a message $m\gets \{0,1\}^\secpar$, a key pair $(\pk, \sk)\gets \KeyGen(1^\secpar)$, a verification key $\verkey \gets \VerKeyGen(1^\secpar)$, and a ciphertext $\ct \gets \Enc(\pk, m, \verkey_s)$.
        \item Initialize $\adv_1$ with $(\pk, \ct, \verkey_p)$. Note that access to $\verkey_p$ allows them to run $\PoNI$ internally any number of times.
        \item $\adv_1$ outputs two registers $(\calH, \calP)$.
        \item Perform $b \gets \PoNI\langle \adv_P(\calP), \verifier(\verkey)\rangle$, where $b$ is the verifier's decision bit.
        \item Run $m' \gets \adv_H(\regH, \sk, \verkey)$.
        \item Output $1$ (the adversary wins) if $b = \Accept$ and $m' = m$. Otherwise output $0$ (the adversary loses).
    \end{enumerate}

    \noindent We say the public PoNI has search security if for all QPT adversaries $\adv$,
    \[
        \Pr[1\gets \PoNIEncSearch(\adv)] = \negl
    \]

    \noindent If this holds for QPT $\adv_1$ and $\adv_P$, but \emph{unbounded} $\adv_H$, we say that the PoNI has \emph{everlasting} search security.
\end{definition}

\begin{definition}[Public Proofs of No-Intrusion for Encryption: Decisional Security]\label{def:PoNI-security-decision}
    Consider the following security game $\PoNISec(m)$, played by an adversary $\adv(\ket{\psi})$ and parameterized by a message bit $m\in \{0,1\}$.
    \begin{enumerate}
        \item Sample a key pair $(\pk, \sk)\gets \KeyGen(1^\secpar)$, a verification key $\verkey\gets \VerKeyGen(1^\secpar)$ and a ciphertext $\ct \gets \Enc(\pk, m, \verkey)$.
        \item Initialize $\adv(\ket{\psi})$ with $(\pk, \ct, \verkey_p)$ Note that access to $\verkey_p$ allows them to run $\PoNI$ internally any number of times.
        \item $\adv$ outputs a register $\regR_{\Dec}$.
        \item Perform $\PoNI\langle \adv(\cdots), \verifier(\verkey)\rangle$. If the prover rejects, output $\bot$. Otherwise output $(\regR_{\Dec}, \sk, \verkey)$.
    \end{enumerate}

    \noindent We say the public PoNI has decisional security if for all QPT adversaries $\adv(\ket{\psi})$,
    \[
        \{\PoNISec(0)\} \approx \{\PoNISec(1)\}
    \]

    \noindent If these two distributions are instead \emph{statistically} close, we say that the PoNI is everlasting secure.
\end{definition}

Goyal and Raizes show that search $\PoNI$ security can be generically compiled to decision $\PoNI$ security in the quantum random oracle model (QROM).

\begin{theorem}\label{thm:enc-poni-decisional}
    Any $\PKE$ with search $\PoNI$ security can be generically compiled into a $\PKE$ with decisional $\PoNI$ security in the quantum random oracle model.
\end{theorem}

\subsection{Tools}

\paragraph{PKE with Publicly Verifiable Certified Deletion.}
The construction uses a PKE with publicly-verifiable certified deletion, which are known from a variety of assumptions including post-quantum PKE~\cite{C:BarKhuPor23,EC:BGKMRR24,TCC:BKMPW23,TCC:KitNisYam23}. Roughly, this primitive allows generation of a quantum ciphertext that can be destructively measured to produce a certificate $\cert$. The certificate can be checked by anyone using a public verification key; if the certificate is valid, then the ciphertext can no longer be decrypted, even using the secret decryption key.

\begin{definition}[PKE with Publicly Verifiable Certified Deletion]
    A PKE scheme with publicly verifiable certified deletion has the following syntax.
    \begin{itemize}
        \item $(\sk, \pk) \gets \KeyGen(1^\secpar)$ takes as input the security parameter then outputs a secret decryption key $\sk$ and public encryption key $\pk$.
        \item $(\ket{\ct}, \verkey) \gets \Enc(\pk, m)$ takes as input a public key and a message $m$ then outputs a ciphertext and an associated verification key.
        \item $m \gets \Dec(\sk, \ket{\ct})$ takes as input a secret key and a ciphertext then outputs a message.
        \item $\cert \gets \Del(\ket{\ct})$ takes as input a ciphertext then outputs a (classical) certificate.
        \item $\Accept/\Reject \gets \Ver(\verkey, \cert)$ takes as input a verification key and a certificate, then outputs accept or reject.
    \end{itemize}

    It must satisfy the standard decryption correctness and semantic security properties for PKE.

    Additionally, it must satisfy \textbf{correctness of deletion}: for any message $m$,
    \[
        \Pr\left[\Accept \gets \Ver(\verkey, \cert):
        \begin{array}{c}
             (\sk, \pk) \gets \KeyGen(1^\secpar) \\
             (\ket{\ct}, \verkey) \gets \Enc(\pk, m) \\
             \cert \gets \Del(\ket{\ct})
        \end{array}\right]
        =1-\negl
    \]

    Finally, it must also satisfy \textbf{publicly verifiable certified deletion security}: for every QPT adversary $\adv$,
    \[
        \frac{1}{2}\|\PVCDExp(0) - \PVCDExp(1) \|_1 = \negl
    \]
    where, for any two equal-length messages $m_0,m_1$, the experiment $\PVCDExp(b)$ is defined as follows (polynomial-length messages follow from bit encryption by the usual hybrid argument):
    \begin{enumerate}
        \item Sample keys $(\sk, \pk) \gets \KeyGen(1^\secpar)$ and encrypt $(\ket{\ct}, \verkey) \gets \Enc(\pk, m_b)$.
        \item Run $\adv(\ket{\ct}, \pk, \verkey)$ and parse their output as a certificate $\cert$ and a state on register $\calR_{\adv}$.
        \item If $\Ver(\verkey, \cert) = \Accept$, output $(\calR_{\adv},\pk,\sk)$. Otherwise output $\bot$.
    \end{enumerate}
\end{definition}
Releasing $\sk$ after acceptance does not strengthen the statistical deletion requirement: the adversary may retain $\pk$, and an unbounded distinguisher can sample $\sk$ from the key-generation distribution conditioned on $\pk$. Encryption and the pre-deletion adversary depend on the key pair only through $\pk$, so this reproduces the joint distribution with the actual secret key.

\paragraph{Simultaneous Extraction.}
Additionally, the analysis requires the following lemma from \cite{eprint:GR25}. Informally, it states that if two adversaries $B$ and $C$ simultaneously succeed in a search task with noticeable probability without communicating, then they also simultaneously succeed with noticeable probability even when $C$'s challenge is modified in a way which is indistinguishable to $C$, but which may be distinguishable to $B$.

\begin{lemma}\label{lem:simult-search}
    Let $\calD$ be a distribution outputting three registers $(\regA, \regB, \regC)$ along with a classical string $\cAuxClass$. Let $\calD'$ be a distribution taking as input $\cAuxClass$ then outputting $\regC'$ such that
    \[
        \{(\regA, \regC, \cAuxClass): (\regA, \regB, \regC, \cAuxClass)\gets \calD\}
        \approx_c
        \left\{(\regA, \regC', \cAuxClass):
            \begin{array}{c}
             (\regA, \regB, \regC, \cAuxClass)\gets \calD \\
              \regC' \gets \calD'(\cAuxClass)
            \end{array}
            \right\}
    \]
    For all QPT algorithms $(A, C)$ and all (unbounded) quantum algorithms $B$, if
    \[
        \Pr\left[1\gets B(\regB, \regA_B) \land 1\gets C(\regC, \cAuxClass, \regA_C):
        \begin{array}{c}
             (\regA, \regB, \regC, \cAuxClass)\gets \calD \\
             (\regA_B, \regA_C) \gets A(\regA)
        \end{array}
        \right]
        \geq 1/p
    \]
    for some $p = \poly$, then there exists $q = \poly$ such that
    \[
        \Pr\left[
        1\gets B(\regB, \regA_B) \land 1\gets C(\regC', \cAuxClass, \regA_C):
        \begin{array}{c}
             (\regA, \regB, \regC, \cAuxClass)\gets \calD \\
             (\regA_B, \regA_C) \gets A(\regA) \\
             \regC' \gets \calD'(\cAuxClass)
        \end{array}
        \right]
        \geq 1/q
    \]
\end{lemma}

\subsection{Construction}

We construct a PKE with \emph{search} $\PoNI$ security. As mentioned previously, any such scheme can be generically transformed into one with decisional $\PoNI$ security in the quantum random oracle model. We make use of our specific construction of state-preserving arguments for NP to enable compatibility with \Cref{lem:simult-search}.\footnote{Roughly, \Cref{lem:simult-search} allows reasoning about the simultaneous success rate of two searchers $B$ and $C$ when $C$'s input is changed in a way that is indistinguishable to $C$ (but may be distinguishable to $B$). It does \emph{not} allow wholesale replacement of $C$ by some other algorithm $\Ext(C)$, which might in general behave very differently. In our construction of state-preserving arguments for NP, $\Ext$ is a straightline extractor which works by sending indistinguishable messages for the first part of the protocol.}
The construction uses the following tools:
\begin{itemize}
    \item A PKE scheme with certified deletion $\PKE.(\KeyGen, \Enc, \Dec, \Del,\Ver)$.
    \item A signature scheme $\Sig.(\KeyGen, \Sign, \Ver)$.
    \item The state-preserving argument for NP $\SPNP$ from \Cref{constr:state-pres-np}.
\end{itemize}

\begin{construction}\label{constr:pke-poni}
    The core algorithms are as follows:
    \begin{itemize}
        \item $\KeyGen(1^\secpar)$ samples and outputs $(\sk, \pk) \gets \PKE.\KeyGen(1^\secpar)$.

        \item $\VerKeyGen(1^\secpar)$ samples a signing key pair $(\sk_\Sig, \pk_\Sig) \gets \Sig.\KeyGen(1^\secpar)$. It outputs $\verkey_s = \sk_\Sig$ and $\verkey_p = \pk_\Sig$.

        \item $\Enc(\pk, m; \verkey_s)$ parses $\verkey_s = \sk_\Sig$. Then, it encrypts $(\ket{\ct}, \verkey') \gets \PKE.\Enc(\pk, m)$. Finally, it signs $\sigma \gets \Sig.\Sign(\sk_\Sig, \verkey')$ and outputs
        \[
            \ket{\ct'} = \left(\ket{\ct},\ \verkey',\ \sigma\right)
        \]

        \item $\Dec(\sk, \ket{\ct'})$ parses $\ket{\ct'} = (\ket{\ct},\ \verkey',\ \sigma)$, then computes $m \gets \PKE.\Dec(\sk, \ket{\ct})$. It outputs $m$.
    \end{itemize}

    \noindent The \textbf{proof of no-intrusion} is as follows.
    \begin{enumerate}
        \item \textbf{Prover:} Parse $\ket{\ct'} = (\ket{\ct}, \verkey', \sigma)$ and send $(\verkey', \sigma)$ to the verifier. Coherently compute $\Del(\ket{\ct})$ to obtain $\sum_{\cert} \alpha_\cert \ket{\cert, \mathsf{garbage}_\cert}$.

        \item \textbf{Verifier:} Parse $\verkey_p = \pk_\Sig$. Then, compute $\Sig.\Ver(\pk_\Sig, \sigma, \verkey')$. If it rejects, immediately output reject.

        \item \textbf{Both.} Interact in an execution of $\SPNP$ for the language
        \[
            \lang = \{\verkey': \exists \cert \text{ s.t. } \PKE.\Ver(\verkey', \cert) = \Accept\}
        \]
        and the statement $\verkey'$.

        \item \textbf{Prover:} Uncompute the coherent implementation of $\Del$.
    \end{enumerate}
\end{construction}

\begin{theorem}
    \Cref{constr:pke-poni} is a PKE with everlasting search proofs of no-intrusion, assuming the existence of PKE with publicly verifiable certified deletion (which can be based on post-quantum PKE~\cite{TCC:BKMPW23,TCC:KitNisYam23}), post-quantum signature schemes, and that \Cref{constr:state-pres-np} is a state-preserving argument of knowledge for NP (all of which can be based on post-quantum LWE with superpolynomial modulus-to-noise ratio).
\end{theorem}
\begin{proof}
    Decryption correctness follows from the correctness of $\PKE$. Fix the outer keys and an input ciphertext. Suppose its classical signature is valid and coherent deletion produces an accepting certificate with probability at least $1-\zeta$. Projecting onto valid certificates changes the joint state of the ciphertext and an untouched external register by at most $\sqrt\zeta$ in trace distance. On the normalized valid component, the repeated NP argument has uniform preservation error $\Gamma_{\NP}$ and accepts perfectly. Contractivity, uncomputation of deletion, and a second triangle inequality give joint distance at most
    \[
       \Gamma_{\rm PoNI}:=2\sqrt\zeta+\Gamma_{\NP}
    \]
    from the original ciphertext together with an accepting verdict. The probability that signature verification rejects is added separately when averaging over honest encryption. For a pure joint input on the ciphertext and the external register, the final joint state conditioned on the protocol randomness and measurement outcomes has trace distance greater than $\tau$ from the input with probability at most $\Gamma_{\rm PoNI}/\tau^2$, and rejection has probability at most $\Gamma_{\rm PoNI}$. Let $\epsilon_{\rm del}$ and $\epsilon_{\rm sig}$ bound the honest deletion and signature-verification rejection probabilities, averaged over key generation and encryption. Averaging the preceding bounds and using Jensen's inequality gives total error at most $\epsilon_{\rm sig}+2\sqrt{\epsilon_{\rm del}}+\Gamma_{\NP}$, which is negligible. The conditional trace-distance bound above, applied to a purification and then restricted to the ciphertext register, proves the stated preservation guarantee for honestly generated ciphertexts. Decryption correctness and contractivity imply successful decryption afterward. With repeated audits, sum the channel errors and charge each fresh inner NP key once.

    Semantic security follows from the semantic security of $\PKE$.

    Next we show PoNI security.
    In the final PoNI, the NP statement proved in the $\SPNP$ execution must be for the original $\verkey'$ generated at encryption time, or else the unforgeability of $\Sig$ is broken.

    We show $\PoNIEncSearch$ search security by reducing to the publicly verifiable certified deletion property of $\PKE$. Specifically, given a $\PoNIEncSearch$ search adversary $(\adv_1, \adv_P, \adv_H)$, the reduction uses an efficient pre-deletion stage followed by an unbounded post-deletion distinguisher, as follows.
    \begin{enumerate}
        \item Sample random $m_0 \gets \{0,1\}^\secpar$ and $m_1\gets \{0,1\}^\secpar$ to be the challenge plaintexts. Receive $(\ket{ct_b}, \verkey')$, an encryption of a random $m_b$ along with the corresponding verification key, from the certified deletion challenger.
        \item Sample $(\sk_\Sig,\pk_\Sig)\gets\Sig.\KeyGen(1^\secpar)$ and set $\verkey=(\verkey_p,\verkey_s)=(\pk_\Sig,\sk_\Sig)$ and compute $\sigma = \Sig.\Sign(\sk_\Sig, \verkey')$ to complete the ciphertext $\ket{\ct'} = (\ket{\ct_b}, \verkey', \sigma)$, then give it to $\adv_1$ along with $\pk$ and $\verkey_p = \pk_\Sig$.
        \item Receive the adversary's split registers $(\regH, \calP)$.
        \item Execute the final $\PoNI$ with $\adv_P(\calP)$ using the argument of knowledge extractor $\Ext(\adv)$ for $\SPNP$. If it succeeds in extraction, it outputs $\cert$ such that $\PKE.\Ver(\verkey', \cert) = \Accept$ (since $\verkey'$ is the statement used as argued previously). If extraction fails, set $\cert = \bot$.
        \label{step:poni-red-cert}
        \item Send $\cert$ and $\regH$ to the certified deletion challenger.
        \item Receive either $(\regH,\pk,\sk)$ or $\bot$ from the certified deletion challenger, depending on whether $\cert$ is valid.\label{step:poni-red-valid}
        \item If the latter, output a random bit. If the former, run $m' \gets \adv_H(\regH,\sk,\verkey)$. Output $i$ if $m'=m_i\neq m_{1-i}$ for $i\in\{0,1\}$, and otherwise guess randomly.
    \end{enumerate}

    We now analyze the reduction's advantage in the certified deletion security game, supposing that the adversary violates $\PoNIEncSearch$ search security for the sake of contradiction. Observe that $\adv_H$'s view is independent of $m_{1-b}$, so $m' = m_{1-b}$ with negligible probability. Its guessing advantage is therefore at least half the probability that $m'=m_b$ and $\cert$ is valid, minus a negligible term. It suffices to show that this joint probability is noticeable.

    Consider the hybrid experiment where steps \ref{step:poni-red-cert} and \ref{step:poni-red-valid} are replaced by running the honest execution of $\SPNP$ (i.e. the honest $\PoNI$), and giving $(\regH,\sk,\verkey)$ to $\adv_H$ if the $\PoNI$ accepts (instead of checking an extracted $\cert$). This matches the $\PoNIEncSearch$ game for random message $m_b$, so by assumption $\adv_H$ outputs $m' = m_b$ with noticeable probability.\footnote{It is tempting to assume that this immediately implies the above reduction also finds $m' = m_b$ with noticeable probability. However, because we consider unbounded $\adv_H$ for everlasting security, running $\Ext$ is distinguishable to them, which may affect the simultaneous success probability of $\adv_H$ and $\adv_P$. We rely on \Cref{lem:simult-search} to avoid this issue.}

    Now consider the following simultaneous search task: $\adv_1$ acts as the splitter $A$; $\adv_H(\regH,\sk,\verkey)$ acts as (unbounded) $B$ and searches for $m_b$; $C$ takes as input a view from $\adv_P$ just before the opening phase in \Cref{constr:state-pres-np} and searches for an opening that would lead to an accepting response. This simultaneous search task matches the hybrid experiment described above, so $B$ and $C$ succeed with noticeable probability. Here $(\sk,\verkey,m_b)$ are included in $B$'s side information and remain unchanged in the following hybrid. Now consider replacing $C$'s input by a view from $\adv_P$ in an execution of $\SPNP$ where the public parameters are generated in \emph{injective} mode (this matches the extractor from the proof of \Cref{thm:witness-preserving-NP}). This view is computationally indistinguishable to $C$ because the trapdoor has not yet been revealed in the $\SPNP$ execution. Therefore \Cref{lem:simult-search} implies that $B$ and $C$ simultaneously succeed in their search tasks with noticeable probability even when $C$ receives the modified view. Finally, we argue that if $C$ succeeds in finding an accepting opening, then with overwhelming probability $\Ext$ successfully extracts a witness. This follows from the fact that if $\Ext$ would fail, then for the overwhelming majority of challenges in $\SPNP$, there does not exist an accepting opening.

    Quantitatively, after applying \Cref{lem:simult-search} to the efficient searcher $C$, remove the event that any inner injective key is bad, at cost at most $\sum_j\nu_{\injective,j}$. On its complement, extraction works for all openings simultaneously. Failure to extract a valid coloring contributes at most $(1-1/|\calE|)^k+\epsilon_{\rm chal}$ for $k$ parallel repetitions. Choose the repetition count and internal accuracy as in Section~\ref{sec:state-pres-np}, so these losses are negligible compared with the noticeable simultaneous-search probability. They are statistical losses and remain valid with an unbounded later searcher $B$. The computational mode change is applied only through the simultaneous-search lemma, to the view before trapdoor release; it is never applied directly to $B$. The outer encryption and signature keys are unaffected by the inner good-key events.

    Therefore with noticeable probability, $\Ext$ extracts a valid $\cert$ and $\adv_H$ outputs $m' = m_b$. As argued previously, this contradicts the publicly verifiable certified deletion security of $\PKE$.
\end{proof}
\section{Post-Hoc Verification of Complex Hamiltonians}
\label{app:post-hoc-complex-hamiltonians}

A complex $n$-qubit state is written
$\ket{\psi}=\ket{\psi_R}+i\ket{\psi_I}$ with
$\ket{\psi_R},\ket{\psi_I}\in\mathbb{R}^{2^n}$ and
$\lVert\psi_R\rVert^2+\lVert\psi_I\rVert^2=1$.

A complex QMA verifier $V=U_T\cdots U_1$ receives an $n$-qubit witness
$\ket{\psi}$, appends $k$ ancilla qubits in $\ket{0}^{\otimes k}$, runs a
polynomial-size circuit over the gate set
$\{H,T,\mathrm{CNOT},X,Z\}$, and accepts if a computational-basis
measurement of the first qubit yields $1$. Its acceptance probability is
\[
  P_V(\psi)
  =
  \bra{\psi}\bra{0}^{\otimes k}
  V^\dagger(\ketbra{1}\otimes I)V
  \ket{\psi}\ket{0}^{\otimes k}.
\]
For a promise problem $\lang$, suppose that $V$ has completeness $c$ and
soundness $s<c$.

\begin{theorem}[X/Z-verifiable encoding of complex QMA witnesses]
\label{thm:complex-xz}
Let $V=U_T\cdots U_1$ be a complex QMA verifier for $\lang$ with completeness
$c$ and soundness $s$. Define
\[
  \ket{-Y}:=\frac{1}{\sqrt{2}}(\ket{0}-i\ket{1}).
\]
There exists a compiled real circuit $\widetilde V$ with the following
properties.

\begin{enumerate}
\item[(i)] \textbf{Efficient encoding.}
The prover maps any witness $\ket{\psi}$ to
$\ket{\psi}\otimes\ket{-Y}$ using $O(1)$ gates---namely, $S^\dagger H$ on a
fresh ancilla---and then prepares the Kitaev history state
$\ket{\eta_{\widetilde V}}$ in $O(T)$ steps. Every gate of $\widetilde V$
is real, symmetric, and  its Pauli expansion contains
only strings in $\{I,X,Z\}^{\otimes r}$ for the appropriate constant $r$.

\item[(ii)] \textbf{Correctness and soundness of the compilation.}
For every witness $\ket{\psi}$,
\[
 \widetilde V
 \bigl(\ket{\psi}\ket{0}^{\otimes k}\ket{-Y}\bigr)
 =
 \bigl(V\ket{\psi}\ket{0}^{\otimes k}\bigr)\ket{-Y}.
\]
Moreover, on a NO instance, $\widetilde V$ accepts every state of the
witness and encoding-ancilla registers with probability at most $s$.

\item[(iii)] \textbf{X/Z-only Kitaev Hamiltonian.}
The Kitaev Hamiltonian $H_{\widetilde V}$ is a real linear combination of
tensor products of Pauli $I$, $X$, and $Z$ operators.

\item[(iv)] \textbf{X/Z measurement sufficiency.}
The energy test of Morimae--Fitzsimons~\cite{MF16} can be applied to
$H_{\widetilde V}$ using only single-qubit $X$- and $Z$-basis measurements.

\item[(v)] \textbf{One-copy verification after in-place amplification.}
Apply Marriott--Watrous amplification~\cite{CC:MW05} to $V$, using no
additional copies of the witness, to obtain a verifier $V^{\mathrm{amp}}$
with completeness at least $1-\varepsilon$ and soundness at most
$\varepsilon$, where $\varepsilon=2^{-q(n)}$ for a sufficiently large
polynomial $q$. Apply the compilation above to $V^{\mathrm{amp}}$, and
write $T'$ for the number of gates of the resulting circuit
$\widetilde V^{\mathrm{amp}}$.

Let
\[
 H_{\widetilde V^{\mathrm{amp}}}=\sum_P d_P P,
 \qquad
 D:=\sum_P|d_P|.
\]
Here $P$ ranges over Pauli strings, and $D=\poly(n)$. Define
\[
 a:=\frac{\varepsilon}{T'+1},
 \qquad
 b:=\Omega\!\left(\frac{1-\sqrt{\varepsilon}}{T'^3}\right).
\]
For sufficiently large $q$, $a<b$. The standard Kitaev analysis
\cite{Book:KSV02} gives a YES-instance state of energy at most $a$, whereas
on a NO instance every state has energy at least $b$. Consequently, the
single-round energy test has completeness and soundness satisfying
\[
 c'\geq \frac{1}{2}-\frac{a}{2D},
 \qquad
 s'\leq \frac{1}{2}-\frac{b}{2D},
\]
and hence
\[
 c'-s'\geq \frac{b-a}{2D}=\frac{1}{\poly(n)}.
\]
The honest prover uses exactly one copy of the original QMA witness.
\end{enumerate}
\end{theorem}

\paragraph{Outline of the proof of \Cref{thm:complex-xz}.}
We want to verify a complex QMA witness using only $X$- and $Z$-basis
measurements. The obstacle is that the $T$ gate is complex, so its
imaginary part introduces Pauli-$Y$ terms into the usual propagation
Hamiltonian. We append one encoding ancilla in the state $\ket{-Y}$ and
replace each $T$ gate by a two-step sequence of real, symmetric, Y-free
gates. Their combined action on the encoded subspace reproduces the action
of $T$.

The phase qubit is an unrestricted witness register in the soundness proof and is not initialized or checked as $\ket{-Y}$; every other clean verifier or synthesis ancilla receives the usual initialization check.
An arbitrary state splits into orthogonal $\ket{-Y}$ and $\ket{+Y}$
sectors. The compiled circuit applies $V$ to the first sector and $V^*$ to
the second. Since the accept operators of $V$ and $V^*$ have the same
spectrum, this introduces no new accepting witness on a NO instance.

Finally, part~(v) uses two standard ingredients with their usual
normalizations. First, in-place QMA amplification makes the honest
history-state energy negligible while leaving a NO-instance energy lower
bound of order $T'^{-3}$. Second, for a Pauli expansion
$H=\sum_P d_PP$, the single-round energy test accepts a state $\rho$ with
probability
\[
  p_{\mathrm{acc}}(\rho)
  =\frac{1}{2}-\frac{\Tr(H\rho)}{2\sum_P|d_P|}.
\]

\subsection{The compiled circuit}

Every compiled gate acts on the $n+k$ work qubits and on one dedicated
encoding ancilla, denoted $\mathsf{anc}$.

\paragraph{Real gates.}
For $U_t\in\{H,X,Z,\mathrm{CNOT}\}$, define
\[
  \widetilde G_t:=U_t\otimes I_{\mathsf{anc}}.
\]
These gates are real, symmetric, and Y-free.

\paragraph{The $T$ gate.}
Replace $T$ by the two-gate sequence
\[
 \mathrm{CH}
 :=\ketbra{0}_{\mathsf{work}}\otimes I_{\mathsf{anc}}
   +\ketbra{1}_{\mathsf{work}}\otimes H_{\mathsf{anc}},
\]
\[
 \mathrm{CZ}
 :=\ketbra{0}_{\mathsf{work}}\otimes I_{\mathsf{anc}}
   +\ketbra{1}_{\mathsf{work}}\otimes Z_{\mathsf{anc}}.
\]
Both gates are real, symmetric, and Y-free. The gate $\mathrm{CZ}$ is
applied first, followed by $\mathrm{CH}$, so the combined matrix is
$\mathrm{CH}\cdot\mathrm{CZ}$.

Thus $\widetilde V$ has $T'\leq 2T$ timesteps and gate set
\[
 \{U_t\otimes I_2:U_t\in\{H,X,Z,\mathrm{CNOT}\}\}
 \cup\{\mathrm{CZ},\mathrm{CH}\}.
\]

Let
\[
  J:=\begin{pmatrix}0&-1\\1&0\end{pmatrix}.
\]

\begin{lemma}[$\ket{-Y}$ is a $J$-eigenstate]
\label{lem:J-eigenstate}
$J\ket{-Y}=i\ket{-Y}$ and $J\ket{+Y}=-i\ket{+Y}$.
\end{lemma}

\begin{proof}
This follows by direct multiplication, using
$\ket{\mp Y}=(\ket{0}\mp i\ket{1})/\sqrt{2}$.
\end{proof}

\begin{lemma}[$\mathrm{CH}\cdot\mathrm{CZ}$ simulates the $T$ gate]
\label{lem:t-gadget}
Write $T=A_T+iB_T$, where
\[
 A_T=\operatorname{diag}\!\left(1,\frac{1}{\sqrt{2}}\right),
 \qquad
 B_T=\operatorname{diag}\!\left(0,\frac{1}{\sqrt{2}}\right).
\]
Then:
\begin{enumerate}
\item[(a)] $A_T\otimes I_2+B_T\otimes J$ is a real orthogonal matrix;
\item[(b)] $\mathrm{CH}\cdot\mathrm{CZ}=A_T\otimes I_2+B_T\otimes J$;
\item[(c)] $\mathrm{CH}$ and $\mathrm{CZ}$ are Y-free, with Pauli
decompositions
\begin{align*}
 \mathrm{CH}
 &=\frac{1}{2}\bigl(I\otimes I+Z\otimes I\bigr)
   +\frac{1}{2\sqrt{2}}
    \bigl(I\otimes X+I\otimes Z-Z\otimes X-Z\otimes Z\bigr),\\
 \mathrm{CZ}
 &=\frac{1}{2}\bigl(I\otimes I+Z\otimes I+I\otimes Z-Z\otimes Z\bigr).
\end{align*}
\end{enumerate}
\end{lemma}

\begin{proof}
For (a), use $J^T=-J$, $J^2=-I_2$,
$A_TB_T=B_TA_T$, and $A_T^2+B_T^2=I_2$:
\begin{align*}
 &(A_T\otimes I_2+B_T\otimes J)^T
   (A_T\otimes I_2+B_T\otimes J)\\
 &\quad=(A_T\otimes I_2-B_T\otimes J)
         (A_T\otimes I_2+B_T\otimes J)\\
 &\quad=(A_T^2+B_T^2)\otimes I_2
       +(A_TB_T-B_TA_T)\otimes J
   =I_4.
\end{align*}

For (b), in the ordered basis
$\ket{0,0},\ket{0,1},\ket{1,0},\ket{1,1}$,
\[
 A_T\otimes I_2+B_T\otimes J
 =
 \begin{pmatrix}
 1&0&0&0\\
 0&1&0&0\\
 0&0&1/\sqrt{2}&-1/\sqrt{2}\\
 0&0&1/\sqrt{2}& 1/\sqrt{2}
 \end{pmatrix}.
\]
On the other hand,
\[
 \mathrm{CZ}=\operatorname{diag}(1,1,1,-1),
 \qquad
 \mathrm{CH}=
 \begin{pmatrix}
 1&0&0&0\\
 0&1&0&0\\
 0&0&1/\sqrt{2}& 1/\sqrt{2}\\
 0&0&1/\sqrt{2}&-1/\sqrt{2}
 \end{pmatrix},
\]
and multiplication gives the displayed matrix.

For (c), expanding
$\ketbra{0}=(I+Z)/2$, $\ketbra{1}=(I-Z)/2$, and
$H=(X+Z)/\sqrt{2}$ gives the two stated decompositions directly. In
particular, all coefficients of Pauli strings containing $Y$ vanish.
\end{proof}

\begin{remark}
The matrix $A_T\otimes I_2+B_T\otimes J$ is used only to analyze the
two-gate block. It is not itself a gate of $\widetilde V$; the circuit
uses $\mathrm{CZ}$ and $\mathrm{CH}$ as separate timesteps.
\end{remark}

\begin{lemma}[Simulation on encoded inputs]
\label{lem:block-simulation}
Let $B_t$ denote the compiled logical block corresponding to an original
gate $U_t$: $B_t=U_t\otimes I_2$ for a real gate, while
$B_t=\mathrm{CH}\cdot\mathrm{CZ}$ for a $T$ gate. Then for every state
$\ket{\alpha}$ of the work register,
\[
  B_t\bigl(\ket{\alpha}\otimes\ket{-Y}\bigr)
  =U_t\ket{\alpha}\otimes\ket{-Y}.
\]
\end{lemma}

\begin{proof}
The claim is immediate for a real gate. For a $T$ gate,
\Cref{lem:t-gadget,lem:J-eigenstate} give
\begin{align*}
 &(A_T\otimes I_2+B_T\otimes J)
   (\ket{\alpha}\otimes\ket{-Y})\\
 &\qquad=A_T\ket{\alpha}\otimes\ket{-Y}
       +B_T\ket{\alpha}\otimes i\ket{-Y}
   =T\ket{\alpha}\otimes\ket{-Y}.
\end{align*}
\end{proof}

\begin{corollary}[Full-circuit simulation]
\label{cor:full-simulation}
For every witness $\ket{\psi}$,
\[
 \widetilde V
 \bigl(\ket{\psi}\ket{0}^{\otimes k}\ket{-Y}\bigr)
 =
 \bigl(V\ket{\psi}\ket{0}^{\otimes k}\bigr)\ket{-Y}.
\]
\end{corollary}

\begin{proof}
Apply \Cref{lem:block-simulation} successively to the logical blocks of the
compiled circuit. The encoding ancilla returns to $\ket{-Y}$ at the end of
every block.
\end{proof}

\subsection{Y-freeness implies an X/Z-only Hamiltonian}

\begin{definition}[Y-free gate]
An $r$-qubit gate $U$ is Y-free if its Pauli coefficient
$2^{-r}\Tr(PU)$ vanishes for every Pauli string $P$ containing a $Y$.
\end{definition}

\begin{lemma}[Real Y-free implies symmetric]
\label{lem:yfree-symmetric}
Every real Y-free gate is symmetric.
\end{lemma}

\begin{proof}
The Pauli coefficients of a real Y-free gate are real, so the gate is a
real linear combination of strings in $\{I,X,Z\}^{\otimes r}$. Every such
string is symmetric.
\end{proof}

\begin{lemma}[X/Z-only Kitaev Hamiltonian]
\label{lem:xz-hamiltonian}
$H_{\widetilde V}$ is a real linear combination of tensor products of
Pauli $I$, $X$, and $Z$ operators.
\end{lemma}

\begin{proof}
$H_{\mathrm{clock}}$, $H_{\mathrm{in}}$, and $H_{\mathrm{out}}$ are
diagonal in the computational basis and therefore have only $I/Z$ Pauli
support. Consider a propagation step with compiled gate
$\widetilde G_t$. Every such gate is real and Y-free, and hence
$\widetilde G_t^T=\widetilde G_t$ by \Cref{lem:yfree-symmetric}.

In the unary-clock construction, the off-diagonal part of the local
propagation term has the form
\[
 -\frac{1}{2}\Pi^{(Z)}_{t,\mathrm{nbr}}
 \otimes
 \left(
   \ket{1}\!\bra{0}_t\otimes\widetilde G_t
   +\ket{0}\!\bra{1}_t\otimes\widetilde G_t^T
 \right),
\]
where $\Pi^{(Z)}_{t,\mathrm{nbr}}$ is a product of computational-basis
projectors on the neighboring clock qubits. Symmetry of
$\widetilde G_t$ turns the parenthesis into
$X_t\otimes\widetilde G_t$. The remaining pieces of the propagation term
are diagonal clock projectors. Since both
$\Pi^{(Z)}_{t,\mathrm{nbr}}$ and $\widetilde G_t$ have only $I/X/Z$
Pauli support, so does every propagation term.
\end{proof}

The encoding ancilla participates in the propagation terms through the
$\mathrm{CH}$ and $\mathrm{CZ}$ gadgets, but only through $I$, $X$, and
$Z$. It is therefore measured only in the $X$ or $Z$ basis in the
single-qubit energy test.

\subsection{Soundness for arbitrary encoding-ancilla states}

Let
\[
 E_V
 :=(I\otimes\bra{0}^{\otimes k})
 V^\dagger\Pi_{\mathrm{acc}}V
 (I\otimes\ket{0}^{\otimes k})
\]
be the accept operator of $V$, where $\Pi_{\mathrm{acc}}$ is the final
computational-basis accepting projector.

\begin{lemma}[General-witness acceptance decomposition]
\label{lem:general-witness}
Let
\[
 \ket{\chi}
 =\ket{\chi_0}\ket{0}_{\mathsf{anc}}
  +\ket{\chi_1}\ket{1}_{\mathsf{anc}},
 \qquad
 \lVert\chi_0\rVert^2+\lVert\chi_1\rVert^2=1,
\]
be an arbitrary state of the witness and encoding-ancilla registers.
Define the generally unnormalized vectors
\[
 \ket{\phi_-}:=\frac{\ket{\chi_0}+i\ket{\chi_1}}{\sqrt{2}},
 \qquad
 \ket{\phi_+}:=\frac{\ket{\chi_0}-i\ket{\chi_1}}{\sqrt{2}}.
\]
Then
\[
 P_{\widetilde V}(\chi)
 =\bra{\phi_-}E_V\ket{\phi_-}
  +\bra{\phi_+}E_V^*\ket{\phi_+},
\]
where $E_V^*$ denotes entrywise complex conjugation in the
computational basis. In particular, if $x\in\lang_\no$ and
$P_V(\psi)\leq s$ for every normalized witness $\ket{\psi}$, then
$P_{\widetilde V}(\chi)\leq s$.
\end{lemma}

\begin{proof}
Rewrite the encoding ancilla in the $Y$ basis:
\[
 \ket{\chi}
 =\ket{\phi_-}\ket{-Y}_{\mathsf{anc}}
  +\ket{\phi_+}\ket{+Y}_{\mathsf{anc}},
\]
and note that
$\lVert\phi_-\rVert^2+\lVert\phi_+\rVert^2=1$.

For a real original gate, the compiled block acts identically on both
$Y$ sectors. For a $T$ gate, the block
$A_T\otimes I+B_T\otimes J$ acts as $T=A_T+iB_T$ on the
$\ket{-Y}$ sector and as $T^*=A_T-iB_T$ on the $\ket{+Y}$ sector.
Induction over the logical blocks gives
\begin{align*}
 \widetilde V(\ket{\chi}\ket{0}^{\otimes k})
 &=V(\ket{\phi_-}\ket{0}^{\otimes k})\ket{-Y}_{\mathsf{anc}}\\
 &\quad+V^*(\ket{\phi_+}\ket{0}^{\otimes k})\ket{+Y}_{\mathsf{anc}}.
\end{align*}
The accepting projector acts trivially on the encoding ancilla, and
$\braket{-Y | +Y}=0$, so the cross terms vanish.

Let $E_{V^*}$ denote the accept operator obtained from $V^*$. Since
$\Pi_{\mathrm{acc}}$ and $\ket{0}^{\otimes k}$ are real in the
computational basis, $E_{V^*}=E_V^*$. This proves the decomposition.

On a NO instance, $E_V\preceq sI$. Entrywise complex conjugation
preserves the positive-semidefinite order, so $E_V^*\preceq sI$ as well.
Therefore
\[
 P_{\widetilde V}(\chi)
 \leq s\lVert\phi_-\rVert^2+s\lVert\phi_+\rVert^2=s.
\]
The same conclusion holds for mixed input states by linearity.
\end{proof}

\subsection{Proof of the theorem}

Parts (i) and (ii) follow from
\Cref{lem:t-gadget,cor:full-simulation,lem:general-witness}. Part (iii) is
\Cref{lem:xz-hamiltonian}. Part (iv)
follows from the single-qubit Pauli energy test of
Morimae--Fitzsimons~\cite{MF16}.

For part (v), first apply Marriott--Watrous in-place amplification
\cite{CC:MW05} to the original verifier. This uses the original witness
register together with fresh ancillas and does not require additional
witness copies. It produces a verifier $V^{\mathrm{amp}}$ with completeness
$1-\varepsilon$ and soundness $\varepsilon$, for exponentially small
$\varepsilon$, with only polynomial circuit-size overhead. After decomposing
the amplified verifier over the fixed gate set, compile it as above.

On a YES instance, the honest history state has energy at most
\[
 a=\frac{\varepsilon}{T'+1}.
\]
On a NO instance, \Cref{lem:general-witness} ensures that the compiled
verifier accepts every possible input with probability at most
$\varepsilon$. The standard Kitaev geometric-lemma analysis
\cite{Book:KSV02} therefore gives
\[
 \lambda_{\min}(H_{\widetilde V^{\mathrm{amp}}})
 \geq b
 =\Omega\!\left(\frac{1-\sqrt{\varepsilon}}{T'^3}\right).
\]
Choosing the amplification error exponentially small makes $a<b$.

Finally, write
$H_{\widetilde V^{\mathrm{amp}}}=\sum_Pd_PP$ and
$D=\sum_P|d_P|$. The single-qubit Pauli energy test accepts a state $\rho$
with probability
\[
 \frac{1}{2}-\frac{\Tr(H_{\widetilde V^{\mathrm{amp}}}\rho)}{2D}.
\]
It follows that
\[
 c'\geq\frac{1}{2}-\frac{a}{2D},
 \qquad
 s'\leq\frac{1}{2}-\frac{b}{2D},
 \qquad
 c'-s'\geq\frac{b-a}{2D}=\frac{1}{\poly(n)}.
\]
This proves part (v) and the theorem.

\section{Generic Approach to One-Copy Amplified Completeness and Soundness}

In this section, we give a generic approach to witness-preserving $\CVQC$ starting from \emph{any} non-adaptive $\CVQC$ with strict polynomial-time provers and with arbitrary completeness and soundness $c-s\geq 1/\poly[\secp,|x|]$. Notably, $c$ may be far from $1$.
The end result is a $\CVQC$ which uses one copy of the witness and has completeness $c = 1-\negl$ and soundness error $s = \negl$.

As mentioned previously, not all witnesses can be amplified. If a state is in superposition over a good witness and junk, running the CVQC may collapse it to just junk. We quantify what it means to be a good witness similarly to Marriott and Watrous~\cite{CC:MW05}: a witness is good if it is an eigenstate of the QMA verifier with eigenvalue $\geq a$. Recall that we define $\eigenspace_{\geq a}(V)$ to be the span of eigenvectors of $V$ with eigenvalue $\geq a$.
For technical reasons, we require that the witnesses be slightly better than the minimum acceptance threshold $a$; this can be generically achieved by lowering the acceptance threshold to $(a + b)/2$ where $b$ is the QMA soundness threshold.

\begin{theorem}[1-Copy Amplification, Generic Compiler]
\label{thm:generic-high-completeness}
    Let $\calL=(\calL_\yes,\calL_\no)\in \QMA_{a,b}$ with a $\QMA$ verifier $V_\QMA$.
    Let $(P, V)$ be any non-adaptive $\CVQC$ w.r.t.\ $V_\QMA$ which has strict polynomial-time provers.

    Then there exists a generic compiler that converts any such $(P, V)$ with completeness $c$ and soundness error $s$ (or $(\eta,s)$-testable soundness) where $c - s = 1/\poly[\secpar]$ into a new $\CVQC$ protocol $(P', V')$ for $\calL$ which has the following guarantees:
    \begin{itemize}
        \item \textbf{Amplified Completeness.} For every witness $\ket{w} \in \eigenspace_{\geq a+p}(V_{\QMA})$ where $p = 1/\poly[\secpar]$,
        \[
            \Pr[(P'(\ket{w}), V')(1^\secpar, x) = 1] = 1-\negl
        \]
        \item \textbf{Amplified Computational Soundness.} It has $\negl$ computational soundness error (as defined in \Cref{def:CVQC}), assuming the post-quantum hardness of ${\sf LWE}$ with superpolynomial modulus-to-noise ratio.

        \item \textbf{Complexity.} The number of rounds and communication complexity are strict $\poly[\secpar, |x|]$ time. The prover's computational complexity is expected $\poly[\secpar, |x|, 1/p]$ time.

        \item{\bf $\eta$-Argument of Knowledge:} If the base protocol satisfies \Cref{def:2partpok} for the same specified relation $\calR$ with fixed inverse-polynomial extraction bound and completeness--soundness gap, it is an $\eta$-classical argument of quantum knowledge for $\calR$ (Definition \ref{def:pok}) with $\eta = \frac{1}{p(\secp)}$ for some polynomial $p(\cdot)$.
    \end{itemize}
\end{theorem}

The compiled $\CVQC$ unfortunately takes variable prover time and is thus not eligible for our gentle compiler. However, the prover time can be truncated to $1/\poly[\epsilon]$ if one is willing to accept an $\epsilon$ completeness error.
By the gentle measurement theorem, using the gentle compiler on such a truncation permits $2\sqrt{\epsilon}$-disturbance of the witness (conditioned on it accepting).

\paragraph{Construction Sketch.}
The main idea is to construct a $\CVQC$ which preserves the \emph{quality} of the prover's witness, up to an $\epsilon$ loss for some tunable parameter $\epsilon = \poly[\secpar]$. We achieve this via a fine-grained analysis of \cite{CMSZ22}'s state repair procedure. Then, we can sequentially repeat the resulting $\CVQC$ $N$ times to amplify completeness and soundness. The witness can degrade by up to $\epsilon$ each time, so the slack $p$ needs to be at least $\epsilon N$.

The parameter $\epsilon$ only affects the prover's runtime. Intuitively, they need to do extra work to repair the witness with higher precision. The slack $p$ does not need to be known a-priori, since the prover's first step is to estimate it.

\subsection{Preliminaries}
\subsubsection{Jordan's Lemma}

We provide here a few useful facts for reasoning about the interaction of two projectors.

\begin{lemma}[Jordan's Lemma]\label{lem:jordan}
    For any two Hermitian projectors $\Pi_A$ and $\Pi_B$ on a Hilbert space $\calH$, there exists an orthogonal decomposition of $\calH = \bigoplus_j S_j$ into one-dimensional and two-dimensional subspaces $\{\calS_j\}_{j}$ (referred to as the \emph{Jordan subspaces}), where each $\calS_j$ is invariant under both $\Pi_A$ and $\Pi_B$. Moreover:
    \begin{itemize}
        \item in each one-dimensional subspace, $\Pi_A$ and $\Pi_B$ act as identity or rank-zero projectors
        \item and in each two-dimensional subspace $S_j$, $\Pi_A$ and $\Pi_B$ are rank-one projectors. In particular, there exist distinct orthogonal bases $\{\ket{\jor_{j,1}^A, \ket{\jor_{j,0}^A}}\}$ and $\{\ket{\jor_{j,1}^B, \ket{\jor_{j,0}^B}}\}$ such that $\Pi_A$ projects onto $\ket{\jor_{j,1}^A}$ and $\Pi_B$ projects onto $\ket{\jor_{j,1}^B}$.
    \end{itemize}
\end{lemma}

\begin{lemma}\label{lem:aba}
    Let $\Pi_A$ and $\Pi_B$ be Hermitian projectors on a Hilbert space $\calH$ and let $\calH = \bigoplus_j S_j$ be the Jordan decomposition of the space corresponding to $\Pi_A$ and $\Pi_B$. Then $\Pi_B \Pi_A\Pi_B$ can be eigen-decomposed as
    \[
        \Pi_B \Pi_A\Pi_B = \sum_{j} p_j \ketbra{\jor_{j,1}^B}
    \]
    where $p_j = \left|\braket{\jor_{j,1}^B | \jor_{j,1}^A}\right|^2$.
\end{lemma}
\begin{proof}
    For each $\ket{\jor_{j,1}^B}$, we have
    \begin{align*}
        \Pi_B \Pi_A\Pi_B  \ket{\jor_{j,1}^B}
        &= \Pi_B (\ketbra{\jor_{j,1}^A}\ket{\jor_{j,1}^B})
        \\
        &=  \left|\braket{\jor_{j,1}^B | \jor_{j,1}^A}\right|^2 \ket{\jor_{j,1}^B}
    \end{align*}
    Furthermore, $\Tr[\Pi_B \Pi_A\Pi_B \ket{\jor_{j,0}^B}] = 0$ and $\calH = \bigoplus_j \calS_j$, so this characterizes all eigenvectors of $\Pi_B \Pi_A\Pi_B$ with nonzero eigenvalue.
\end{proof}

\subsubsection{State Repair}

\begin{definition}[Almost Projective Measurement]
A real-valued measurement ${\sf M}$ on $\calH$ is $(\epsilon,\delta)$-almost-projective if applying ${\sf M}$ twice in a row to any state $\brho\in S(\calH)$  produces measurement outcomes $p,p'$ where
\[
\Pr[|p-p'|\leq \epsilon     ]\geq 1-\delta
\]

\end{definition}

\noindent We borrow the following lemmas from \cite{CMSZ22} (using the formalism from \cite{LMS22}).

\begin{lemma}[\textbf{Value Estimation}, \cite{CMSZ22,LMS22}] \label{lemma:CMSZ-valest}
    Let $\mathcal{H}$ be a Hilbert space. There exists a quantum algorithm $\ValEst$ that satisfies the following guarantees:
    \begin{enumerate}
    \item
   $(\boldsymbol{\rho^{*}}, p^{*}) \leftarrow \ValEst_{V}(\boldsymbol{\rho}, \epsilon, \delta)
    $
   is given black-box access to the unitary dilation of a quantum verifier $V$ which takes as input the state $\brho$ and outputs  a bit $\{0, 1\}$.\footnote{$V$ is implemented by a unitary $U$ acting on registers $\calX \otimes \calA$. It takes as input a state on register $\calX$, then initializes $\calA$ to $\ket{0^n}$ and runs $U$. By black-box access to the unitary dilation, we mean that $\ValEst_{V}$ has black-box access to the unitaries $U$ and $U^\dagger$, along with direct access to $\calA$.} On input a quantum state $\boldsymbol{\rho} \in S(\mathcal{H})$ and accuracy parameters $\epsilon, \delta \in (0, 1]$, $\ValEst_V$ outputs a quantum state $\boldsymbol{\rho^{*}} \in S(\mathcal{H})$ and value $p^{*} \in X \subseteq [-\frac{1}{2}, \frac{3}{2}]$ for some discrete set $X$ where $N_{\epsilon, \delta} := |X| = 2t+1 = O\left(\epsilon^{-2}\log(2/\delta)\right)$, where $t=\lceil64\epsilon^{-2}\ln(16/\delta)\rceil$.\footnote{The reader should think of $p^{*}$ as a noisy estimate of the success probability, whose expectation is in $[0,1]$.}

       \item $\ValEst$ is an oracle circuit with $O\left(\frac{1}{\epsilon^2} \log \frac{2}{\delta}\right)$ gates.

        \item
       $          \mathbb{E}\left[p^{*} \hspace{3pt} \middle| \hspace{5pt} (\boldsymbol{\rho^{*}}, p^{*}) \leftarrow \ValEst_{V}(\boldsymbol{\rho}, \epsilon, \delta)\right] = \emph{Pr}\left[V(\brho) = 1
            \right].
        $

   \item  \label{prop-valest-est-to-actual}        For every $p^*\in [0,1]$, if \begin{equation}\label{eqn:val}\Pr_{(p,\brho')\gets \ValEst_V(\brho, \epsilon, \delta)}[p\geq p^*] \geq 1-\gamma,
   \end{equation}
   then
   \[\Pr[\Accept \gets V(\overline\rho')] \geq p^* - \gamma - \epsilon -\delta. \] Here $\overline\rho'$ is the residual state averaged over all estimator outcomes, including reset failures.
 \item \textbf{$\ValEst^{V}_{\epsilon, \delta} := \ValEst_{V}(\cdot, \epsilon, \delta)$ is an Almost Projective Family:}
        \begin{align*}
            \Pr\left[|p^{*}-p^{**}| \geq \max\{\epsilon, \epsilon'\} \hspace{3pt} \middle| \hspace{5pt}
            \begin{aligned}
                &(\boldsymbol{\rho}^{*}, p^{*}) \leftarrow \ValEst_{V}(\boldsymbol{\rho}, \epsilon, \delta) \\
                &(\boldsymbol{\rho^{**}}, p^{**}) \leftarrow \ValEst_{V}(\boldsymbol{\rho^{*}}, \epsilon', \delta')
            \end{aligned}
            \right] \leq \max\{\delta, \delta'\}.
        \end{align*}

    \end{enumerate}
\end{lemma}

\begin{lemma}[\textbf{Repair}, \cite{CMSZ22,LMS22}] \label{lemma:CMSZ-repair}
Let $\mathcal{H}$ be a Hilbert space. There exists a quantum algorithm $\Repair$ that satisfies the following guarantees:
\begin{enumerate}

\item $\boldsymbol{\sigma^{*}} \leftarrow \Repair_{{\sf M}, \Pi}(\boldsymbol{\sigma}, y, p, T)$ is given black-box access to an $(\epsilon, \delta)$-almost projective measurement ${\sf M}$ and a projective measurement $\Pi=(\Pi_y)_{y \in Y}$ on $\mathcal{H}$. On input a quantum state $\boldsymbol{\sigma} \in S(\mathcal{H})$, an outcome $y \in Y$, a real target $p \in \bbR$, and a maximum runtime $T\in \bbN$, it outputs a quantum state $\boldsymbol{\sigma^{*}} \in S(\mathcal{H})$.

\item For any $(\epsilon, \delta)$-almost projective measurement ${\sf M}$ on $\calH$, any projective measurement $\Pi$,
        and any mixed state $\rho$,
        \begin{align*}
            \emph{Pr}\left[|p^{*}-p^{**}| \geq 2\epsilon\hspace{3pt} \middle| \hspace{5pt} \begin{aligned}
                &(\brho_{p^*}, p^*) \leftarrow {\sf M}(\brho) \\
                &(\boldsymbol{\sigma}, y) \leftarrow \Pi(\brho_{p^*}) \\
                &\boldsymbol{\sigma^{*}} \leftarrow \Repair_{{\sf M}, \Pi}(\boldsymbol{\sigma}, y, p^{*}, T) \\
                &(\boldsymbol{\rho^{**}}, p^{**}) \leftarrow {\sf M}(\boldsymbol{\sigma^{*}})
            \end{aligned}\right] \leq N \cdot (\delta+1/T) + 4\sqrt{\delta},
        \end{align*}
        where $N=|Y|$ is the number of outcomes the projective measurement $\Pi$ can obtain.

        \item $\Repair$ is a variable-runtime oracle algorithm making
        $N+4T\sqrt{\delta} + 1$ oracle queries in expectation in the displayed estimate--disturb--repair experiment, for every initial mixed state. In the measurement-step convention of this bound, there are at most $2T+1$ alternating measurements: one initial test, followed by at most $T$ pairs. Counting calls to coherent measurement implementations and their inverses instead gives a cap of $4T+3$, including the final uncomputation. Thus the conversion between these query conventions costs only a fixed factor and an additive constant.
        \end{enumerate}
Setting $T = 1/\sqrt{\delta}$, the expected number of oracle calls is $N+5$ and the probability that $|p^* - p^{**}| \geq 2\epsilon$ is
$O(N\sqrt{\delta})$.
\end{lemma}

\subsection{Repairing High-Quality Witnesses}
\label{sec:high-quality-repair}

Before moving to our generic compiler for witness-preserving $\CVQC$, it will be useful to take a finer-grained approach to \cite{CMSZ22}'s state repair procedure. We show in this section that it can repair high quality $\QMA$ witnesses -- those that are accepted with overwhelming probability -- with overwhelming probability.

In particular, \cite{CMSZ22}'s state repair procedure uses a special probability estimation algorithm $\ValEst$.\footnote{Technically their state repair procedure can be used with any almost-projective measurement, but a probability estimation algorithm is required to repair the probability of acceptance of a QMA witness.}
Their repair procedure can be thought of as using $\ValEst$ to estimate the probability $p$ of $V(\ket{\psi})$ accepting, then repairing a damaged $\ket{\psi}$ so that the repaired state is accepted with probability $p -\epsilon$. However, $\ValEst$ is only guaranteed to be accurate \emph{on average}. For example, if $\ket{\psi}$ has success probability $p$, and is a superposition of $\ket{\psi_1}$ and $\ket{\psi_2}$, where $\ket{\psi_1}$ is accepted with probability $1$ and $\ket{\psi_2}$ is accepted with probability $0$, the $\ValEst$ procedure may collapse $\ket{\psi}$ into a mixture of $\ket{\psi_1}$ and $\ket{\psi_2}$. If the witness is collapsed to $\ket{\psi_2}$, then the repair procedure will never return the success probability to $\approx p$.
Although intuitively an overwhelmingly-good witness cannot collapse like this, that fact surprisingly does not follow from the black-box properties of $\ValEst$ presented in \cite{CMSZ22}.\footnote{In more detail, they only give the implication that $p=\Pr[\Accept \gets V(\rho)]$ implies that $\ValEst(\rho)$ outputs~$p$ on expectation  (property 3 of \Cref{lemma:CMSZ-valest}). Unfortunately, even if $p$ is close to~$1$ this is not enough to argue that $\ValEst$ outputs a value close to~$1$ with high probability because $\ValEst$ can return an estimate which is larger than~$1$, since its output is in $[-\frac12,\frac32]$.}

To ensure that we can repair a $\QMA$ witness $\ket{\psi}$ for $\QMA$, we need to ensure that $\ValEst(\ket{\psi})$ gives an estimate close to the probability that $\ket{\psi}$ passes the $\QMA$ verification procedure (e.g.\ $2/3$) \emph{with overwhelming probability}.
This will rule out the possibility of a problematic collapse, except with negligible probability.

We show in this section that if $\ket{\psi}$ is an eigenvector of the verification procedure $V$ with eigenvalue $p$, then $\ValEst(\ket{\psi})$ returns an estimate $p \pm \epsilon$ with overwhelming probability.
Such states are exactly Marriott-Watrous witnesses when $p$ is greater than the QMA threshold (e.g. $2/3$).
As a corollary, any damaged witness which was originally accepted by $V$ with overwhelming probability can be repaired to a state which is accepted with $1-3\epsilon-\negl$ probability.\footnote{Running $\ValEst$ may reduce the accepting probability by an additive factor of $\epsilon$, and running $\Repair$ may reduce the accepting probability by another additive factor of $2\epsilon$ (see \Cref{lemma:CMSZ-repair}).}

We prove a concentration inequality for $\ValEst$ on high-quality witnesses in \Cref{sec:high-quality-estimation} (\Cref{lem:cmsz-eigenstates} and \Cref{coro:valest-high-quality}). Then, we use this to show in \Cref{sec:high-quality-repair-subsec} that repairing a damaged high-quality witness results in a state $\rho$ which admits a convex decomposition assigning overwhelming weight to witnesses of $1-3\epsilon-\negl$ quality.

\paragraph{A Note on Other Probability Estimation Algorithms.}
\cite{CC:MW05}'s original probability estimation algorithm is unfortunately \emph{not} compatible with the repair technique, because it is not approximately projective. \cite{TCC:Zhandry20} refines \cite{CC:MW05}'s algorithm so that it becomes approximately projective, although the resulting algorithm runs in {\em expected} polynomial time as opposed to strict polynomial time. Unfortunately, \cite{CMSZ22}'s repair procedure requires the $(\epsilon,\delta)$-approximately-projective measurement to run in \emph{strict} polynomial time.\footnote{For exampley be tempting to simply truncate the runtime in Zhandry's algorithm. However the extra runtime is necessary to return the state to its ``original form'' so that the procedure is approximately projective. Specifically, if the true value being estimated is $p$, it takes $\approx 1/p$ time to return the state to its ``original form.''} \cite{CMSZ22} modifies Zhandry's algorithm so that it can be truncated while still being approximately projective.

\subsubsection{Estimating the Acceptance Probability of a High-Quality Witness}\label{sec:high-quality-estimation}

\begin{lemma}\label{lem:cmsz-eigenstates}
    Let $V$ be a binary-outcome quantum algorithm and let $\ValEst$ be \cite{CMSZ22}'s probability estimation algorithm. For every eigenstate $\ket{\psi}$ of $V$ with eigenvalue $p^*$,
    \[
        \Pr\big[|p - p^*| > \epsilon : (p, \rho) \gets \ValEst_{V}(\ket{\psi}, \epsilon, \delta)\big] \leq \delta.
    \]
    Furthermore, if $\ket{\phi} = \sum_{j}\alpha_j \ket{\psi_j}$ for eigenstates $\ket{\psi_j}$ of $V$ with eigenvalues $p_j$ and there is a range $[p^*_1, p^*_2]$ such that
    \[
        \sum_{j:p_j \notin [p^*_1, p^*_2]} |\alpha_j|^2 \leq \eta,
    \]
    then
    \[
        \Pr\big[p \notin [p^*_1 - \epsilon, p^*_2 + \epsilon] : (p, \rho) \gets \ValEst_{V}(\ket{\phi}, \epsilon, \delta) \big] \leq \delta + \eta
    \]
\end{lemma}

Using this lemma, we can bound the probability that $\ValEst(\ket{\psi})$ gives an estimate close to $1$ when $\ket{\psi}$ is accepted by $V$ with overwhelming probability by reasoning about its representation in the eigenbasis of $V$.

\begin{corollary}\label{coro:valest-high-quality}
    For any $\ket{\phi}$ such that $\Pr[\Accept \gets V(\ket{\phi})] \geq 1-\eta$,
    where $\sqrt{\eta} \leq \epsilon/4$,
    \[
        \Pr[|1-p|\geq \epsilon: (p, \brho) \gets \ValEst_V(\ket{\phi}, \epsilon, \delta)] \leq \delta + \sqrt{\eta}.
    \]

\end{corollary}
\begin{proof}
    Decompose $\ket{\phi} = \sum_{j} \alpha_j \ket{\psi_j}$ in terms of eigenstates $\ket{\psi_j}$ of $V$ with eigenvalues $p_j$. The probability of acceptance is
    \begin{align*}
        \bra{\phi}V \ket{\phi}
        &= \sum_{j} |\alpha_j|^2 p_j
        \\
        &= \sum_{j: p_j \geq 1-\sqrt{\eta}} |\alpha_j|^2 p_j + \sum_{j: p_j < 1-\sqrt{\eta}} |\alpha_j|^2 p_j
        \\
        &\leq \sum_{j: p_j \geq 1-\sqrt{\eta}} |\alpha_j|^2 + \sum_{j: p_j < 1-\sqrt{\eta}} |\alpha_j|^2 (1-\sqrt{\eta})
        \\
        &= 1 - \sqrt{\eta}\sum_{j: p_j < 1-\sqrt{\eta}} |\alpha_j|^2
    \end{align*}
    Comparing to the bound $\Pr[\Accept \gets V(\ket{\phi})] \geq 1-\eta$ and rearranging (with the case $\eta=0$ immediate),
    \[
        \sum_{j: p_j < 1-\sqrt{\eta}} |\alpha_j|^2  \leq \sqrt{\eta}.
    \]
    By the binomial concentration bound in the proof of \Cref{lem:cmsz-eigenstates}, with the same accuracy $\epsilon$, the estimate satisfies $|p-p_j|<\epsilon/3$ except with probability at most $\delta/8$ on each eigenstate. Whenever $p_j\geq1-\sqrt\eta$ and this event holds, $|1-p|<\sqrt\eta+\epsilon/3\leq\epsilon/4+\epsilon/3<\epsilon$. The total failure probability is at most $\sqrt\eta+\delta/8\leq\sqrt\eta+\delta$.

\end{proof}

\begin{proof}[Proof of \Cref{lem:cmsz-eigenstates}]
    We recall $\ValEst$ for completeness. Let the unitary $U_V$ and the projective measurement $\Pi_V$ denote the implementation of $V$. Let $\calH$ be the input register to $V$. We take $t = \left\lceil 64\epsilon^{-2}\ln\left(\frac{16}{\delta}\right)\right\rceil = \Theta\left(\frac{1}{\epsilon^2}\log\left(\frac{2}{\delta}\right)\right)$, which controls the runtime.
    $\mathsf{NReps}(\vec{b})$ is the number of consecutive repeated bits in $\vec{b}$, divided by $|\vec{b}| - 1$; for example, $\mathsf{NReps}(001110) = 3/5$. \begin{enumerate}
        \item Initialize register $\calR$ to $\ket{0}$ and initialize register $\calT$ to $\ket{\cmszmodereg} \coloneqq \frac{1}{\sqrt{2}}(\ket{0} + \ket{\top/\bot})$, where $\ket{\topbot} \coloneqq \frac{1}{\sqrt{2}}\left(\ket{\top} + \ket{\bot}\right)$. Here, $\ket{\top}$ and $\ket{\bot}$ are special basis elements representing ``automatically'' winning or losing the game, respectively.

        \item Define measurement $M_{G} = (\Pi_G, I - \Pi_G)$ where
        \[
            \Pi_G \coloneqq \left(U_V^\dagger \Pi_{V}^\Accept U_V \right) \otimes \ketbra{0}_{\calT} + I_{\calH, \calR} \otimes \ketbra{\top}_{\calT}
        \]
        In words, $\Pi_G$ evaluates $V$ when $\calT$ contains $\ket{0}$ and accepts any input when $\calT$ contains $\ket{\top}$. It rejects any input when $\calT$ contains $\ket{\bot}$.

        Define the meaurement $M_{\reset} = (\ketbra{0,\cmszmodereg}_{\calR, \calT}, I - \ketbra{0,\cmszmodereg}_{\calR, \calT})$.

        \item For $i = 1$ to $t$:
        \begin{enumerate}
            \item Apply $M_G$, obtaining outcome $L_{2i-1} \in \{0,1\}$.
            \item Apply $M_{\reset}$ to registers $(\calR, \calT)$, obtaining outcome $L_{2i} \in \{0,1\}$.
        \end{enumerate}
        \item If $L_{2t} = 1$, skip to the next step. Otherwise, apply $M_G$ and $M_\reset$ in an alternating fashion as above until $M_{\reset}$ outputs $1$ or a further $2t$ measurements have been applied.
        \item Discard registers $\calR$, $\calR'$, and $\calT$, then output the estimate $P \coloneq 2\mathsf{NReps}(1, L_1, \dots, L_{2t}) - 1/2$ (denoted by $p$ in the surrounding statements).
    \end{enumerate}
    This is equivalent to \cite{CMSZ22}'s algorithm, but the description varies slightly. In \cite{CMSZ22}, they consider $V$ to take in ancillas of the form $\frac{1}{\sqrt{|R|}}\sum_{r\in R}\ket{r, r}$, rather than $\ket{0}$. These are easily accounted for by having $U_V$ first map $\ket{0}$ to $\frac{1}{\sqrt{|R|}}\sum_{r\in R}\ket{r, r}$. The other difference is that we separate the ``automatic'' win/loss symbols into their own register, rather than squeezing them into $\calR$ as well. Since $\ket{\top}$ and $\ket{\bot}$ are already orthogonal to the other computational basis vectors in \cite{CMSZ22}, this has the same effect but makes the definition of $\Pi_G$ easier to read.

    Choose an orthonormal eigenbasis $\{\ket{\psi_j}\}_j$ of the acceptance operator $V$, with eigenvalues $e_j\in[0,1]$. Writing $B=I_{\calH}\otimes\ketbra{0,\cmszmodereg}$, direct compression gives
    \[
        (I\otimes\bra{0,\cmszmodereg})\Pi_G
        (I\otimes\ket{0,\cmszmodereg})=\tfrac12 V+\tfrac14 I.
    \]
    Thus the vectors $\ket{\psi_j}\ket{0,\cmszmodereg}$ are the reset-$1$ vectors in an orthonormal Jordan eigenbasis for $B,\Pi_G$, with squared overlaps $q_j=e_j/2+1/4\in[1/4,3/4]$. Both alternating measurements preserve these orthogonal Jordan blocks.

    \begin{claim}[\cite{CMSZ22}, Lemma 4.5 + Proposition 4.7 + Chernoff Bound]
        If $\ket{\psi}=\sum_j\alpha_j\ket{\psi_j}$ in this orthonormal basis, the estimator distribution is obtained by sampling $j$ with probability $|\alpha_j|^2$ and then
        \[
            X\sim\mathsf{Bin}(2t,q_j),\qquad
            P=X/t-\tfrac12,\qquad
            \mathbb E[P\mid j]=e_j.
        \]
        Every run is included, including reset failures; the estimate is neither clipped nor postselected. Moreover,
        \[
            \Pr[|P-e_j|\geq\epsilon/3\mid j]
            \leq2e^{-t\epsilon^2/9}\leq\delta/8,
            \qquad
            \Pr[\text{reset failure}\mid j]\leq(5/8)^t\leq\delta/8.
        \]
    \end{claim}

    Indeed, each repeated-bit indicator in a fixed Jordan block has conditional probability $q_j$, independently of the preceding outcomes, so their sum has the stated binomial law. A failed reset succeeds in each additional measurement pair with probability $2q_j(1-q_j)\geq3/8$. The reset procedure does not change the recorded estimate. There are $2t+1$ possible estimates and $O(t)=O(\epsilon^{-2}\log(2/\delta))$ oracle calls. Unbiasedness on arbitrary mixed inputs follows by linearity.

    This also verifies almost-projectivity for two different choices $(\epsilon,\delta)$ and $(\epsilon',\delta')$. Except for the first reset failure and the two estimation tails, both estimates concern the same eigenvalue $e_j$, with errors less than $\epsilon/3$ and $\epsilon'/3$. Their difference is then less than $\max\{\epsilon,\epsilon'\}$. The union bound is $\delta/8+\delta/8+\delta'/8\leq\max\{\delta,\delta'\}$. The second tail includes its reset failures. The argument applies to superpositions and mixed states because the classical outcome laws are mixtures over the orthogonal blocks.

    For completeness, property~\ref{prop-valest-est-to-actual} of \Cref{lemma:CMSZ-valest} concerns the averaged residual state. The initial weight on eigenvalues below $p^*-\epsilon/3$ is at most $\gamma+\delta/8$ by the first tail bound. Successful resets preserve the eigenvalue; their failure costs at most another $\delta/8$. For $p^*\in[0,1]$, this yields residual acceptance at least $p^*-\gamma-\epsilon-\delta$ (trivially if the right-hand side is negative).

    An eigenstate of $V$ with eigenvalue $p^*$ therefore has estimate within $\epsilon$ of $p^*$ except with probability at most $\delta$. For the spectral-subspace statement, the initial weight outside $[p_1^*,p_2^*]$ is at most $\eta$; adding the estimation tail proves the $\delta+\eta$ bound in \Cref{lem:cmsz-eigenstates}.

\end{proof}

\subsubsection{Result of Repairing a High-Quality Witness}\label{sec:high-quality-repair-subsec}

Next, we prove that repairing a damaged high-quality witness will, with overwhelming probability, produce a witness which is accepted with probability $1-3\epsilon - \negl$. Using \cref{coro:valest-high-quality}, estimating the value of a high quality witness $\ket{w}$ will return $\geq 1-\epsilon$ with overwhelming probability. \Cref{lemma:CMSZ-repair} guarantees that the repaired mixed state $\rho$ then returns an estimate within $2\epsilon$ of this with overwhelming probability, i.e. the final estimate is $\geq 1-3\epsilon$.
Finally, we show in the following lemma that this implies that $\rho$ has a convex decomposition assigning overwhelming weight to \emph{pure state} witnesses which are accepted with probability $\geq 1-3\epsilon - \negl$.

\begin{lemma}\label{lem:valest-to-pure}
    Let $\rho$ be a state such that
    \[
        \Pr[p\geq p^*: (p, \rho') \gets \ValEst_V(\rho, \epsilon, \delta)]
        \geq 1-\gamma
    \]
    for some $p^*\in[0,1]$. Then $\rho = q_\mathsf{good} \rho_\mathsf{good}  + q_\mathsf{bad} \rho_\mathsf{bad}$ can be decomposed as a mixture over two orthogonal mixed states where
    \[
        q_\mathsf{bad} \leq \sqrt{\gamma}
    \]
    and $\rho_\mathsf{good}$ admits a convex decomposition into pure states $\ket{\psi}$ such that
    \[
        \Pr[\Accept \gets V(\ket{\psi})] \geq p^* - \tfrac32\sqrt{\gamma}
    \]
\end{lemma}

\begin{proof}
    Consider the binary-outcome measurement where one (coherently) estimates $p \gets \ValEst_V(\ket{\psi}, \epsilon, \delta)$ and measures whether $p\geq p^*$ or $p<p^*$. By \Cref{lem:mixed-to-pure},
    we can decompose $\rho = q_\mathsf{good} \rho_\mathsf{good}  + q_\mathsf{bad} \rho_\mathsf{bad}$ where
    \[
        q_\mathsf{bad} \leq \sqrt{\gamma}
    \]
    and
    $\rho_\mathsf{good}$ admits a convex decomposition into pure states $\ket{\psi}$ such that
    \[
        \Pr[p\geq p^* \gets \ValEst_V(\ket{\psi}, \epsilon, \delta)] \geq 1 - \sqrt{\gamma}
    \]
    By unbiasedness and $P\geq-1/2$, each such pure component satisfies
    \[
        \Pr[\Accept\gets V(\ket{\psi})]
        =\mathbb E[P]\geq p^*(1-\sqrt\gamma)-\tfrac12\sqrt\gamma
        \geq p^*-\tfrac32\sqrt\gamma,
    \]
    where $p^*\in[0,1]$. This is a bound for each component of the stated convex decomposition, and hence also for $\rho_{\mathsf{good}}$.

\end{proof}

\subsection{Generic Compiler}\label{sec:generic-compiler}\label{sec:epsilon-nd}

The main idea is to construct a $\CVQC$ which preserves the \emph{quality} of the prover's witness, up to an $\epsilon$ loss.
As mentioned previously, not all witnesses can be re-used. We use the following notion of the quality of a witness, which roughly corresponds to the eigenvalue of the witness with respect to $V_{\QMA, x}$. Generally, it is associated with two parameters:
\begin{itemize}
    \item \textbf{Quality:} A parameter $p\in[0,1]$ corresponding to the probability that the witness is accepted by $V_{\QMA,x}$.

    \item \textbf{Repair Decay:} A parameter $\epsilon \in (0,1]$ bounding how much worse the witness gets after running the protocol: an $\epsilon$-non-destructive protocol which takes in a witness in $\Repairable_{p, \epsilon}$ should output one in $\Repairable_{p-\epsilon,\epsilon}$.
\end{itemize}

\label{repairable-sec}
\begin{definition}[$\Repairable_{p,\epsilon}$]\label{def:high-qual}
    Fix any language $\calL=(\calL_\yes,\calL_\no)\in\QMA_{a,b}$ and any $\QMA$ verifier $V_\QMA$. Then for any $x\in\calL_\yes$,  $\rho\in\Repairable_{p,\epsilon}(x)$
    if and only if, at the current security parameter $\secp$ and with $\delta=2^{-\secp}$
    \[
        \Pr\left[p'<p-\epsilon : ~(p',\brho)\gets\ValEst_{V_{\QMA,x(1^\secp,\cdot)}}(\rho,\epsilon,\delta)\right]\leq \sqrt{2\delta}
    \]
    Here $p\in[0,1]$, $0<\epsilon\leq1$, and $\rho$ may be mixed. The dependence of $\Repairable$ on the current security parameter is suppressed. Here $\ValEst$ is the almost projective measurement defined in \Cref{lemma:CMSZ-valest}.
\end{definition}
By unbiasedness in \Cref{lemma:CMSZ-valest}, since estimates lie in $[-1/2,3/2]$ and $p\in[0,1]$, every $\rho\in\Repairable_{p,\epsilon}(x)$ satisfies
\[
    \Pr[\Accept \gets V_{\QMA,x}(\rho)] \geq p - \epsilon - \tfrac32\sqrt{2\delta}
\]
Conversely, there is a way to generically transform any $\ket{w}$ which is accepted with probability $1-\negl$ into a repairable witness.
\Cref{coro:valest-high-quality} shows that, for inverse-polynomial $0<\epsilon\leq1$, the first estimate is at least $1-\epsilon$ with overwhelming probability. Almost-projectivity bounds the joint event that this holds but a second estimate is below $1-2\epsilon$ by $\delta$. Conditioning on the first event gives at most $2\delta$ for sufficiently large parameters. Markov's inequality then shows that, with overwhelming probability over the first history, the residual state is in $\Repairable_{1-\epsilon,\epsilon}$, whose defining failure threshold is $\sqrt{2\delta}$.

\Cref{lem:cmsz-eigenstates} gives a more precise characterization of the $\Repairable$ set.
We mention that every witness $\ket{w}\in\eigenspace_{\geq p}(V_{\QMA,x})$ is in $\Repairable_{p,\epsilon}(x)$ for every $0<\epsilon\leq1$.
Recall from \Cref{def:eigen} that $\eigenspace_{\geq p}$ is the span of eigenvectors of $V_{\QMA,x}$ with eigenvalue $\geq p$. $\eigenspace_{\geq p}$ is precisely the set of Marriott-Watrous witnesses when $p$ is the $\QMA$ threshold.

Given a measurement of the quality of a witness, we can now define $\epsilon$-non-destructive. Roughly, it guarantees that if the prover starts with a $(p, \epsilon)$-repairable witness, then after running the protocol the prover ends with a $(p-\epsilon, \epsilon)$-repairable witness.

\begin{definition}
    [$\epsilon$-Non-Destructive]  We say that a $\CVQC$ protocol $(P,V)$ for $\calL\in\QMA_{a,b}$ is $\epsilon$-non-destructive if there exists a negligible function~$\mu$ such that for every $\secp\in\mathbb{N}$, every $x\in\calL_\yes$ and every witness state $\rho\in \Repairable_{p,\epsilon}(x)$, for $p\geq a+\epsilon$,
    \[
        \Pr_{(\rho_h,b)\gets (P(\rho),V)(1^\secp,x)}\big[\rho_h\in\Repairable_{p-\epsilon,\epsilon}(x)\big]=1-\mu(\secp).
    \]
    Here $h$ records the protocol randomness and classical measurement outcomes, and $\rho_h$ is the conditional residual witness state, with auxiliary registers discarded.
\end{definition}

The main result of this section is to generically compile any \emph{gentle} $\CVQC$ into an $\epsilon$-non-destructive $\CVQC$, for any $\epsilon = \poly[1/\secpar]$.

\begin{theorem}\label{thm:eps-non-des}
    Let $\calL=(\calL_\yes,\calL_\no)\in \QMA_{a,b}$ with a $\QMA$ verifier $V_\QMA$.
    Let $(P, V)$ be any $\CVQC$ w.r.t.\ $V_\QMA$ and let $(P', V')$ be any $\CVQC$ which is $(1-2^{-\secpar})$-gentle with respect to $(P, V)$.

    There exists a generic compiler, associated with $0<\epsilon=\epsilon(\secp,|x|)\leq1$ satisfying $\tfrac32\sqrt{2\delta}\leq2\epsilon$, where $\delta=2^{-\secp}$ and $T=\lceil2^{\secp/2}\rceil$, that converts any such $(P', V')$ and $(P, V)$ with completeness $c$ and soundness~$s$ (or $(\eta,s)$-testable soundness) into a new $\CVQC$ protocol $(P'',V'')$ for $\calL$ (described in  \Cref{constr:eps-nd}) which is
    \begin{itemize}
        \item {\bf Quality preservation with loss $4\epsilon$ at estimation accuracy $\epsilon$:}  There exists a negligible function~$\mu$ such that for every $\secp\in\mathbb{N}$, every $x\in\calL_\yes$ and every witness state $\rho\in \Repairable_{p,\epsilon}(x)$, for $p\geq a+4\epsilon$,
        \[
            \Pr_{(\rho_h,b)\gets (P''(\rho),V'')(1^\secp,x)}\big[\rho_h\in\Repairable_{p-4\epsilon,\epsilon}(x)\big] \geq 1-\mu(\secp).
        \]
    \end{itemize}
    Moreover, it preserves the following guarantees from the original protocol $(P', V')$.
    \begin{enumerate}
        \item {\bf Computational soundness:}  $(P'',V'')$ has soundness $s+\negl$; in the testable case, it has $(\eta,s+\negl)$-testable soundness with the same $\eta$.

        \item {\bf Completeness: } For every $x\in\calL_\yes$, the protocol $(P'',V'')$ has completeness $c-\negl$, for the repairable witnesses specified above.

        \item{\bf Complexity:}  The communication complexity and computational complexity grow by a multiplicative factor of ${\sf poly}(\secp)$ and the prover's expected computational complexity grows by an additive factor of $\poly[|x|,\secpar,1/\epsilon]$.

        \item {\bf Argument of Knowledge:} If $(P',V')$ is a $(d, s,\delta)$ (testable) argument of knowledge, then $(P'',V'')$  is a $(d+1, s+\negl,\delta-\negl)$ (testable) argument of knowledge.
    \end{enumerate}
\end{theorem}

We give the construction in \Cref{sec:generic-eps-construction} and analyze it in \Cref{sec:generic-eps-analysis}.

Using the new $\epsilon$-non-destructive $\CVQC$, we prove the generic approach to amplification with one copy of the witness (\Cref{thm:generic-high-completeness}) by sequentially repeating the compiled protocol using \Cref{constr:sequential-repetition}.

\begin{proof}[Proof of \Cref{thm:generic-high-completeness}]
    Let $(P, V)$ be any non-adaptive $\CVQC$ satisfying the theorem's hypotheses. In the ordinary-soundness case, first give it a test/check structure: sample a fair private branch bit, use the original verifier's messages in both branches, let $V_\test$ always accept, and let $V_\honest=V$. The honest test has perfect completeness and the check has the original completeness $c_0$ and soundness $s_0$, so this wrapper has $(1,s_0)$-testable soundness and remains non-adaptive. Apply the compilers below to this wrapper, taking its base test tolerance to be $1$. This adjustment establishes ordinary soundness; it does not supply the additional knowledge premise. First apply \Cref{thm:gentle}, with the quantitative gentleness budget of \Cref{par:compiler-accuracy}, then apply \Cref{constr:eps-nd} to obtain $(P', V')$ (properties shown by \Cref{thm:eps-non-des}); let $(P'', V'')$ be its sequential repetition. Let $c_0,s_0$ be fixed base bounds with inverse-polynomial gap. Choose $c_*=c_0-(c_0-s_0)/4$ and $s_*=s_0+(c_0-s_0)/4$ to absorb the total negligible compiler losses, and write $c=c_*$, $s=s_*$, $\Delta=c_*-s_*$. Take the base test tolerance divided by $2$ to account for the gentle compiler. For the conditional knowledge conclusion with base parameters $(d,s_0,\delta_0)$ for the same relation, the gentle compiler and \Cref{thm:eps-non-des} give $d_*=d+2$; take $\delta_*=\delta_0/2$ and reduce $\eta$ to at most $1/(n^{d_*}\secp^{d_*})$ before choosing $N$. All these are fixed computable choices, independent of the adversary. They absorb the losses for sufficiently large parameters, uniformly over allowed instances and advice for each efficient prover. Use \Cref{constr:sequential-repetition} with
    \[
        N=\left\lceil128\secp\left(\frac{1}{\eta\Delta}+\frac{1}{\Delta^2}\right)\right\rceil.
    \]

    We prove that $(P'',V'')$ has $1-\negl$ completeness and $\negl$ computational soundness error. Choose the invocation accuracy as described below; $n=|x|\leq q(\secp)$ for the fixed polynomial $q$, including any lower-bound padding.

    Soundness follows directly from \Cref{thm:sequential-soundness}. Under the stated base-knowledge hypothesis for the same relation, Section~\ref{sec:sequential-rep} gives the knowledge conclusion with coefficient bounded below uniformly by $1/{\sf poly}(\secp)$.

\paragraph{Completeness:}
    Choose $\epsilon$ after $N$, and fix any $x\in\calL_\yes$ and $\rho_0\in\Repairable_{p,\epsilon}(x)$ with $p\geq a+4\epsilon(N+1)$.
    For any available inverse-polynomial slack $p-a$, taking $\epsilon=(p-a)/(8(N+1))$ suffices and keeps $1/\epsilon$ polynomial.
    Run each estimate--interaction--repair invocation with logical accuracy $\widehat\secpar=\secpar+\lceil4\log_2(N+1)\rceil+\lceil2\log_2(1/\epsilon)\rceil+c_1$, for a sufficiently large universal constant $c_1$, while keeping the target QMA computation and the base protocol's logical security parameter $\secp$ fixed. Within each invocation the formulas remain $\delta=2^{-\widehat\secpar}$ and $T=\lceil2^{\widehat\secpar/2}\rceil$; all estimation and compiler-accuracy parameters are chosen for this accuracy and instance length; the fixed repetition parameters above refer to the base logical parameter $\secp$. The initial eigenspace witnesses satisfy the repairable condition at this accuracy. Each invocation uses fresh setup; reuse inside it is already accounted for in $\Gamma$. The explicit bound $\mu_{\rm qual}=O(2^{-\widehat\secpar/2})$ gives $N\mu_{\rm qual}=O(2^{-\secpar/2})$, uniformly in $|x|$, and ensures $\tfrac32\sqrt{2\delta}\leq2\epsilon$. The conditional residual state after round $i$ is denoted by $\rho_i$. The per-round quality guarantee applies uniformly to every repairable mixed input after each fixed earlier history. Thus, by a union bound, with overwhelming probability every intermediate witness satisfies $\rho_i\in\Repairable_{p-4i\epsilon,\epsilon}(x)$.
    These quality bounds leave at least $4\epsilon$ above $a$, so the value-estimation error is absorbed and each intermediate witness is accepted by $V_\QMA$ with probability at least $a$. Thus, on every history before a repair failure, the next check round accepts with conditional probability at least $c$, and the next honest test accepts perfectly, by the perfect test completeness in \Cref{def:2part} and the verdict-preserving compilers.
    To account for the random check set $B$ and the evolving witness, let $\mathcal F_{i-1}$ denote the full history before round $i$ and set $Z_i=\mathbf{1}_{\{i\in B\}}(b_i-(c+s)/2)$. For analysis, after the first repair or test failure, set all remaining $Z_i=\Delta/4$. Since the branch choice is fair, $\mathbb E[Z_i\mid\mathcal F_{i-1}]\geq\Delta/4$, and the centered increments have absolute value at most $1$. By \Cref{thm:azuma} and a union bound over the repair failures, whose total probability is at most $N\mu_{\rm qual}$,
    \[
        \Pr\left[\sum_{i\in B}b_i < \frac{c+s}{2}|B|\ \text{or a test round rejects}\right]
        \leq \exp\left(-\frac{\Delta^2N}{32}\right)+\negl
        =\negl.
    \]

    \paragraph{$4\epsilon N$ Quality Loss.} Finally, we show that the witness quality decreases by at most $4\epsilon N$ with overwhelming probability. After an early abort, keep the residual witness unchanged for the remaining indices.\footnote{This additional claim does not appear in \Cref{thm:generic-high-completeness} and is not necessary for our main application of witness-preserving $\CVQC$, but is a useful property nonetheless.}
    \begin{align*}
        &\Pr[\rho_N\in\Repairable_{p-4N\epsilon,\epsilon}(x)]\geq \\
        &\Pr[\rho_N\in\Repairable_{p-4N\epsilon,\epsilon}(x)~|~\rho_{N-1}\in\Repairable_{p-4(N-1)\epsilon,\epsilon}(x)]\cdot\\ &\Pr[ \rho_{N-1}\in\Repairable_{p-4(N-1)\epsilon,\epsilon}(x) ]\geq \\
        &(1-\mu_N)\cdot\Pr[ \rho_{N-1}\in\Repairable_{p-4(N-1)\epsilon,\epsilon}(x) ]\geq\\
         &(1-\mu_N)(1-\mu_{N-1})\cdot\Pr[ \rho_{N-2}\in\Repairable_{p-4(N-2)\epsilon,\epsilon}(x) ]\geq\\
        &...\\
&(1-\mu_N)(1-\mu_{N-1})\cdot\ldots\cdot(1-\mu_1)\Pr[\rho_0\in\Repairable_{p,\epsilon}(x)]=\\
&(1-\mu_N)(1-\mu_{N-1})\cdot\ldots\cdot(1-\mu_1)=\\
&1-\negl.
    \end{align*}
    where all the equations except the last two follow from the fact that $(P',V')$ loses at most $4\epsilon$ in quality per round, the second to last equality follows by assumption, and the last equation follows from $\sum_i\mu_i\leq N\mu_{\rm qual}=O(2^{-\secp/2})$ with the invocation accuracy above. Fresh initial estimation makes the expected-time bound uniform even after a quality failure; the tower property therefore bounds total expected runtime by $N$ times the per-invocation bound.
\end{proof}

\subsubsection{Construction}\label{sec:generic-eps-construction}
Fix a $\QMA_{a,b}$ verifier $V_\QMA$. The construction makes use of the following two components:
\begin{itemize}
    \item a $\CVQC$ protocol $(P, V)$ with respect to $V_\QMA$.
    \item and a $\CVQC$ protocol $(P', V')$ which is $(1-2^{-\secpar})$-\emph{gentle} (\Cref{def:gentle}) with respect to $(P, V)$.
\end{itemize}
We show how to construct the latter generically from any non-adaptive $(P, V)$ under the LWE assumption and quantitative parameter choice in \Cref{sec:gentle}.

\begin{construction}[$\epsilon$-Non-Destructive $\CVQC$]\label{constr:eps-nd}
Fix a $\QMA$ language $\lang = (\lang_\yes, \lang_\no)\in\QMA_{a,b}$ with verifier $V_\QMA$. To prove to the verifier that $x\in\calL_\yes$,  the prover  begins with a witness state $\ket{w}\in\Repairable_{p,\epsilon}(x)$, where  $p\geq a+4\epsilon$, and the protocol proceeds as follows.
\begin{enumerate}
    \item \textbf{Prover:}
    Compute $(p^*,\brho)\gets\ValEst_{V_{\QMA,x}}(\ket{w},\epsilon,\delta)$, with $\delta=2^{-\secp}$.

    Retain the classical estimate $p^*$ and the possibly mixed residual witness $\brho$. All subsequent operations are applied to the actual state; the quality guarantee is proved below without replacing it by an ideal good witness. Initialize the base prover's workspace to zero, and extend $\ValEst$ by the identity on that workspace.

    \item \textbf{Both:} Execute $(\rho_d, d) \gets (P'(\brho), V')$, where the prover acts as $P'$ and the verifier acts as $V'$. The verifier outputs $d$.

    \item \textbf{Verifier:} Send the base verifier's coins $r$ used within $(P',V')$ to the prover. The honest repair ignores the compiler's other classical transcript records.

    \item \textbf{Prover:} Repair the witness as follows.
    \begin{enumerate}
        \item Define $\Pi_r = \{\Pi_r^{\mathsf{acc}}, \Pi_r^{\mathsf{rej}}\}$ to be the projective measurement corresponding to a coherent execution of $(P, V)$ and measurement of $V$'s decision bit (see \Cref{def:gentle}).
        \item Define ${\sf M}=\ValEst_{V_{\QMA,x}}(\cdot, \epsilon,\delta)$, where $\delta = 2^{-\secp}$. $\ValEst$ is the almost projective measurement from \Cref{lemma:CMSZ-valest}.

    \item Compute $\bsigma^*\gets \Repair_{{\sf M},\Pi_r}(\rho_d,d,p^*,T)$, where $T=\lceil1/\sqrt\delta\rceil=\lceil2^{\secpar/2}\rceil$, where $p^*$ is the quality estimate of the initial witness, and where $\Repair$ is from \Cref{lemma:CMSZ-repair}.
    \item Estimate $(p^{**}, \bsigma^{**}) \gets \ValEst_{V_{\QMA,x}}(\bsigma^*, \epsilon, \delta)$ and output its witness register along with the new quality estimate $p^{**}$, discarding the base workspace. Each new invocation performs step 1 on its actual input, including after an earlier repair failure.
    \end{enumerate}
\end{enumerate}
\end{construction}

\subsubsection{Analysis}\label{sec:generic-eps-analysis}
\begin{proof}[Proof of \Cref{thm:eps-non-des}]
Write $(P'',V'')$ for the compiled protocol. The verifier decides before revealing its coins or running local repair, so soundness and both test/check acceptance probabilities are those of $(P',V')$.

\paragraph{Comparison with ideal repair.}
Fix $x$ and any initial joint state of the witness and an external register that the protocol leaves untouched. Let $\Hyb_0$ be the real execution, retaining the first value estimate $p^*$, the base verifier's tape $r$, the verdict $d$, and the recovered base registers. In $\Hyb_1$, replace only the gentle interaction by the binary projective measurement $\Pi_r$ of the base protocol. The reference-preserving statement of \Cref{def:gentle}, with the uniform budget in \Cref{par:compiler-accuracy}, gives
\[
 \frac12\|\Hyb_0-\Hyb_1\|_1\leq\Gamma\leq2^{-\secpar}=\delta.
\]
This is a joint bound including $p^*,r,d$. Apply exactly the same controlled repair procedure and final value estimation to both states, also outputting its runtime counter. Contractivity preserves this bound. No compiler transcript other than $r,d$ is used by repair.

\paragraph{Expected runtime.}
The ideal experiment is precisely \Cref{lemma:CMSZ-repair}, with two damaging outcomes and a fresh initial value estimate. Let $C$ count the repair measurement steps in that lemma's convention. Since $T=\lceil\delta^{-1/2}\rceil$,
\[
 \mathbb E[C_{\rm ideal}]
 \leq 3+4T\sqrt\delta\leq7+4\sqrt\delta,
 \qquad 0\leq C\leq2T+1.
\]
The joint comparison including the counter therefore implies
\[
 \mathbb E[C_{\rm real}]
 \leq7+4\sqrt\delta+(2T+1)\Gamma=O(1).
\]
For calls to coherent measurement implementations and their inverses, the hard cap is $4T+3$ and the expectation is at most $2\mathbb E[C_{\rm real}]+1$, also $O(1)$. Each such call has cost polynomial in $|x|,\secpar,1/\epsilon$. The initial and final value estimates add polynomial cost. These bounds hold uniformly for every initial mixed state, not only repairable witnesses. Hence a fresh invocation after any earlier history, including a repair failure, still has expected polynomial runtime. This argument averages over the fresh estimate and verifier tape; it does not assert a constant conditional expectation for every individual estimate or verdict.

\paragraph{Completeness and quality.}
Now assume the stated hypothesis $\rho\in\Repairable_{p,\epsilon}(x)$, with $p\geq a+4\epsilon$. In the ideal repair experiment, \Cref{lemma:CMSZ-repair} bounds the probability of $|p^{**}-p^*|\geq2\epsilon$ by $2(\delta+1/T)+4\sqrt\delta$. The comparison above adds $\Gamma$. Together with the defining bound for the initial estimate, this gives
\begin{align*}
 \Pr[p^{**}<p-3\epsilon]&\leq\alpha,\\
 \alpha&:=\sqrt{2\delta}+\Gamma+2\delta+2/T+4\sqrt\delta
 \leq3\delta+(6+\sqrt2)\sqrt\delta.
\end{align*}
Let $B=\{p^{**}\geq p-3\epsilon\}$. Consider one further value estimate $p^{***}$ only for analysis. Almost-projectivity, applied to the actual state before the final estimate, gives $\Pr[|p^{***}-p^{**}|\geq\epsilon]\leq\delta$. Thus, for sufficiently large $\secpar$ so that $\alpha\leq1/2$,
\[
 \Pr[p^{***}<p-4\epsilon\mid B]
 \leq\frac{\delta}{\Pr[B]}
 \leq\frac{\delta}{1-\alpha}\leq2\delta.
\]
Let $h$ be the full classical history through the final estimate and let $\rho_h$ be the conditional residual witness state. Set $q_h=\Pr[p^{***}<p-4\epsilon\mid h]$. The preceding bound gives $\mathbb E[q_h\mid B]\leq2\delta$, so Markov's inequality yields
\[
 \Pr[q_h>\sqrt{2\delta}\mid B]\leq\sqrt{2\delta}.
\]
For every remaining history in $B$, the defining inequality for $\Repairable_{p-4\epsilon,\epsilon}(x)$ holds for $\rho_h$ (indeed, $p-4\epsilon$ is a stronger cutoff than the required $p-5\epsilon$). Consequently $\Pr[\rho_h\notin\Repairable_{p-4\epsilon,\epsilon}(x)]$ is at most
\[
 \mu_{\rm qual}:=\alpha+\sqrt{2\delta}
 \leq3\delta+(6+2\sqrt2)\sqrt\delta.
\]
Discarding the base workspace does not change any value-estimation probabilities, since $\ValEst$ acts only on the witness.

For completeness, the first estimate obeys $p^*\geq p-\epsilon$ except with probability $\sqrt{2\delta}$. Almost-projectivity of a second estimate gives, conditioned on this event, failure at most $\delta/(1-\sqrt{2\delta})\leq2\delta$. Applying \Cref{lem:mixed-to-pure} shows that the post-estimation witness has a convex decomposition assigning total weight at least $1-2\sqrt{2\delta}$ to pure states in $\Repairable_{p-\epsilon,\epsilon}(x)$. Such witnesses have QMA acceptance at least $p-2\epsilon-\tfrac32\sqrt{2\delta}\geq a$ whenever $\tfrac32\sqrt{2\delta}\leq2\epsilon$. Completeness of $(P',V')$ and linearity therefore give $c-2\sqrt{2\delta}$. This uses the specified repairable-witness class; the expected-runtime argument alone gives no completeness or quality claim for arbitrary witnesses.
\end{proof}

\end{document}